\documentclass[draftclsnofoot, onecolumn]{IEEEtran}
\usepackage[english]{babel}
\usepackage[utf8x]{inputenc}
\usepackage{amssymb,amsfonts,amsthm}
\usepackage{dsfont}
\usepackage{graphicx}
\usepackage{blkarray}
\usepackage{array}
\usepackage{multicol}
\usepackage[colorinlistoftodos]{todonotes}
\usepackage{amsmath}

\newtheorem{mydef}{Definition}
\newtheorem{proposition}{Proposition}
\newtheorem{Lemma}{Lemma}
\newtheorem{corollary}{Corollary}
\newtheorem{assumption}{Assumption}
\newtheorem{remark}{Remark}
\newtheorem{theorem}{Theorem}

\ifCLASSINFOpdf
\else
\fi
\usepackage{amsthm}
\usepackage{cite}
\usepackage{balance}
\usepackage{latexsym}
\usepackage{graphicx}
\usepackage{caption}
\usepackage{soul,color}
\usepackage{amsmath}
\usepackage{algorithm}
\usepackage{subfigure}
\usepackage{authblk}
\usepackage[noend]{algpseudocode}
\usepackage{romanbar}

\begin{document}
\title{On the Global Optimality of Whittle's index policy for minimizing the age of information}
\author[*]{KRIOUILE Saad}
\author[*]{ASSAAD Mohamad}
\author[*] {MAATOUK Ali}
\affil[*] {TCL Chair on 5G, Laboratoire des Signaux et Systemes, CentraleSupelec, Gif-sur-Yvette, France}

\maketitle

\begin{abstract}
This paper examines the average age minimization problem where only a fraction of the network users can transmit simultaneously over unreliable channels. Finding the optimal scheduling scheme, in this case, is known to be challenging. Accordingly, the Whittle's index policy was proposed in the literature as a low-complexity heuristic to the problem. Although simple to implement, characterizing this policy's performance is recognized to be a notoriously tricky task. In the sequel, we provide a new mathematical approach to establish its optimality in the many-users regime for specific network settings. Our novel approach is based on intricate techniques, and unlike previous works in the literature, it is free of any mathematical assumptions. These findings showcase that the Whittle's index policy has analytically provable asymptotic optimality for the AoI minimization problem. Finally, we lay out numerical results that corroborate our theoretical findings and demonstrate the policy's notable performance in the many-users regime.  
\end{abstract}

\section{Introduction}
Technological advances in wireless communications and the cheap cost of hardware have led to the emergence of real-time monitoring services. In these systems, an entity is interested in knowing the status of one or multiple processes observed by a remote source. Accordingly, the source sends packets to the monitor to provide information about the process/processes of interest. The main goal in these applications is to keep the monitor up to date. In fact, in such applications, information has the highest value when it is fresh since the outcome of the monitor's tasks is better when it is based on new rather than outdated data. To quantify this notion of freshness, the Age of Information (\textbf{AoI}) was introduced in \cite{kaul2012real}. Ever since, 
the AoI has become a hot research topic, and a considerable number of research works have been published on the subject \cite{buyukatestimely,farazi2018age,maatouk2018age,zou2019waiting,sun2017update,talak2017minimizing,bedewy2017age,kosta2019age}. 

Among the most fundamental issues that the research community aimed to address is age-based resource allocation. In most real-time applications, numerous sources share the same transmission channel where the available resources are scarce. The scarcity can be a consequence of battery considerations for the devices involved or physical interference that may limit the number of simultaneous transmissions. Consequently, a smart resource allocation scheme has to be adopted to minimize the AoI and attain the desired timeliness objective. In \cite{sun2018age}, the authors proposed both age-optimal and near age-optimal scheduling policies for the single and multi-server cases, respectively. In particular, they have shown that a greedy policy is age-optimal under certain assumptions in the single exponential server case. In \cite{ceran2019average}, the authors examined a single-source scenario where the source’s update rate cannot exceed a predefined limit due to battery considerations. In this case, they were able to propose an age-optimal scheduling policy when the channel exhibit possible decoding errors. Age-optimal policies were also proposed in various network settings such as distributed scheduling and random access environments \cite{maatouk2020age,jiang2018timely,talak2018distributed}.

Among the scheduling problems investigated in the literature, we cite the following: consider $N$ users communicating with a central entity over unreliable channels where, at most, $M<N$ users can transmit simultaneously. What is the age-optimal strategy in this case? The wide range of applications that this problem encompasses let it emerge as a fundamental one that needs to be investigated. Unfortunately, this problem belongs to the family of Restless Multi-Armed Bandit (\textbf{RMAB}) problems,
which are generally difficult to solve optimally. To address this difficulty, the authors in \cite{kadota2018scheduling} have examined this problem and proved that a greedy algorithm is optimal when users have identical channel statistics. For the asymmetric case, the authors proposed a sub-optimal policy, known as the Whittle's index policy. The Whittle's index policy has been embraced by many works in various frameworks \cite{ansell2003whittle,kriouile2018asymptotically,kriouile2019whittle,larranaga2017asymptotically,liu2010indexability,maatouk2020optimality,ouyang2015downlink,papadaki2002exploiting,weber1990index,whittle1988restless} as it is recognized for its low complexity and its notable performance. For example, in \cite{kriouile2018asymptotically}, the Whittle's index policy was adopted to minimize the average delay of queues. In another line of work, the authors in \cite{ouyang2015downlink} employed a Whittle's index-based policy to maximize the average throughput over Markovian channels. Although it is simple to implement, the main challenge that arises when adopting this policy is characterizing its performance since its analysis is known to be notoriously difficult. To attend to this difficulty, the authors in \cite{weber1990index} provided a sufficient condition, dubbed as Weber’s condition, for the Whittle’s index policy's asymptotic optimality in the many-users regime. However, this condition requires ruling out the existence of both closed orbits and chaotic behavior of a high-dimensional non-linear differential equation, which is extremely difficult to verify even numerically. To further facilitate the analysis of the policy, the works in \cite{ouyang2015downlink,kriouile2018asymptotically} have provided an approach based on a fluid limit model for the delay minimization and throughput maximization frameworks. By leveraging this model, they proved the asymptotic optimality of the Whittle's index policy in these frameworks under a recurrence assumption that is easier than Weber's condition but still requires numerical verification. Following the same footsteps, the present authors adopted the fluid limit model and provided proof of the asymptotic optimality of the Whittle's index policy in the AoI framework under similar assumptions \cite{maatouk2020optimality}. This raises the following important question: can we prove the Whittle's index policy's asymptotic age-optimality in specific network settings without recoursing to \textbf{any} assumptions? Answering this question is extremely difficult and has yet to be answered even for the standard delay and throughput metrics. In this paper, we examine this question in the AoI framework, and we provide rigorous theoretical results that showcase the validity of the Whittle's index asymptotic optimality in \emph{certain} network settings without imposing any assumptions. Note that the importance of the asymptotic many-users regime stems from the astronomical growth in the number of interconnected devices. For example, machine-type communications and the IoT in 5G networks require supporting tens of thousands of connected devices in a \emph{single} cell. To that end, we summarize in the following the structure of the paper along with its key contributions:
\begin{itemize}
\item We start by formulating the problem of minimizing the average age of a network where $M$ out of $N$ users can communicate simultaneously with the central entity. As previously explained, this problem belongs to the class of RMAB problems, which are known to be notoriously difficult to solve. Accordingly, the Whittle's index policy has been proposed in previous works as a low-complexity solution, which is the main focus of our work. To establish the Whittle's index policy, the following steps have to be taken:
\begin{enumerate}
\item Provide a relaxed version of the original problem and tackle it through a Lagrangian approach.
\item Prove the indexability property of the relaxed problem and derive the Whittle’s index expressions.
\end{enumerate}
These steps have been carried out in previous works by the authors in \cite{kadota2018scheduling}, and their main results are reported in our paper for completeness. 
\item Next, we present a fluid limit model that approximates the Whittle's index policy behavior. In the many-users regime, we prove that the fluid limit can be made arbitrarily close to the actual network's evolution. Therefore, we mainly focus on the evolution of the fluid limit vector in our optimality analysis. The method previously carried out in the literature to establish the Whittle's index policy's asymptotic optimality follows a spectral analysis approach \cite{ouyang2015downlink}. However, this approach is highly contingent on the initial state of the system. Accordingly, to extend their results to any random initial state, the authors imposed a restrictive assumption, which can only be verified numerically. In our paper, we take a different approach to analyze the fluid model. Specifically, we propose a novel method based on intricate techniques (e.g., Cauchy criterion) to prove the fluid model's convergence to a fixed point. We stress that this step's technical details are intricate and constitute our paper's main technical contribution. Note that, even for the standard delay and throughput metrics, such proof was not provided in the literature, which further highlights our approach's novelty. Afterwards, we establish the global optimality of Whittle's index policy leveraging the fact that the aforementioned fixed point is nothing but the optimal system's operating point in the many-users regime. Finally, we provide numerical results that corroborate the theoretical results and highlight the Whittle's index policy's notable performance in the many-users regime. 
\end{itemize}
The rest of the paper is organized as follows: Section \ref{sec:systemModel} is devoted to the system model and the problem formulation. Section \ref{sec:Relaxed_Dual_Prob} is dedicated to the establishment of the Whittle's index policy. In Section \ref{sec:Whittle_Index_optimal}, we provide our main results where we prove the asymptotic optimality of the Whittle's index policy. Numerical results that corroborate our theoretical findings are given in Section \ref{sec:Num_Reslts} while Section \ref{sec:Conclusion} concludes the paper.

\section{System Model and Problem Formulation}\label{sec:systemModel}
\subsection{System Model}
We consider a time-slotted system with one base station, $M$ uncorrelated channels, and $N$ users ($N>M$). Time is considered to be normalized to the slot duration (i.e, $t=1,2,\ldots$). We suppose that any of the $M$ channels can be allocated to at most one user. Hence, at most $M$ users will be able to transmit in each time slot $t$. If a user is scheduled at time $t$, it generates a fresh new packet and sends it to the base station. This packet is successfully decoded by the base station at time $t+1$ with a certain success probability. We consider that if a decoding error takes place, the packet is discarded (i.e., users are not
equipped with buffers). In practice, users may share similar
channel conditions. Accordingly, we suppose that the users can be partitioned into $K=2$ different classes such that users within the same class share the same decoding success probability. In other words, each user $i$ belonging to class $k\in\{1,2\}$ has a decoding success probability $p_k$, which is assumed to be known by the scheduler. We let $\gamma_k$ be the proportion of users belonging to class $k$. To that end, the following always holds: $\gamma_1+\gamma_2=1$.

A scheduling policy $\pi$ is defined as a sequence of actions $\pi=(\boldsymbol{a}^{\pi}(0),\boldsymbol{a}^{\pi}(1),\ldots)$ where \\ $\boldsymbol{a}^{\pi}(t)=(a_1^{1,\pi}(t),a_2^{1,\pi}(t),\ldots,a_{\gamma_1N}^{1,\pi}(t),a_1^{2,\pi}(t),a_2^{2,\pi}(t),\ldots,a_{\gamma_2N}^{2,\pi}(t))$ is a binary vector such that $a_i^{k,\pi}(t)=1$ if user $i$ of class $k$ is scheduled at time $t$.
We also let the binary random variable $c_i^k(t)$ denote the channel state of user $i$ of class $k$ such that $c_i^k(t)=1$ if no decoding error takes place. As per our system model, we always have $\Pr(c_i^k(t)=1)=p_k$ and $\Pr(c_i^k(t)=0)=1-p_k$ for any user $i$ of class $k$. We let $B_i^{k,\pi}(t)$ denote the time-stamp of the freshest packet delivered by user $i$ of class $k$ to the base station at time $t$ under the scheduling policy $\pi$. The age of information, or simply the age, of user $i$ of class $k$ is defined as \cite{kaul2012real}:
\begin{equation}
s_i^k(t)=t-B_i^k(t)
\end{equation}
By taking into account the variables defined, the age of this user under policy $\pi$ evolves as follows:
\begin{align}
 s_i^{k,\pi}(t+1)=\left\{
    \begin{array}{ll}
        1 & \text{if} \ a_i^{k,\pi}(t)=1, c_i^k(t)=1  \\
        s_i^{k,\pi}(t)+1 & \text{if}  \ a_i^{k,\pi}(t)=1, c_i^k(t)=0\\
        s_i^{k,\pi}(t)+1 & \text{if} \ a_i^{k,\pi}(t)=0,
    \end{array}
\right.
\end{align}
We let $\boldsymbol{s}^{\pi}(t)$ denote the vector of all users' age $\boldsymbol{s}^{\pi}(t)=(s_1^{1,\pi}(t),\cdots,s_{\gamma_1 N}^{1,\pi}(t),s_1^{2,\pi}(t),\cdots,s_{\gamma_2 N}^{2,\pi}(t))$ under policy $\pi$. With all these notations in mind, we can formulate the optimization problem that we focus on in our paper.
\subsection{Problem Formulation}
In this paper, we are interested in minimizing the total expected average age of information of the network under the constraint on the number of users scheduled at each time slot $t$. The latter must be less than the total number of channels $\alpha N$ where $\alpha$ is equal to $\frac{M}{N}$. We let $\Pi$ denote the set of all causal scheduling policies in which the scheduling decisions are made based on the
history and current states of the system. To that end, and given an initial system state $\boldsymbol{s}(0)=(s_1^{1}(0),\cdots,s_{\gamma_1 N}^{1}(0),\cdots,s_1^{K}(0),\cdots,s_{\gamma_K N}^{K}(0))$, our problem can be formulated as follows:
\begin{align}
& \underset{\pi \in \Pi}{\text{min}} \limsup\limits_{T\rightarrow\infty}\frac{1}{T} \mathbb{E}^{\pi}\left[ \sum_{t=0}^{T-1}  \sum_{k=1}^{K} \sum_{i=1}^{{\gamma_k}N} s_i^{k,\pi}(t) \mid \boldsymbol{s}(0)\right]
\nonumber \\
& \text{s.t.}  \sum_{k=1}^{K} \sum_{i=1}^{{\gamma_k}N} a_i^{k,\pi}(t)\leq {\alpha}N, \quad t=0,1,2,\ldots
\label{eq:original_problem}
\end{align}   
This problems belongs to the family of RMAB problems,
which are generally difficult to solve optimally (see Papadimitriou et al. \cite{papadimitriou1999complexity}). For this reason, one should aim to develop a well-performing sub-optimal policy. As it has been mentioned, the low-complexity scheduling policy that we are interested in throughout this paper is the Whittle's index policy. To establish this policy and derive the Whittle's indices expressions, one has to follow the steps below:
\begin{enumerate}
\item Provide a relaxed version of the original problem and tackle it through a Lagrangian approach.
\item Prove the indexability property of the problem and derive the Whittle’s index expressions.
\end{enumerate}
As previously mentioned, these steps have been carried out in previous works by the authors in \cite{kadota2018scheduling}. For completeness, and as we will use these steps later in our optimality analysis, we report them along with the main results of \cite{kadota2018scheduling} in the following section. 

\section{Relaxed Problem and Whittle's Index Policy}\label{sec:Relaxed_Dual_Prob}
\subsection{Relaxed Problem}
The first step toward establishing the Whittle's index policy consists of relaxing the constraint on the number of scheduled users of the problem in (\ref{eq:original_problem}). Specifically, instead of having the constraint satisfied at each time slot, we consider that it has to be satisfied on average. Therefore, the relaxed problem can be formulated as follows:
\begin{align}
& \underset{\pi \in \Pi}{\text{min}} \limsup\limits_{T\rightarrow\infty}\frac{1}{T} \mathbb{E}^{\pi}\left[ \sum_{t=0}^{T-1}  \sum_{k=1}^{K} \sum_{i=1}^{{\gamma_k}N} s_i^{k,\pi}(t) \mid \boldsymbol{s}(0)\right]\nonumber \\
& \text{s.t.}   \limsup\limits_{T\rightarrow\infty}\frac{1}{T} \mathbb{E}^{\pi}\left[ \sum_{t=0}^{T-1}  \sum_{k=1}^{K} \sum_{i=1}^{{\gamma_k}N} a_i^{k,\pi}(t)\right] \leq {\alpha}N
\label{eq:constraint_relaxed}
\end{align}   
To study this problem, one has to introduce a Lagrangian approach to transform the problem into an unconstrained one as will be detailed in the sequel.
\subsection{Dual Problem}
To circumvent the difficulty of studying the constrained problem in (\ref{eq:constraint_relaxed}), a Lagrangian approach has to be adopted. In particular, let us denote by $\lambda\geq0$ the Lagrangian parameter. For a fixed $\lambda$, the Lagrangian function of the relaxed problem is:
\begin{align}
F(\lambda, \pi)= \limsup\limits_{T\rightarrow\infty}\frac{1}{T} \mathbb{E}^{\pi}\left[ \sum_{t=0}^{T-1}  \sum_{k=1}^{K} \sum_{i=1}^{{\gamma_k}N} s_i^{k,\pi}(t)+\lambda (a_i^{k,\pi}(t)-\alpha) \mid \boldsymbol{s}(0)\right]
\end{align}
Based on the dual approach, the next step consists of finding the policy $\pi$ that minimizes $F(\lambda, \pi)$. Note that the term $\frac{1}{T}\sum_{t=0}^{T-1}  \sum_{k=1}^{K} \sum_{i=1}^{{\gamma_k}N}\lambda \alpha$, which is equal to $N \lambda \alpha$, doesn't depend on $\pi$. Therefore, the policy that minimizes the above function $F(\lambda, \pi)$ also minimizes the following function:
\begin{align}
f(\lambda, \pi)= \limsup\limits_{T\rightarrow\infty}\frac{1}{T} \mathbb{E}^{\pi}\left[ \sum_{t=0}^{T-1}  \sum_{k=1}^{K} \sum_{i=1}^{{\gamma_k}N} s_i^{k,\pi}(t)+\lambda a_i^{k,\pi}(t) \mid \boldsymbol{s}(0)\right]
\label{lagrangianfunction}
\end{align}
Then, we can formulate the dual problem as follows: 
\begin{equation}\label{eq:dual_relaxed_problem}
\underset{\pi \in \Pi}{\text{min}} f(\lambda,\pi)
\end{equation}

\subsection{Structural Results}
To solve the problem in (\ref{eq:dual_relaxed_problem}), it can be shown that this $N$-dimensional problem can be decomposed into $N$ one-dimensional problems
that can be solved independently \cite{kadota2018scheduling}. Therefore, we can drop the $i$ and $k$ indices from (\ref{lagrangianfunction}) and simply investigate the following one-dimensional problem:
\begin{equation}\label{eq:individual_dual_problem}
\underset{\pi \in \Pi}{\text{min}} \limsup\limits_{T\rightarrow\infty}\frac{1}{T} \mathbb{E}^{\pi}\left[ \sum_{t=0}^{T-1} s^{\pi}(t)+\lambda a^{\pi}(t) \mid s(0)\right]
\end{equation}
It turns out that the above one dimensional problem can be cast into an infinite horizon average cost Markov Decision Process (\textbf{MDP}) that is defined as follows:
\begin{itemize}
\item \textbf{States}: The state of the MDP at time $t$ is the age of the user $s(t)$ that can take any integer value strictly higher than $0$. Therefore, the considered state space is countable and infinite.
\item \textbf{Actions}: The action at time $t$, denoted by $a(t)$, indicates if a transmission is attempted (value $1$) or the user
remains idle (value $0$).
\item \textbf{Transitions probabilities}: The transitions probabilities
between the different states have been previously detailed
in Section \ref{sec:systemModel}.
\item \textbf{Cost}: The cost function at time $t$ is designated by $C(s(t),a(t))=s(t)+\lambda a(t)$.
\end{itemize}
To solve this MDP, the authors in \cite{kadota2018scheduling} have leveraged the Bellman equation and studied the characteristics of the value function involved. Based on the particularity of the value function, the following result was found: 
\begin{proposition}\label{prop:threshold_policy}
The optimal policy that solves problem \eqref{eq:individual_dual_problem} is of a threshold nature.
\end{proposition}
\begin{proof}
See \cite[Proposition~14]{kadota2018scheduling}.  
\end{proof}
The above results tell us that there exists an integer $l_k \in \mathbb{N}^*$ such that by only letting users of class $k$ with an age larger or equal to $l_k$ to transmit, we attain the optimal operating point of (\ref{eq:individual_dual_problem}). These results are pivotal to proceed with establishing the Whittle's index policy.

\subsection{Indexability and Whittle's Index Expressions}
To proceed toward our goal, one has to analyze the behavior of the MDP when a threshold policy is adopted. To that end, we note that for any fixed threshold $n$, the MDP can be modeled through a Discrete Time Markov Chain (\textbf{DTMC})
where:
\begin{itemize}
\item The state is the age $s(t)$.
\item For any state $s(t)<n$, the user is idle. On the other hand, when $s(t)\geq n$, the user is scheduled. 
\end{itemize}
The DTMC is reported in Fig. \ref{dtmc}. 
\begin{figure}[!ht]
\center
\includegraphics[scale=0.65]{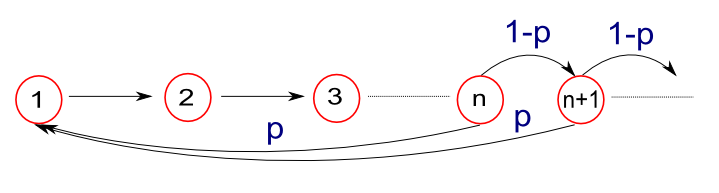}
\caption{The states transition when a threshold policy is adopted}
\label{dtmc}
\end{figure}
To be able to prove the indexability property and find the Whittle's index expression, one has to find the average objective function in (\ref{eq:individual_dual_problem}) when a threshold policy is adopted. To that end, we provide the following propositions.
\begin{proposition}\label{prop:stat_distribution}
For a fixed threshold $n$, the stationary distribution $u^n$ of the DMTC when the decoding success probability is equal to $p$ is:
\begin{equation}
 u^n(i)=\left\{
    \begin{array}{ll}
        \frac{p}{np+1-p} & \text{if} \ 1 \leq i \leq n  \\
        (1-p)^{i-n} \frac{p}{np+1-p} & \text{if} \ i \geq n
    \end{array}
\right.
\end{equation}
\end{proposition}
\begin{proof}
The results can be easily obtained by solving the full balance equations.
\end{proof}
The next step consists of calculating the average objective function in (\ref{eq:individual_dual_problem}) when a threshold policy is employed. 

\begin{proposition}\label{prop:stat_cost}
For a fixed threshold $n$, the average cost of the threshold policy of the problem (\ref{eq:individual_dual_problem}) is:
\begin{equation}
\overline{C}(n,\lambda)=\frac{[(n-1)^2+(n-1)]p^2+2p(n-1)+2}{2p((n-1)p+1)}+\frac{\lambda}{np+1-p}
\label{costgamma}
\end{equation}
\end{proposition}
\begin{proof}
The results can be concluded by leveraging the stationary distribution expressions and the fact that $\overline{C}(n,\lambda)=\sum_{i=1}^{+\infty} iu^n(i)+\lambda \sum_{i=n}^{+\infty}u^n(i)$.
\end{proof}
Using the stationary distribution, and the average cost, one can then prove the indexability property of the problem, which ensures
the existence of the Whittle’s indices. Before providing these results, we first lay out the definition of the aforementioned property.
\begin{mydef}[Indexability]
For a fixed $\lambda$, consider the vector $\boldsymbol{l}(W)=(l_1(\lambda),\ldots,l_K(\lambda))$ where $l_k(\lambda)$ is the optimal threshold for the problem in (\ref{eq:individual_dual_problem}) for each user of class $k$. We define $D^k(\lambda)=\{s\in\mathbb{N}^*:s<l_k(\lambda)\}$ as the set of states for which the optimal action is to not schedule the users belonging to class $k$. The one-dimensional problem associated with these users is said to be indexable if $D^k(\lambda)$ is increasing in $\lambda$. More specifically, the following should hold:
\begin{equation}
\lambda'\leq \lambda \Rightarrow D^k(\lambda) \subseteq D^k(\lambda)
\end{equation}
\label{indexabilitydefinition}
\end{mydef}
The indexability property for the problem in (\ref{eq:individual_dual_problem})
was established by the authors in \cite{kadota2018scheduling}. With the Whittle's indices ensured to exist, one can then leverage the stationary distribution and the average cost reported in Proposition \ref{prop:stat_distribution} and \ref{prop:stat_cost} to derive the Whittle's index expressions as previously done in \cite{kadota2018scheduling} and \cite{maatouk2020optimality}. 
\begin{proposition}\label{prop:whittle_index}
For any given class $k$, the Whittle's index expression of state $i$ is:
\begin{equation}
W^k(i)=\frac{(i-1)p_k i}{2}+i
\end{equation}
\end{proposition}
\begin{proof}
See \cite[pp.~10]{kadota2018scheduling}. 
\end{proof}
With the Whittle's index expression derived, we can now establish the Whittle's index scheduling policy. This can be summarized in the following algorithm description.
\begin{algorithm}
\caption{Whittle's index scheduling policy}\label{euclid}
\begin{algorithmic}[1]
\State At each time slot $t$, calculate the Whittle's index of all users in the network using (\ref{prop:whittle_index}).
\State Schedule the $M$ users having the highest Whittle's index values at time $t$, with ties broken arbitrarily.
\end{algorithmic}
\end{algorithm}\\
Although the above scheduling policy is easy to implement,
it remains sub-optimal. Accordingly, characterizing its performance
compared to the optimal policy is important. Equipped with the above results and notations, we can now tackle the main issue that we aim to address in our paper: the asymptotic optimality of this policy. 

\section{Asymptotic Optimality of the Whittle's Index Policy}
\subsection{Optimal Solution of the Relaxed Problem}\label{sec:opt_sol_relax_prob}
To be able to prove the asymptotic optimality of the Whittle's index policy, one has to compare its performance to the optimal policy that solves (\ref{eq:original_problem}). However, as previously explained, the optimal policy of (\ref{eq:original_problem}) is not known. To circumvent this, and to have a benchmark performance to compare to, we note that the following always holds:
\begin{equation}
\frac{C^{RP,N}}{N}\leq \frac{C^{OP,N}}{N}\leq \frac{C^{WIP,N}}{N}
\end{equation}
where $\frac{C^{WIP,N}}{N}$ is the average age per-user under the Whittle's index policy, $\frac{C^{OP,N}}{N}$ is the optimal expected average age per-user of the original problem \eqref{eq:original_problem}, and $\frac{C^{RP,N}}{N}$ is the optimal average age per-user of the relaxed problem (\ref{eq:constraint_relaxed}). Thus, in order to show the asymptotic optimality, it is sufficient to prove that for a large number of users $N$, $\frac{C^{WIP,N}}{N}$ converges to $\frac{C^{RP,N}}{N}$. To that end, the next task is to find an expression of $C^{RP}=\frac{C^{RP,N}}{N}$. For this purpose, we provide the following proposition.

\begin{proposition}\label{prop:optmial_and_cost_relaxed_problem}
The optimal solution of the relaxed problem is of type threshold for each class. More precisely, it is a linear combination between two threshold vectors $(l_1^1,\cdots,l_K^1)$ and $(l_1^2,\cdots,l_K^2)$ such that:
\begin{itemize}
\item There exists a unique real value $W^*\in\mathbb{R}$, a class $m$ and state $p$ such that $W^*=W^m(p)$.
\item The expressions of $l_k^1$ and $l_k^2$ are as follows:
\begin{equation*}
l_k^1=\underset{i\in\mathbb{N}^*}{\text{argmax}} \{W^k(i):W^k(i)\leq W^*\}+1 \quad\forall k\in\{1,\ldots,K\}
\end{equation*}
\begin{equation}
l_k^2=\underset{i\in\mathbb{N}^*}{\text{argmax}} \{W^k(i):W^k(i)<W^*\}+1 \quad\forall k\in\{1,\ldots,K\}
\label{twoexpressions}
\end{equation}
\item There exists a unique $0 < \theta \leq 1$ that satisfies $\theta \sum_{k=1}^{K} \gamma_k \sum_{i=l_k^1}^{+\infty} u_k^{l_k^1}(i)+ (1-\theta)\sum_{k=1}^{K} \gamma_k \sum_{i=l_k^2}^{+\infty} u_k^{l_k^2}(i)=\alpha$, where $u_k^n$ is the stationary distribution of the age given a threshold $n$ for class $k$.  
\end{itemize}
\end{proposition}
\begin{proof}
See \cite[Proposition~5]{maatouk2020optimality}.
\end{proof}
Thanks to this proposition, we can conclude that the optimal per-user cost of the relaxed problem has the following expression:
\begin{equation}
C^{RP}=\sum_{k=1}^{K} \gamma_k \sum_{i=1}^{+\infty} [\theta u_k^{l_k^1}(i)+(1-\theta) u_k^{l_k^2}(i)]i 
\label{costrelaxedproblemmm}
\end{equation}
By leveraging these results, we can proceed with characterizing the performance of the Whittle's index policy.
\color{black}
\subsection{Global Optimality of the Whittle's index policy}\label{sec:Whittle_Index_optimal}
This section constitutes the main contribution of the paper where we show the asymptotic optimality of the Whittle's index policy. The idea is to show that the performance of this policy converges to $C^{RP}$ when $N$ is large and the ratio $\alpha=\frac{M}{N}$ is kept constant.

We let $Z_i^{k,N}(t)$ denote the proportion of users belonging to class $k$ in state $i$ at time $t$. In other words, it denotes the ratio of the number of users in class $k$ having an age equal to $i$ to the total number of users $N$. We have that $\boldsymbol{Z}^N(t)=(Z^{1,N}(t),.....,Z^{K,N}(t))$ with $Z^{k,N}(t)=(Z_1^{k,N}(t),......,Z^{k,N}_{m^k(t)}(t))$, where $m^k(t)$ is the highest state at time $t$ in class $k$ and $\sum_{i=0}^{m^k(t)} Z_i^{k,N}(t)=\gamma_k$ for each class $k$. We also denote by $\boldsymbol z^*$ the proportion corresponding to the optimal policy of the relaxed problem. Thus, the elements of the vector $\boldsymbol z^*$ are exactly the set $\{\gamma_k(\theta u_k^{l_k^1}(i)+(1-\theta) u_k^{l_k^2}(i))\}_{1\leq k \leq K \atop 1 \leq i}$ where $i$ and $k$ refer to the user $i$ and class $k$ respectively. This can be easily concluded from the results previously laid out in eq. (\ref{costrelaxedproblemmm}).

In the sequel, we will establish the global optimality for two different classes of users where $p_1$ and $p_2$ are the successful transmission probabilities of class $1$ and $2$ respectively ($p_1 > p_2$).
In order to prove that, we show that when the Whittle's index policy is adopted, $\boldsymbol Z^N(t)$ converges in probability to $\boldsymbol z^*$ when $N$ and $t$ are very large. To that extent, we follow the steps below:
\begin{itemize}
\item We show that the fluid approximation of $\boldsymbol Z^N(t)$, denoted by $\boldsymbol z(t)$, converges to $\boldsymbol z^*$. Such a convergence has been proven in previous works under restrictive mathematical assumptions that can only be verified numerically \cite{ouyang2015downlink}. We escape these assumptions as we will detail in the following. 
\item Since the relation between $\boldsymbol z(t+1)$ and $\boldsymbol z(t)$ is not linear, our approach to establish the convergence of $\boldsymbol z(t)$ involves two terms: $\alpha_1(t)$ and $\alpha_2(t)$.  These two proportions are nothing but the scheduled proportion at time $t$ of class $1$ and $2$, respectively. Note that we always have $\alpha_1(t)+\alpha_2(t)=\alpha$.  Based on Lemma \ref{lem:z_evolution}, we show that for a large enough time $t$, there exists $T_t$ such that we can find a partial relation between each element of the vector $\boldsymbol z(t+T_t)$ and terms of the sequence $\{ \alpha_k(t')\}_{k=1,2 \atop t' \leq t+T_t}$. More precisely, we prove that for $T_t$, we can express each proportion that is not scheduled at time $t+T_t$ in function of one term of $\{\alpha_k(t')\}_{k=1,2 \atop t' \leq t+T_t}$. This allows us to obtain $1-\alpha$ as a linear combination between the terms of $\{\alpha_k(t')\}_{k=1,2 \atop t' \leq t+T_t}$ at time $t+T_t$. 
\item Subsequently, we introduce in Definition \ref{def:time_max}, $T_{\max}$ that satisfies these two following properties proven in Propositions \ref{prop:interval_alternation_under_assump} and \ref{prop:alternation_Whittle_index} using Lemma \ref{lem:alpha_greater_stric_0}: the Whittle's index alternates between the two classes from state $1$ to $T_{\max}$ under a given assumption on $\alpha$; the instantaneous thresholds $l_1(.)$ and $l_2(.)$ are bounded by $T_{\max}$ at time $t+T_t$. 
\item Based on that, we derive the relation between the instantaneous thresholds at time $t+T_t$ in Proposition \ref{prop:unique_threshold_and_expression}. Taking as initial time $t+T_t=T_0$, we show by induction in Proposition \ref{prop:unique_expression_for_all_T} that, for all $T \geq T_0$, the instantaneous thresholds are less than $T_{\max}$ and that all none scheduled proportions can be expressed in function of terms of the sequence $\{ \alpha_k(t')\}_{k=1,2 \atop t' \leq T}$. Next, we define for each class $k$ a vector $A_k(T)$ composed by $\alpha_k(T)$ (the scheduled proportion at time $T$) plus the finite subset of the sequence $\{ \alpha_k(t')\}_{ t' \leq T}$ such that for all proportion of users in class $k$ at a given state at time $T$ that is not scheduled can be expressed by one element belonging to this subset. After that, we provide the relation between the elements of the vectors $A_k(T)$ and $A_k(T+1)$ in Propositions \ref{prop:elements_inclusion} and \ref{prop:inequality_satisfied_by_first_element}
\item As was mentioned in Introduction, our proof is based on Cauchy criterion which states that in the real number space $\mathbb{R}$, a given sequence $h(T)$ is convergent if and only if its terms become closer together as $T$ increases. To that extent, we show that the elements of the vector $A_k(\cdot)$ which are nothing but the terms of the sequence $\alpha_k(\cdot)$ are getting closer when $T$ increases. For that purpose, we prove that the highest and the smallest element of $A_k(T)$ converge to the same limit when $T$ grows. 
To that end, we start by establishing the convergence of the highest and the smallest element of $A_k(T)$ in Theorem \ref{theor:convergence}. Then, we demonstrate by contradiction that the highest and the smallest element of $A_k(T)$ must converge to the same limit in Proposition \ref{prop:max_alpha_less_l_2}. This last result implies that $\alpha_k(t)$ converges when $t$ scales. In light of that fact, we prove that $\boldsymbol z(t)$ converges to $\boldsymbol z^*$ in Proposition \ref{prop:z_convergence}. Finally, using Kurth theorem, we show in Proposition \ref{prop:kurth_theorem_age} that $\boldsymbol Z^N(t)$ converges to $\boldsymbol z^*$ in probability. And finally we establish in Proposition \ref{prop:optimality_whittle_index} the convergence of $\frac{C^{WIP,N}}{N}$ to $C^{RP}$.   
\end{itemize}

With the steps of our approach clarified, we can proceed with introducing the fluid limit approximation. \color{black}The fluid limit technique consists of analyzing the evolution of the expectation of $\boldsymbol Z^N(t)$ under the Whittle's index policy. For that, we define the vector $\boldsymbol z(t)$ as follows:
\begin{equation}\label{eq:fluid_approxiamtion}
\boldsymbol z(t+1)-\boldsymbol z(t)|_{\boldsymbol z(t)=\boldsymbol z}=\mathbb{E}\left[\boldsymbol Z^N(t+1)-\boldsymbol Z^N(t)|\boldsymbol Z^N(t)=\boldsymbol z\right]
\end{equation}

This above equation reveals to us that we have a sequence $\boldsymbol{z}(t)$ defined by recurrence for a fixed initial state $\boldsymbol{z}(0)$ that we should study its behavior when $t$ is very large. Hence, we end up with a function $\boldsymbol{z}(t)$ that depends on two variables, $t$ and the initial value $\boldsymbol{z}(0)$.    
To that extent, our aim is to prove that $\boldsymbol{z}(t)$ converges to $\boldsymbol{z}^*$ regardless of the initial state $\boldsymbol{z}(0)$.
We let $\boldsymbol{z}(t)=(\boldsymbol{z}^{1}(t),.....,\boldsymbol{z}^{K}(t))$ with $\boldsymbol{z}^{k}(t)=(z_1^{k}(t),......,z^{k}_{m^k(t)}(t))$ where $z_i^k(t)$ is the expected proportion of users at state $i$ in class $k$ at time $t$ with respect to the equation \eqref{eq:fluid_approxiamtion}. Accordingly we have that $\sum_{i=0}^{m^k(t)} z_i^{k}(t)=\gamma_k$ for each class $k$.\\

One can notice that $\boldsymbol{z}^*$ is a particular vector with respect to the equation \eqref{eq:fluid_approxiamtion}.

\begin{proposition}\label{prop:fixed_point}
$\boldsymbol{z}^*$ is the \textbf{unique} fixed point of the fluid approximation equation. In other words, $\boldsymbol{z}(t)=\boldsymbol{z}(t+1)$, if and only if $\boldsymbol{z}(t)=\boldsymbol{z}^*$.
\end{proposition}
\begin{proof}
The proof follows the same methodology of the paper \cite[Lemma~9]{ouyang2015downlink} 
\end{proof}

According to this proposition, it is sufficient to show that $\boldsymbol{z}(t)$ converges starting from any initial state $\boldsymbol{z}(0)$, as the only eventual finite limit of $\boldsymbol{z}(t)$ when $t$ tends to $+\infty$ is the fixed point of the equation \eqref{eq:fluid_approxiamtion}, $\boldsymbol{z}^*$.\\

\begin{remark}
We highly emphasize that the proportion $\alpha_k(t)$ and $1-\alpha$ refer to the scheduled users'  proportion at time $t$ in class $k$ and the non scheduled users'  proportion either for class $1$ or $2$ respectively. Meanwhile, for any other proportion $A$, it refers only to the number of users in this proportion over the total users' number of the system whatever the different states of users that contains. Having said that, $A=B$ means that they are equal in terms of proportion, while they can contain users in different states.
\end{remark}

In the following, we prove that the fluid approximation vector of $\boldsymbol{Z}^N(t)$, $\boldsymbol{z}(t)$ under the Whittle Index Policy converges starting from any initial state. We prove this result for 2 different classes of users where $p_1$ and $p_2$ are the successful transmission probabilities of the class 1 and 2 respectively ($p_1 > p_2$), given a sufficient condition on $\alpha$.
Throughout this section, we denote by $w_1(n)$ and $w_2(n)$, the Whittle's index, whose expression is given in Proposition \ref{prop:whittle_index}, of state $n$ in class 1 and class 2 respectively.
We need to prove that $z_i^k(t)$ converges for each state $i$ in class $k$.\\

Now, focusing on the Whittle index policy, we can see it as an instantaneous threshold policy for each class, where the thresholds vary over time $t$. Moreover, under the Whittle index policy, the proportion of users that are scheduled at each time slot $t$ is fixed and equals to $\alpha$ since the number of scheduled users at each time slot $t$ is $\alpha N$. This proportion $\alpha$ contains the users with the highest Whittle index values. In that respect, we define $\alpha_1(t)$ and $\alpha_2(t)$ the proportion of users in class 1 and class 2 respectively at time $t$ with the highest Whittle index values such that $\alpha_1(t)+\alpha_2(t)=\alpha$.
The remaining proportion of users which are not scheduled at each time slot $t$, which is equal to $1-\alpha$, contains the users with the smallest Whittle index values. Now, regarding this proportion, we give its decomposition into proportions of users at different states in different classes.    
Denoting by $l_1(t)$ and $l_2(t)$ at time $t$ the instantaneous threshold integers under Whittle index policy, then there exists two real values between 0 and 1, $\beta(t)$ and $\gamma(t)$, with $\gamma(t)=1$ and $0 < \beta(t) \leq 1$, or $0 < \gamma(t) \leq 1$ and $\beta(t)=1$, such that: 
\begin{equation}
\sum_{i=1}^{l_1(t)-1} z^1_i(t)+ \sum_{i=1}^{l_2(t)-1} z^2_i(t)+\beta(t) z^1_{l_1(t)}(t)+\gamma(t) z^2_{l_2(t)}(t)=1-\alpha
\end{equation}
and $\{z_i^1\}_{1 \leq i \leq l_1(t)}\cup \{z_i^2\}_{1 \leq i \leq l_2(t)}$ is exactly the set $\{z_i^k: w_k(i) \leq \max(w_1(l_1(t),w_2(l_2(t))\}$. 



In paper \cite{maatouk2020optimality}, in order to prove the convergence of $\boldsymbol z(t)$, the authors assume that $\boldsymbol z(0)$ is within a precise neighborhood of $\boldsymbol z^*$ and they consider that the number of states is finite. These assumptions allow them to find an easy linear relation between $\boldsymbol z(t)$ and $\boldsymbol z(t+1)$ ($\boldsymbol z(t+1)=\boldsymbol Q \boldsymbol z(t)+\boldsymbol c$ see \cite[Section~IV-C]{maatouk2020optimality}), and then deduce the convergence of $\boldsymbol z$ by establishing that the spectral value of $\boldsymbol Q$ is less strictly than one. In our case, as we aim to prove the convergence of $\boldsymbol z$ from any initial state, the relation between $\boldsymbol z(t+1)$ and $\boldsymbol z(t)$ is as follows 
\begin{equation}
\boldsymbol z(t+1)=\boldsymbol Q(\boldsymbol z(t)) \boldsymbol z(t)+\boldsymbol c(t)
\end{equation}
This equation is not linear which makes studying the evolution of $\boldsymbol{z}(\cdot)$ a hard task. Moreover, as the number of state is infinite, then the dimensions of $\boldsymbol{z}(t)$ varies per time. Therefore, the matrix $\boldsymbol{Q}(\boldsymbol{z}(t))$ is not square. Hence we can not apply the same method as in \cite{maatouk2020optimality} since the spectral values are not defined for a non square matrix.
For these reasons, we proceed differently than \cite{maatouk2020optimality}. Our method consists in fact on expressing each proportion $z_i^k(t)$ that belongs to a non scheduled users'  proportion at time $t$ in function of a term of $\alpha_k(\cdot)$ at a given time less than $t$. By this way, we will obtain a part of the vector $\boldsymbol{z}(t)$ in function of $\{\alpha_k(t')\}_{\substack{t' \leq t, \\ k \in \{0,1\}}}$, and the sum of the other part equal to $\alpha$. Then, we show that $\alpha_k(\cdot)$ converges for $k=1,2$. We will see later that it is sufficient to show that $\alpha_k(\cdot)$ converges in order to conclude for the convergence of $\boldsymbol{z}(\cdot)$. To find the partial relation between $\boldsymbol{z}(t)$ and $\{\alpha_k(t')\}_{\substack{t' \leq t \\ k \in \{0,1\}}}$, we prove the following lemma.   
\begin{Lemma}\label{lem:z_evolution}
Knowing $z^k(t)$, $\alpha_k(t)$ and $l_k(t)$, we have that:\\
For $i=1$:\\
$z^k_1(t+1)=p_k \alpha_k(t)$.\\
For $1 \leq i < l_k(t)$:\\
$z^k_{i+1}(t+1)=z^k_i(t)$. 
\end{Lemma}
\begin{proof}
See Appendix \ref{app:lem_z_evolution}.     
\end{proof}
According to Lemma \ref{lem:z_evolution}, after scheduling under the Whittle's Index Policy, we get at time $t+1$, a proportion of $p_1\alpha_1(t)$ of users at state 1 in class 1 and $p_2\alpha_2(t)$ of users at state 1 in class 2 respectively (i.e. $z_1^1(t+1)=p_1\alpha_1(t)$ and $z_1^2(t+1)=p_2\alpha_2(t)$).\\
According to the same lemma, at time $t+2$, a proportion of $p_1\alpha_1(t)$ and $p_2\alpha_2(t)$ of users will go to state 2 in class 1 and class 2 respectively and $p_1\alpha_1(t+1)$, $p_2\alpha_2(t+1)$ of users will move to state 1 in class 1 and class 2 respectively (i.e. $z_1^1(t+2)=p_1\alpha_1(t+1)$, $z_1^2(t+2)=p_2\alpha_2(t+1)$, $z_2^1(t+2)=p_1\alpha_1(t)$ and $z_2^2(t+2)=p_2\alpha_2(t)$).\\
At time $t+3$, a proportion of $p_1\alpha_1(t)$ and $p_2\alpha_2(t)$ of users will go to state 3 in class 1 and class 2 respectively, $p_1\alpha_1(t+1)$, $p_2\alpha_2(t+1)$ of users will move to state 2 in class 1 and class 2 respectively, $p_1\alpha_1(t+2)$ and $p_2\alpha_2(t+2)$ of users will move to state 1 in class 1 and class 2 respectively, (i.e. $z_1^1(t+3)=p_1\alpha_1(t+2)$, $z_1^2(t+3)=p_2\alpha_2(t+2)$, $z_2^1(t+3)=p_1\alpha_1(t+1)$ and $z_2^2(t+3)=p_2\alpha_2(t+1)$, $z_3^1(t+3)=p_1\alpha_1(t)$, $z_3^2(t+3)=p_2\alpha_2(t)$) 
\begin{figure}
\includegraphics[scale=0.9]{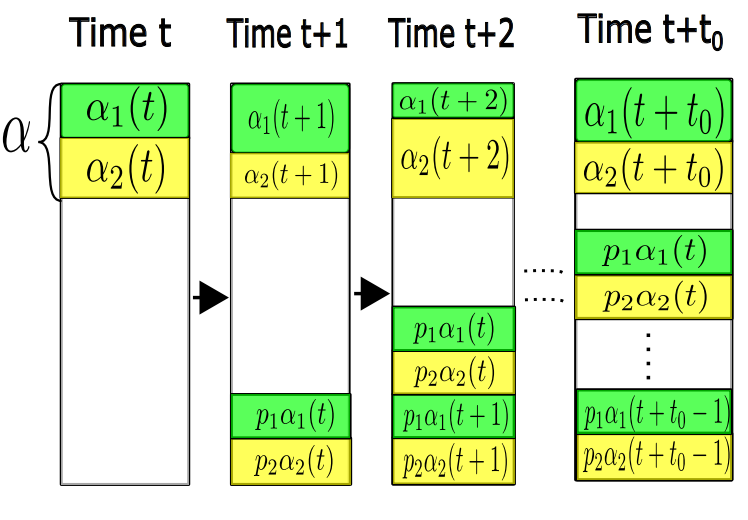}
\caption{Evolution of $z_i^k(\cdot))$ for different states $i$ in function of $\alpha_1(t)$ and $\alpha_2(t)$ under the Whittle Index Policy (the green and the yellow colors refer to class 1 and 2 respectively)}
\end{figure}

Thereby, at time $t+t_0$ where the instantaneous threshold $l_k(t+t_0) \geq t_0$, we get a set of proportions \\ $\{p_1\alpha_1(t),p_2\alpha_2(t),\cdots,p_1\alpha_1(t+t_0-1),p_2\alpha_2(t+t_0-1)\}$ that belong to the proportion $1-\alpha$ of users with the lowest Whittle index values, such that $z_1^1(t+t_0)=p_1\alpha_1(t+t_0-1)$, $z_1^2(t+t_0)=p_2\alpha_2(t+t_0-1)$, $\cdots$, $z^1_{t_0}(t+t_0)=p_1\alpha_1(t)$ and $z^2_{t_0}(t+t_0)=p_2\alpha_2(t)$. Hence, we obtain a $z^k_i(t+1)$ which is well expressed in function of terms of $\alpha_k(\cdot)$ ($k=1,2$) for $i \in [1,t_0]$, $k=1,2$.

\begin{remark}\label{rem:set_form_lowest_whittle_index}
Considering Whittle index policy framework, the order of the different users' proportions with respect to their Whittle index values must be taking into account throughout this analysis.
In fact, as we have already mentioned, we need to give the expression of the non scheduled users' proportions in function of the terms of $\alpha_k(\cdot)$ for $k=1,2$, which can not be done only if we consider the order of the Whittle index values. To that extent, since the set of the non scheduled users'  proportions, according to the Whittle's index policy, is exactly the set of users'  proportions with the lowest Whittle index values among all the different users'  proportions of the system, then the form at time $t$ of this specific set will be $\{z^{k}_{i}(t): w_k(i) \leq w_m(n)\}$ for a given $m$ and $n$ that vary with $t$.   
\end{remark}

Based on this remark above, we need to find at time $t+t_0$, a set of the form $\{z^{k}_{i}(t+t_0): w_k(i) \leq w_m(n)\}$ for a given class $m$ and state $n$, such all the elements of this set are well expressed in function of $\alpha_k(\cdot)$. We show in the sequel that the highest Whittle index of this set could be $w_2(t_0)$.   

Indeed, given that the Whittle index function is increasing with $n$ where $n$ refers to a given age of information state, then for any state in class $2$ with Whittle index less than $w_2(t_0)$, belongs to $[1,t_0]$. Moreover, considering the state $q$ in class $1$ such that $w_1(q)\leq  w_2(t_0)\leq w_1(t_0)$ ($p_1 > p_2$), then $w_1(q) \leq w_1(t_0)$, which means that $q \in [1,t_0]$. Hence, for any element in $\{z^{k}_{i}(t+t_0): w_k(i) \leq w_2(t_0)\}$, can be expressed in function of terms of $\alpha_k(\cdot)$ ($k=1,2$). Accordingly, $\{z^{k}_{i}(t+t_0): w_k(i) \leq w_2(t_0)\}$  equals to the set $\{p_2\alpha_2(t),\cdots,p_2\alpha_2(t+t_0-1),p_1\alpha_1(t+t_0-l(t+t_0)),\cdots,p_1\alpha_1(t+t_0-1)\}$, where $l(t+t_0)$ is the greatest state $q$ in class $1$ such that $w_1(q) \leq w_2(t_0)$. We note that $l(t+t_0) \leq t_0$ because $w_2(l(t+t_0)) \leq w_1(l(t+t_0)) \leq w_2(t_0)$.\\
Therefore, in that regards, for a fixed $t$, we associate for each $t_0$ the corresponding sum $\sum_{j=1}^{l(t+t_0)} z_j^1(t+t_0)+ \sum_{j=1}^{t_0} z_j^2(t+t_0)=\sum_{j=1}^{l(t+t_0)} p_1\alpha_1(t+t_0-j)+ \sum_{j=1}^{t_0} p_2\alpha_2(t+t_0-j)$. To that extent, we define in the following the time $t_0$ when this aforementioned sum exceeds $1-\alpha$.  

\begin{mydef}\label{def:T_t}
Starting at time $t$, we define $T_t$ such that $t+T_t$ is the first time that verifies: 
\begin{equation}\sum_{j=1}^{l(t+T_t)} p_1\alpha_1(t+T_t-j)+ \sum_{i=1}^{T_t} p_2\alpha_2(t+T_t-j) \geq 1-\alpha
\end{equation}  
In other words, the first time when $\sum_{j=1}^{l(t+t_0)} p_1\alpha_1(t+t_0-j)+ \sum_{i=1}^{t_0} p_2\alpha_2(t+t_0-j)$ exceeds $1-\alpha$ is $t+t_0=t+T_t$.  
\end{mydef}

Then, at time $t+T_t$, there exists $l'_1(t+T_t) \leq l(t+T_t)$, $l'_2(t+T_t) \leq T_t$, such that the set $\{z_i^1(t+T_t)\}_{1 \leq i \leq l'_1(t+T_t)} \cup \{z_i^2(t+T_t)\}_{1 \leq i \leq l'_2(t+T_t)}$ is exactly the set $\{z_i^k(t+T_t): w_k(i) \leq \max(w_1(l'_1(t+T_t),w_2(l'_2(t+T_t))\}$\footnote{According to Remark \ref{rem:set_form_lowest_whittle_index}, the form of this set means that it contains the users' proportions with the lowest Whittle index values among all users' proportions of the system}, and $\gamma(t+T_t)=1$ and $0 < \beta(t+T_t) \leq 1$, or $0 < \gamma(t+T_t) \leq 1$ and $\beta(t+T_t)=1$ such that:
\begin{equation}
\sum_{j=1}^{l'_1(t+T_t)-1} p_1\alpha_1(t+T_t-j)+ \sum_{j=1}^{l'_2(t+T_t)-1} p_2\alpha_2(t+T_t-j)+\beta(t+T_t) p_1\alpha_1(t+T_t-l'_1(t+T_t))+\gamma(t+T_t) p_2\alpha_2(t+T_t-l'_2(t+T_t))=1-\alpha,
\end{equation} 
with $l'_1(t+T_t)$ and $l'_2(t+T_t)$ being the instantaneous thresholds in class $1$ and $2$ respectively at time $t+T_t$. $\alpha_1(t+T_t)$ and $\alpha_2(t+T_t)$ are the users'  proportions with the highest Whittle index values, and their sum is equal to $\alpha$. Without loss of generality, we let $l'_k(t+T_t)=l_k(t+T_t)$.

\begin{figure}[H]
\includegraphics[scale=1]{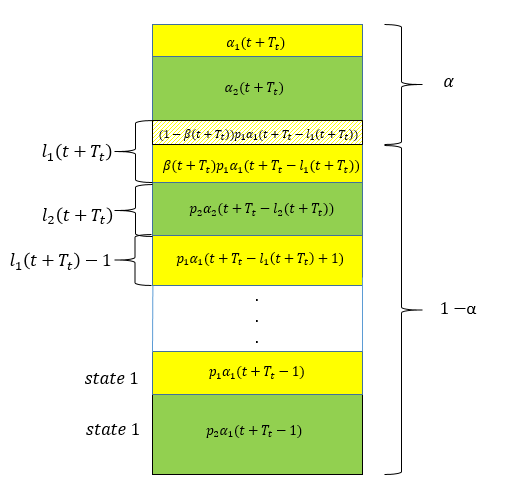}
\caption{The proportions of users at different states at time $t+T_t$ when $\gamma(t+T_t)=1$ and $0< \beta(t+T_t) \leq 1$}
\end{figure}
\begin{figure}[H]
\includegraphics[scale=1]{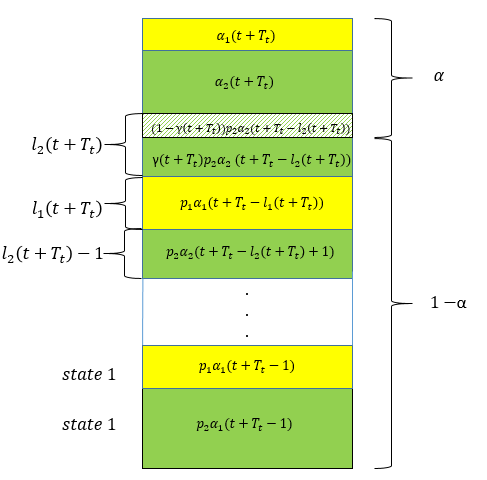}
\caption{The proportions of users at different states at time $t+T_t$ when $\beta(t+T_t)=1$ and $0 <\gamma(t+T_t) \leq 1$}
\end{figure}

As we can see, at time $t+T_t$, all the expressions of the users' proportions that belong to the $1-\alpha$ of users with the smallest Whittle index values, are in function of $\alpha_1(t)$ or $\alpha_2(t)$ at various time. In fact, at time $t+T_t$, we end up with $z_1^1(t+T_t)=p_1\alpha_1(t+T_t-1)$, $z_1^2(t+T_t)=p_2\alpha_2(t+T_t-1)$, $\cdots$, $z^1_{l_1(t+T_t)}(t+T_t)=p_1\alpha_1(t+T_t-l_1(t+T_t))$ and $z^2_{l_2(t+T_t)}(t+T_t)=p_2\alpha_2(t+T_t-l_2(t+T_t))$, and the rest of the proportions belongs to $\alpha_1(t+T_t)$ for class 1 and $\alpha_2(t+T_t)$ for class 2.       
For this reason, we work only with $\alpha_1(\cdot)$ and $\alpha_2(\cdot)$ in order to prove the convergence.
As we have mentioned earlier, the proof of the optimality is valid under an assumption on $\alpha$. This later relies on the maximum value that can take the instantaneous thresholds $l_k(t+T_t)$ at time $t+T_t$ for $k=1,2$. To that extent, we start by defining and bounding a certain constant $T_{\max}$. Then under an assumption on $\alpha$, we show that the order of Whittle index alternates between the two classes in the set $[1,T_{\max}+1]$ (this will be detailed later). Based on this, we establish that $T_{\max}$ is an upper bound of $l_k(t+T_t)$.    

First of all, we give a lemma which will be useful to prove the propositions \ref{prop:alternation_Whittle_index}, \ref{prop:unique_threshold_and_expression} and  \ref{prop:unique_expression_for_all_T}.   
\begin{Lemma}\label{lem:alpha_greater_stric_0}
There exists a time $t_f$ such that for all $t \geq t_f$, $\alpha_1(t) > 0$. 
\end{Lemma}

\begin{proof}
See appendix \ref{app:lem:alpha_greater_stric_0}.    
\end{proof}

In this following definition, we define $T_{\max}$, and we check later that it coincides with the upper bound of $l_k(t+T_t)$ for $k=1,2$. 
\begin{mydef}\label{def:time_max}
Starting at time $t$, we define $T_{\max}$ as $T_t$ defined in Definition \ref{def:T_t}, that verifies the following:
\begin{itemize}
\item $\sum_{j=1}^{l(t+T_t)} p_1\alpha_1(t+T_t-j)+ \sum_{j=1}^{T_t} p_2\alpha_2(t+T_t-j) \geq 1-\alpha$
\item $\alpha_1(t+i)=0$ for all $i \in [0,T_{\max}-1]$ 
\end{itemize}
\end{mydef}

\begin{figure}
\begin{center}
\includegraphics[scale=1]{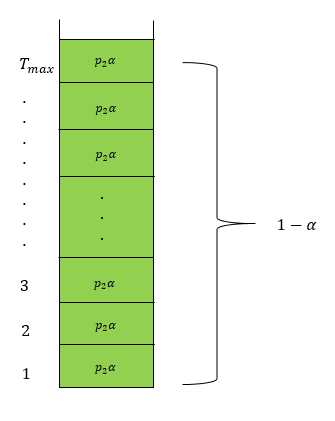}
\end{center}
\caption{Graphical representation of $T_{\max}$}
\end{figure}

In the next lemma, we determine the upper and the lower bound of $T_{\max}$.
\begin{Lemma}\label{lem:upper_bound_threshold_max}
$T_{\max}$ doesn't depend on $t$ and satisfies:
$ \frac{1-\alpha}{p_2 \alpha} \leq T_{\max} \leq \frac{1-\alpha}{p_2 \alpha}+1$.
\end{Lemma}
\begin{proof}
See appendix \ref{app:lem:upper_bound_threshold_max}.
\end{proof}
We say that the order of the Whittle index strictly alternates between the two classes in $[1,n]$ or from state 1 to $n$, if we have $w_2(1)<w_1(1)<w_2(2)<w_1(2)<w_2(3)<w_1(3)<\cdots<w_2(n)<w_1(n)$.
To that extent, the proof of $\alpha_k(\cdot)$ convergence is feasible when the alternation condition is satisfied from $1$ to $l_k(t+T_t)+1$ for all $t$. We note that this condition will be relevant in the proof of the proposition \ref{prop:inequality_satisfied_by_first_element}. To that end, we start by introducing the assumption on $\alpha$. Then, we demonstrate effectively that under this assumption the condition of alternation is satisfied from $1$ to $l_k(t+T_t)+1$.
 
\begin{assumption}\label{assump:assump_on_alpha}
Denoting $\frac{1}{p_1-p_2}(\frac{p_1+p_2}{2}+\sqrt{2(p_1-p_2)+\frac{(p_1+p_2)^2}{4}})$ by $D$.
Then, the users'  proportion scheduled at each time $\alpha$ satisfies: 
\begin{equation}
\alpha > \frac{1}{1+(D-2)p_2}
\end{equation}
\end{assumption}
If $T_{\max}$ is the highest value that $l_k(t+T_t)$ can take, (this will be shown later in proposition \ref{prop:alternation_Whittle_index}), then it is sufficient to prove that the hypothesis of the Whittle index alternation is satisfied from $1$ to $T_{\max}+1$. This will be shown in the next proposition.      
\begin{proposition}\label{prop:interval_alternation_under_assump}
Under Assumption~\eqref{assump:assump_on_alpha}, the order of the Whittle index alternates between the two classes from state 1 to state to $T_{max}+1$. 
\end{proposition}

\begin{proof}
See appendix \ref{app:prop:interval_alternation_under_assump}. 
\end{proof}
Now we prove that the instantaneous thresholds of the two classes can not exceed $T_{max}$. 

\begin{proposition}\label{prop:alternation_Whittle_index}
Denoting by $l_{max}$ the highest instantaneous threshold in the sense that $\forall t\geq t_f ,\max(l_1(t+T_t),l_2(t+T_t)) \leq l_{max}$, then $T_{max}=l_{max}$. 
\end{proposition}
\begin{proof}
See appendix \ref{app:prop:alternation_Whittle_index}
\end{proof}

According to the last proposition, $T_{\max}$ is truly the upper bound of $l_k(t+T_t)$ for all $t$ and $k=1,2$. As consequence, the order of the Whittle index alternates between the two classes in the set $[1,l_k(t+T_t)+1]$. The next goal is to find a relation between $l_1(t+T_t)$ and $l_1(t+T_t)$. 
To do so, we recall that we have at time $t$:
\begin{equation}\label{eq:assocaiated_sum_combinations_z_i_k}
\sum_{i=1}^{l_1(t)-1} z^1_i(t)+ \sum_{i=1}^{l_2(t)-1} z^2_i(t)+\beta(t) z^1_{l_1(t)}(t)+\gamma(t) z^2_{l_2(t)}(t)=1-\alpha
\end{equation} 
with $l_1(t)$ and $l_2(t)$ being the thresholds in class $1$ and $2$ respectively at time $t$, and $\beta(t)=1$ and $0<\gamma(t) \leq 1$, or $\gamma(t)=1$ and $0<\beta(t) \leq 1$. 
Thereby, the first step consists of establishing the relationship between $l_1(t)$ and $l_2(t)$ when $\max(l_1(t),l_2(t)) \leq T_{\max}$ depending on two different cases that we will explain thereafter in order to give a generalized expression of the aforementioned equation \eqref{eq:assocaiated_sum_combinations_z_i_k} where the index of the class is not specified in the expressions of the thresholds $l_1(t)$ and $l_2(t)$.\\
\begin{remark}
It is worth mentioning that, as we have defined $l_{\max}$ in Proposition \ref{prop:alternation_Whittle_index}, it refers to the highest value that can be attained by the thresholds of the class 1 or 2 at time $t+T_t$ for $t>t_f$ where $t_f$ is a given in Lemma \ref{lem:alpha_greater_stric_0}. Whereas, at any time $t>t_f$, $\max(l_1(t),l_2(t))\leq l_{\max}$ might not be true since we don't have necessary a given $t'$ such that $t'+T_{t'}=t$ for any $t>t_f$.    
\end{remark}

\begin{proposition}\label{prop:unique_threshold_and_expression}
At any time $t>t_f$, if $\max(l_1(t),l_2(t)) \leq T_{\max}=l_{\max}$, then there exists $l(t) \leq l_{\max}$ and, $\beta(t)=0$ and $0 < \gamma(t) \leq 1$, or $0 < \beta(t) \leq 1$ and $\gamma(t)=1$ such that:
\begin{equation}\sum_{i=1}^{l(t)-1} z^1_i(t)+ \sum_{i=1}^{l(t)-1}  z^2_i(t)+\beta(t)  z^1_{l(t)}(t)+ \gamma(t) z^2_{l(t)}(t)=1-\alpha
\end{equation} 
\end{proposition}

\begin{proof}
See appendix \ref{app:prop:unique_threshold_and_expression}.
\end{proof}

Starting at time $t \geq t_f$, we have that at time $t+T_t$, the thresholds $l_1(t+T_t)$ and $l_2(t+T_t)$ are less than $l_{\max}$. Hence, according to Proposition \eqref{prop:unique_threshold_and_expression}, there exists $l(t+T_t)$ such that:
\begin{equation}\sum_{j=1}^{l(t+T_t)-1} p_1\alpha_1(t+T_t-j)+ \sum_{j=1}^{l(t+T_t)-1} p_2\alpha_2(t+T_t-j)+\beta(t+T_t) p_1\alpha_1(t+T_t-l(t+T_t))+ \gamma(t+T_t)p_2\alpha_2(t+T_t-l(t+T_t))=1-\alpha
\end{equation} 
where $\beta(t+T_t)=0$ and $0 \leq \gamma(t+T_t) < 1$, or $0 \leq \beta(t+T_t) < 1$ and $\gamma(t+T_t)=1$.\\ 
Denoting $t+T_t$ by $T_0$, we obtain:
\begin{equation}\sum_{j=1}^{l(T_0)-1} p_1\alpha_1(T_0-j)+ \sum_{j=1}^{l(T_0)-1} p_2\alpha_2(T_0-j)+\beta(T_0) p_1\alpha_1(T_0-l(T_0))+ \gamma(T_0)p_2\alpha_2(T_0-l(T_0))=1-\alpha\end{equation} 
where $\beta(T_0)=0$ and $0 < \gamma(T_0) \leq 1$, or $0 < \beta(T_0) \leq 1$ and $\gamma(T_0)=1$.\\

Now, we prove by induction that this latter expression is valid for all $T \geq T_0$, and that $l(T)$, the instantaneous threshold at time $T$, is less than $l_{\max}$.\\

\begin{proposition}\label{prop:unique_expression_for_all_T}
For all $T \geq T_0$, there exists $l(T)\leq l_{\max}$, $\beta(T)$ and $\gamma(T)$, such that:
\begin{equation}\sum_{j=1}^{l(T)-1} p_1\alpha_1(T-j)+ \sum_{j=1}^{l(T)-1} p_2\alpha_2(T-j)+\beta(T) p_1\alpha_1(T-l(T))+ \gamma(T)p_2\alpha_2(T-l(T))=1-\alpha\end{equation} 
where $\beta(T)=0$ and $0 < \gamma(T) \leq 1$, or $0 < \beta(T) \leq 1$ and $\gamma(T)=1$.\\ 
\end{proposition}

\begin{proof}
See appendix \ref{app:prop:unique_expression_for_all_T}.
\end{proof}

According to the latter proposition, we can now define at each time $T \geq T_0$, for each class $k$, the vector $\boldsymbol{A}_k(T)=(\alpha_k(T),\alpha_k(T-1),\cdots,\alpha_k(T-l(T)))$, such that, there exists $\beta(T)$ and $\gamma(T)$:
\begin{equation}
\sum_{j=1}^{l(T)-1} p_1\alpha_1(T-j)+ \sum_{j=1}^{l(T)-1} p_2\alpha_2(T-j)+\beta(T) p_1\alpha_1(T-l(T))+\gamma(T)p_2\alpha_2(T-l(T))=1-\alpha
\end{equation} 
where $\beta(T)=0$ and $0 < \gamma(T) \leq 1$, or $0 < \beta(T) \leq 1$ and $\gamma(T)=1$. We note that as we have explained previously, the relation between $\boldsymbol{A}_k(T)$ and $\boldsymbol{z}^k(T)$ is: $p_k\alpha_k(T-1)=z^k_1(T),p_k\alpha_k(T-2)=z^k_2(T),\cdots, p_k\alpha_k(T-l(T))=z^k_{l(T)}(T)$.
\begin{remark}
We emphasize that in the following analysis, $T$ is always considered greater than $T_0$.
\end{remark}
We prove in the sequel that $\max \boldsymbol{A}_k(T)$ is decreasing and $\min \boldsymbol{A}_k(T)$ is increasing (with the $\max$ and $\min$ referring to the element of the vector with the greatest value, and the smallest value respectively). After that, we conclude the convergence of $\max \boldsymbol{A}_k(T)$ and $\min \boldsymbol{A}_k(T)$ when $T$ tends to $+\infty$. Then, we prove that they must converge to the same real number.
In order to prove that $\max \boldsymbol{A}_k(T)$ is decreasing and $\min \boldsymbol{A}_k(T)$ is increasing, we first demonstrate this following proposition.

\begin{proposition}\label{prop:elements_inclusion}
All the elements of the vector $\boldsymbol{A}_k(T+1)$ belong to the elements of the vector $\boldsymbol{A}_k(T)$ except $\alpha_k(T+1)$.
\end{proposition}
\begin{proof}
See appendix \ref{app:prop:elements_inclusion} 
\end{proof}

With the intention of proving the monotony of $\max \boldsymbol{A}_1(T)$ and $\min \boldsymbol{A}_1(T)$, we still need to prove that the value of $\alpha_1(T+1)$ must be less than $\max \boldsymbol{A}_1(T)$ and greater than $\min \boldsymbol{A}_k(T)$. 
For that, we introduce the following proposition.\\
Before doing that, we note that, as $\alpha_1(t)+\alpha_2(t)=\alpha$ at each time slot $t$, then it is sufficient for us to prove that $\alpha_1(\cdot)$ is converging. To that extent, we study only the vector function $\boldsymbol{A}_1(T)$ in order to prove the convergence.\\

\begin{proposition}\label{prop:inequality_satisfied_by_first_element}
Under assumption \ref{assump:assump_on_alpha}, for a given vector $\boldsymbol{A}_1(T)=(\alpha_1(T), \alpha_1(T-1), \cdots,\alpha_1(T-l(T)))(T \geq T_0)$, we have four possible cases of inequalities:
 $$\alpha_1(T) \leq \alpha_1(T+1) \leq \alpha_1(T-l(T))$$
 $$\alpha_1(T-l(T)) \leq \alpha_1(T+1) \leq \alpha_1(T)$$
 $$\alpha_1(T-l(T)+1) \leq \alpha_1(T+1) \leq \alpha_1(T)$$
 $$\alpha_1(T) \leq \alpha_1(T+1) \leq \alpha_1(T-l(T)+1)$$

Moreover:
If $\alpha_1(T) \leq \alpha_1(T+1) \leq \alpha_1(T-l(T))$, then: 
 $$\alpha_1(T+1)-\alpha_1(T) \leq p_1(\alpha_1(T-l(T))-\alpha_1(T))$$

If $\alpha_1(T-l(T)) \leq \alpha_1(T+1) \leq \alpha_1(T)$, then:
$$\alpha_1(T)-\alpha_1(T+1) \leq p_1(\alpha_1(T)-\alpha_1(T-l(T)))$$

If $\alpha_1(T-l(T)+1) \leq \alpha_1(T+1) \leq \alpha_1(T)$, then:
$$\alpha_1(T)-\alpha_1(T+1) \leq p_1(\alpha_1(T)-\alpha_1(T-l(T)+1))$$

If $\alpha_1(T) \leq \alpha_1(T+1) \leq \alpha_1(T-l(T)+1)$, then:
$$\alpha_1(T+1)-\alpha_1(T) \leq p_1(\alpha_1(T-l(T)+1)-\alpha_1(T))$$

\end{proposition}
\begin{proof}
See appendix \ref{app:prop:inequality_satisfied_by_first_element}.
\end{proof}

\begin{theorem}\label{theor:convergence}
$\min \boldsymbol{A}_1(T)$ and $\max \boldsymbol{A}_1(T)$ converge and we denote their limits respectively by $l_1$ and $l_2$.  
\end{theorem}
\begin{proof}
According to Proposition \ref{prop:elements_inclusion}, the elements of the vector $\boldsymbol{A}_1(T+1)$ except the first element which is $\alpha_1(T+1)$ belong to the elements of the vector $\boldsymbol{A}_1(T)$. Hence, the values of these elements (except the first element of $\boldsymbol{A}_1(T+1)$) is less than $\max \boldsymbol{A}_1(T)$ and greater than $\min \boldsymbol{A}_1(T)$. 
According to the first result of Proposition \ref{prop:inequality_satisfied_by_first_element}, we deduce that $\alpha_1(T+1)$ is between two values of two elements of the vector $\boldsymbol{A}_1(T)$. Hence, combining the results of Proposition \ref{prop:elements_inclusion} and \ref{prop:inequality_satisfied_by_first_element}, $\max \boldsymbol{A}_1(T+1) \leq \max \boldsymbol{A}_1(T)$ and $\min \boldsymbol{A}_1(T+1)\geq \min \boldsymbol{A}_1(T)$.
Then $\max \boldsymbol{A}_1(T)$ is decreasing with $T$ and $\min \boldsymbol{A}_1(T)$ is increasing with $T$. Given that for all $T$, $0 \leq \alpha_1(T) \leq \alpha$, then $\max \boldsymbol{A}_1(T)$ and $\min \boldsymbol{A}_1(T)$ are bounded by $0$ and $\alpha$. 
Therefore, we can conclude that $\min \boldsymbol{A}_1(T)$ and $\max \boldsymbol{A}_1(T)$ converge and we denote their limits by  $l_1$ and $l_2$ respectively. Moreover  $\max \boldsymbol{A}_1(T)$ is lower bounded by $l_2$ and  $\min \boldsymbol{A}_1(T)$ is upper bounded by $l_1$.
\end{proof}

However, in order to have $\alpha_1(T)$ converges to a unique point, we need to establish that $\max \boldsymbol{A}_1(T)$ and $\min \boldsymbol{A}_1(T)$ converge to the same limit. In other words, we need to prove that $l_1=l_2$. For that, we will use the second result of Proposition \ref{prop:inequality_satisfied_by_first_element}. To that extent, we proceed by contradiction, i.e. we suppose that $l_1 \neq l_2$. More specifically, given that $l_1 \leq l_2$ by definition, the two possible cases satisfied by $l_1$ and $l_2$ are: $l_1 < l_2$ or $l_1 =l_2$, then to show that  $l_1=l_2$, it is sufficient to find a contradiction considering $l_1<l_2$.\\
In fact, we prove that if $l_1 < l_2$, there exists $T_d$ such that all the elements of $\boldsymbol{A}_1(T_d)$ are less strictly than $l_2$, that contradicts with the fact that $\max \boldsymbol{A}_1(T)$ is lower bounded by $l_2$.\\
As $\max \boldsymbol{A}_1(T)$ converges to $l_2$, then for a given $\epsilon >0$, there exists a given time slot that we denote by $T_{\epsilon}\geq T_0$ such that for all $T \geq T_{\epsilon}$, $\max \boldsymbol{A}_1(T) < l_2+\epsilon$.
Our proof consists of showing that for a small enough $\epsilon$, there exists $T \geq T_{\epsilon}$, $\max \boldsymbol{A}_1(T)$ is less strictly than $l_2$.  
We need first to determine an upper bound of the number of the elements of the vector $\boldsymbol{A}_1(T)$ whatever $T$. In fact, as we have demonstrated that at each time $T$, the instantaneous threshold $l(T)$ is less than $l_{\max}$. Then the number of the elements of $\boldsymbol{A}_1(T)$ will not exceed $l_{\max}+1$. In the following proof, we denote $l_{\max}$ by $L$.\\
\begin{proposition}\label{prop:max_alpha_less_l_2}
If $l_1<l_2$, for $\epsilon \leq (l_2-l_1)\frac{(1-p_1)^L}{1-(1-p_1)^L}$, there exist $T_d \geq T_{\epsilon}$ such that all the elements of $\boldsymbol{A}_1(T_d)$ are less strictly than $l_2$.
\end{proposition}
\begin{proof}
See appendix \ref{app:prop:max_alpha_less_l_2}.
\end{proof}
Providing that $l_2$ is a lower bound of $\max \boldsymbol{A}_1(T)$ which contradicts with the result of the above proposition. Hence, the supposition of $l_1 \neq l_2$ is not valid.\\
Therefore, $l_1=l_2$. Consequently, $\max \boldsymbol{A}_1(T)$ and $\min \boldsymbol{A}_1(T)$ converge to the same limit denoted $\alpha_1^*$. Given that $\min \boldsymbol{A}_1(T) \leq \alpha_1(T) \leq \max \boldsymbol{A}_1(T)$ for all $T$, then $\alpha_1(T)$ also converges to $\alpha_1^*$. Similarly, $\alpha_2(T)$ converges to $\alpha-\alpha_1^*=\alpha_2^*$. 
In the following proposition, we prove that $\boldsymbol{z}(t)$ converges.

\begin{proposition}\label{prop:z_convergence}
If $\alpha_k(t)$ converges to $\alpha_k^*$, then for each state $i$ and class $k$, $z^k_i(t)$ converges to $z^{k,*}_i$.
\end{proposition}
\begin{proof}
See appendix \ref{app:prop:z_convergence}.
\end{proof}

However, we still have to establish that the stochastic vector $\boldsymbol{Z}^N(t)$ converges to $\boldsymbol{z}^*$ in probability when $N$ scales. For that, we introduce the following proposition inspired from the discrete-time version of Kurtz Theorem in \cite{kurtz1978strong}.
Before that, knowing that the norms on the infinite dimension vector space are not equivalents, we work only with a specific norm which will be useful to show the optimality of the Whittle index's policy. Accordingly, we define $||\cdot||$ as follows:
\begin{equation}
||\boldsymbol{v}||=\sum_{i=1}^{+\infty}|v_i^1|i+\sum_{i=1}^{+\infty}|v_i^2|i
\end{equation}
where $v_i^k$ is the $i$-th component in the class $k$ of the vector $\boldsymbol{v}$. The reason behind chosen a such norm will be revealed in the proof of Proposition \ref{prop:optimality_whittle_index}. 
\begin{proposition} \label{prop:kurth_theorem_age}
For any $\mu>0$ and finite time horizon $T$, there exists positive
constant $C$ such that
\[P_{\boldsymbol{x}}(\underset{0 \leq t < T}{\text{sup}} ||\boldsymbol{Z}^N(t)-\boldsymbol{z}(t)|| \geq \mu) \leq \frac{C}{N}\]
where $P_{\boldsymbol{x}}$ denotes the probability conditioned on the initial state $\boldsymbol{Z}^N(0) = \boldsymbol{z}(0)= \boldsymbol{x}$. Furthermore, $C$ is independent of $N$.
\end{proposition}

\begin{proof}
See appendix \ref{app:prop:kurth_theorem_age}.
\end{proof}

According to the Proposition above, the system state $\boldsymbol{Z}^N(t)$ behaves very close to the fluid approximation model $\boldsymbol{z}(t)$ when the number of users $N$ is large and starting from any initial state.
To that extent, in order to establish the optimality of Whittle's index policy, we give first this following lemma which is a consequence of the Proposition \ref{prop:kurth_theorem_age}.

\begin{Lemma}\label{lem:result_kurth_theorem}
For any $\mu > 0$, there exists a time $T_0$ such that for each $T > T_0$, there
exists a positive constant $s$ with,
\[P_{\boldsymbol{x}}(\underset{T_0 \leq t < T}{\text{sup}} ||\boldsymbol{Z}^N(t)-\boldsymbol{z}^*|| \geq \mu) \leq \frac{s}{N}\]
\end{Lemma}

\begin{proof}
See appendix \ref{app:lem:result_kurth_theorem}
\end{proof}

We remind that starting from an initial state $\boldsymbol{x}$, our objective is to compare the total expected average age per user under Whittle index policy which can be expressed as $\frac{1}{T} \mathbb{E}^{wi}\left[ \sum_{t=0}^{T-1}  \sum_{k=1}^{K} \sum_{i=1}^{+\infty} Z_i^{k,N}(t)i \mid \boldsymbol{Z}^N(0)=\boldsymbol{x}\right]$ where $\boldsymbol{Z}^N(t)$ evolves under Whittle index policy, with the optimal age of the relaxed problem per user whose expression in function of $\boldsymbol{z}^*$ is, $C^{RP}=\frac{C^{RP,N}}{N} =\sum_1^K \sum_{i=1}^{+\infty} z_i^{k,*}i$, when the number of users $N$ as well as the time duration $T$ grow.

According to Lemma \ref{lem:result_kurth_theorem}, we are ready now to establish the asymptotic optimality of the Whittle index policy.

\begin{proposition}\label{prop:optimality_whittle_index}
Starting from a given initial state $\boldsymbol{Z}^N(0)=\boldsymbol{z}(0)=\boldsymbol{x}$, then:
\begin{equation}
\underset{ T \rightarrow +\infty }{\text{lim}} \lim_{N \rightarrow \infty}  \frac{1}{T} \mathbb{E}^{wi}\left[ \sum_{t=0}^{T-1}  \sum_{k=1}^{K} \sum_{i=1}^{+\infty} Z_i^{k,N}(t)i \mid \boldsymbol{Z}^N(0)=\boldsymbol{x}\right]=\sum_{k=1}^K  \sum_{i=1}^{+\infty} z_i^{k,*}i 
\end{equation} 
\end{proposition}

\begin{proof}
See appendix \ref{app:prop:optimality_whittle_index}.
\end{proof}

\section{Numerical Results}\label{sec:Num_Reslts}
\subsection{Verification of assumption \ref{assump:assump_on_alpha}}\label{subsec:verif_assump}
In this section, we compute the value of the lower bound on $\alpha$ given in Assumption \ref{assump:assump_on_alpha}. We denote this lowerbound by $B_{\alpha}$. For a wide range of parameters $p_1$ and $p_2$, we provide an exhaustive table that represents the lower bound on $\alpha$ in function of $p_1$ and $p_2$. As can be seen, the lowerbound decreases when $p_1$ and $p_2$ are close one to the other. Moreover, it grows even smaller when $p_1$ and $p_2$ have relatively high values.   
\begin{table}
\begin{minipage}{0.3\linewidth}
\centering
\normalsize
\begin{tabular}{|c|c|c|}
\hline
$p_1$&  $p_2$ & $B_{\alpha}$ \\
\hline
0.1 & 0.2 & 0.7034 \\
\hline
0.2 & 0.4 & 0.6250 \\
\hline
0.3 & 0.5 & 0.4711 \\
\hline
0.4 & 0.6 & 0.3556 \\
\hline
0.4 & 0.8 & 0.5328 \\
\hline
0.5 & 0.8 & 0.3612 \\
\hline
0.5 & 1 & 0.5 \\
\hline
0.6 & 0.9 & 0.2893 \\
\hline
0.7 & 0.9 & 0.1675 \\
\hline
0.8 & 0.9 & 0.1351 \\
\hline
\end{tabular}
\caption{\label{fig:table_bound_alpha} Evaluation of $B_{\alpha}$ for a wide range of channel statics}
\end{minipage}
\hfill
\begin{minipage}{0.6\linewidth}
\centering
\includegraphics[scale=0.6]{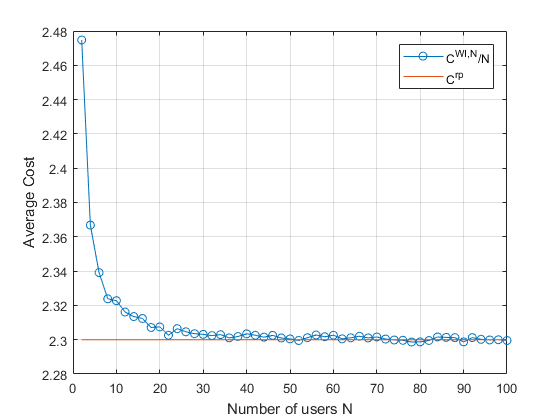}
\captionof{figure}{\label{fig:comparison_WI_RP} Average age per-user under the Whittle's index policy}
\end{minipage}
\end{table}
According to table \ref{fig:table_bound_alpha}, we can notice that in most cases of $(p_1,p_2)$, the lower bound of $\alpha$ doesn't exceed $0.5$. This implies that the interval of $\alpha$ where the assumption \ref{assump:assump_on_alpha} is satisfied, is enough wide for different values of $p_1$ and $p_2$.
\subsection{Implementation of the Whittle's index policy}
In this section, we evaluate the Whittle's index policy's performance by comparing the per-user average age of the Whittle's index policy to the optimal per-user average age of the relaxed problem $C^{rp}$. To that extent, we let the number of users in class $1$ and class $2$ to be equal to $\frac{N}{2}$. The probability of successful transmission of class $1$ and class $2$ are set to $0.8$ and $0.5$, respectively. At each time slot $t$, at most, $M=\frac{N}{2}$ of users can be scheduled per each time slot, i.e., $\alpha=\frac{M}{N}=\frac{1}{2}$.  As seen in Figure \ref{fig:comparison_WI_RP}, the gap between the two policies tightens as the number of users $N$ grows. Indeed, these numerical results corroborate our theoretical analysis and show that the Whittle's index policy is effectively globally asymptotically optimal.  
\section{Conclusion} \label{sec:Conclusion}
In this paper, we have examined the average age minimization problem where only a fraction of the network users can transmit simultaneously over unreliable channels. We presented and derived a novel method based Cauchy criterion to prove the Whittle's index policy's optimality in the many-users regime. Compared to the state of the art methods, our approach does not require imposing strict mathematical assumptions, which can be challenging to verify. We also provided numerical results that corroborate our theoretical findings and highlight the Whittle's index policy's performance. Moving forward, the next research direction is to extend our proof to various other scheduling problems under different system models and objective functions.   
\bibliographystyle{IEEEtran} 
\bibliography{bil_conv}

\begin{appendices}

\section{Proof of Lemma \ref{lem:z_evolution}}\label{app:lem_z_evolution}
We can formulate the fluid limit equation \eqref{eq:fluid_approxiamtion} as follows:
\[\boldsymbol{z}(t+1)=\mathbb{E}\left[\boldsymbol{Z}^N(t+1)\Big | \boldsymbol{Z}^N(t)=\boldsymbol{z}(t)\right]\]
At time $t+1$, applying Whittle index policy, in average exactly a proportion of $p_k \alpha_k(t)$ of users will be at state one since $\alpha_k(t)$ refers to the proportion of users in class $k$ that are scheduled. Accordingly, $z^k_1(t+1)=p_k \alpha_k(t)$. While for $1 \leq i < l_k(t)$, the users' proportion $z^k_i(t)$ is not scheduled. Therefore at time $t+1$, since prescribing idle action to a given user implies that its state will be increased by $1$, the proportion $z^k_i(t)$ at state $i$ in class $k$ will be at state $i+1$. Thus, $\mathbb{E}\left[Z^{N,k}_{i+1}(t+1)\Big | \boldsymbol{Z}^N(t)=\boldsymbol{z}(t)\right]=z^k_{i+1}(t+1)=z^k_i(t)$..

\section{Proof of Lemma \ref{lem:alpha_greater_stric_0}}\label{app:lem:alpha_greater_stric_0}
First of all, we provide an useful lemma.
\begin{Lemma}\label{lem:difference_between_whittle_index_at_i_i+1}
We have for all integer $i$ and for $k=1,2$: $$w_k(i+1)-w_k(i)=ip_k+1$$ 
\end{Lemma}
\begin{proof}
renewcommand{\qedsymbol}{$\blacksquare$}
The result can be obtained directly by replacing $w_k(i)$ by its expression. 
\end{proof}
In order to prove the present lemma, we proceed in two steps:
\begin{itemize}
\item We prove first by contradiction that there exists a given time $t_f$ such that $\alpha_1(t_f)>0$.
\item We prove that if $\alpha_1(t_f) >0$, then $\alpha_1(t)> 0$ for all $t\geq t_f$.
\end{itemize}
\begin{enumerate}
\item For the first point, we suppose that for all $t$, we have that $\alpha_1(t)=0$. Consequently, we get that $z_1^1(t+T_t)=0, \cdots, z^1_{l_1(t+T_t)}(t+T_t)=0$, and $\alpha_1(t+T_t)=0$. This means that, the proportion of all users in class $1$ is equal to $0$. However, the users' proportion of class $1$ is $\gamma_1 \neq 0$. That is, there exists a given time $t_f$ such $\alpha_1(t_f)>0$.\\

\item Before addressing the second point, we recall that $\alpha_1(t)$ refers to the scheduled users' proportion in the class $1$. Thereby, $\alpha_1(t)$ contains all users with the highest Whittle index values among all users in class $1$.
To that extent, at time $t_f$, the Whittle index of $\alpha_1(t_f)$ is greater than the Whittle index of the users' proportion $1-\alpha$ that we denote by $C$. We let $\boldsymbol{S}_{t_f}(C)$ be the set of pair (state,class) at time $t_f$ in the users' proportion $C$. Denoting by $q$ the smallest state of $\alpha_1(t_f)$, $n$ and $m$ a given state and class respectively such that $z_n^m(t)$ belongs to $C$ at time $t_f$, then $w_m(n) \leq w_1(q)$. Under the Whittle index policy, at time $t_f+1$, the states of a users' proportion that equals to $(1-p_1) \alpha_1(t_f)$ among the users' proportion $\alpha_1(t_f)$, will be increased by one in comparison with the time slot $t_f$, as well as the users' proportion $C$. Accordingly, the smallest state of the proportion $(1-p_1) \alpha_1(t_f)$, is $q+1$. $\boldsymbol{S}_{t_f+1}(C)$ is shifted of one with respect to $\boldsymbol{S}_{t_f}(C)$, i.e., $(n,m) \in \boldsymbol{S}_{t_f}(C) \Leftrightarrow (n+1,m) \in \boldsymbol{S}_{t_f+1}(C)$. We compare $w_1(q+1)$ with the Whittle index of $n$ in class $m$ such that $(n,m) \in \boldsymbol{S}_{t_f+1}(C)$.  
In that direction, we let $(n,m) \in \boldsymbol{S}_{t_f+1}(C)$, and we distinguish between two cases:
\begin{itemize}
\item $m=1$:
Leveraging the fact that $(n-1,m) \in \boldsymbol{S}_{t_f}(C)$, then $w_1(q) \geq w_1(n-1)$. That implies that $n-1 \leq q$ since $w_k(.)$ is increasing. Hence $n \leq q+1$. As consequence, $w_1(n) \leq w_1(q+1)$
\item $m=2$:
Again we distinguish between two case:\\
\begin{itemize}
\item If $n-1 \leq q$, then $w_2(n) < w_1(n) \leq w_1(q+1)$.\\
Therefore, we obtain our desired result for the first case.
\item If $n-1 > q$:\\
We have that:
$$w_1(q+1)-w_2(n)=(w_1(q+1)-w_1(q))- (w_2(n)-w_2(n-1))+w_1(q)-w_2(n-1)$$
Applying Lemma \ref{lem:difference_between_whittle_index_at_i_i+1}, we obtain: $(w_1(q+1)-w_1(q))- (w_2(n)-w_2(n-1)) = qp_1-(n-1)p_2$.
Given that $w_2(n-1) \leq w_1(q)$, therefore replacing by their expressions we get:
$$(n-2)(n-1)p_2/2+n-1 \leq (q-1)qp_1/2+q$$
As $n-1 > q$, then: 
$$(n-2)(n-1)p_2/2 \leq (q-1)qp_1/2$$
Hence:
$$(n-1)p_2 \leq qp_1$$
Therefore, $(w_1(q+1)-w_1(q))- (w_2(n)-w_2(n-1)) \geq 0$.
Hence, knowing that $w_1(q)-w_2(n-1)\geq 0$ we end up with our desired result for this case, i.e. $w_1(q+1)-w_2(n) \geq 0$.
\end{itemize}
\end{itemize}
Thus, we have proved that at time $t_f+1$, all the users' proportions in $C$ whose sum is equal to $1-\alpha$ have a Whittle index less than that of $(1-p_1)\alpha_1(t_f)$ defined in the beginning of this proof. That means that there exists at least a users' proportion that equals to $1-\alpha$ with Whittle index values less than those of the states of the users' proportion $(1-p_1)\alpha_1(t_f)$. Then surely, the users' proportion $(1-p_1)\alpha_1(t_f)$ that is different from $0$ belongs to the users' proportion $\alpha$ with the highest Whittle index values.
This implies that surely at time $t_f+1$, there will be at least one queue in class 1 belonging to $\alpha$ with the highest Whittle index values. Therefore, we have that $\alpha_1(t_f+1) >0$. This result can be generalized for all $t \geq t_f$. In other words, we have for all $t \geq t_f$, $\alpha_1(t) > 0$.
\end{enumerate}

\section{Proof of Lemma \ref{lem:upper_bound_threshold_max}}\label{app:lem:upper_bound_threshold_max}
As $\alpha_1(j)+\alpha_2(j)=\alpha$ for all integers $j$, then, if $\alpha_1(t+i)=0$, $\alpha_2(t+i)=\alpha$.
For $j \in [1,T_{\max}]$, we have that $T_{\max}-j \in [0,T_{\max}-1]$. This means that $\alpha_1(t+T_{\max}-j)$ is equal to $0$, which implies that $\alpha_2(t+T_{\max}-j)=\alpha$. Moreover, knowing that $l(t+T_{\max}) \leq T_{\max}$, then for all $j \in [1,l(t+T_{\max})]$, $T_{\max}-j \in [T_{\max}-l(t+T_{\max}),T_{\max}-1] \subset [0,T_{\max}-1]$. Hence, we get that $\alpha_1(t+T_{\max}-j)=0$, for all $j \in [1,l(t+T_{\max})]$.\\
Therefore, according to the definition \ref{def:T_t}, $T_{\max}$ satisfies:
\begin{equation}T_{\max} p_2 \alpha \geq 1-\alpha
\end{equation}
\begin{equation}T_{max} \geq  \frac{1-\alpha}{p_2 \alpha}
\end{equation}  
Providing that $T_{\max}$ by definition is the first time when $\sum_{j=1}^{l(t+T_{\max})} p_1\alpha_1(t+T_{\max}-j)+ \sum_{j=1}^{T_{\max}} p_2\alpha_2(t+T_{\max}-j)$ exceeds $1-\alpha$, then at time $t+T_{\max}-1$,  $\sum_{j=1}^{l(t+T_{\max}-1)} p_1\alpha_1(t+T_{\max}-1-j)+ \sum_{j=1}^{T_{\max}-1} p_2\alpha_2(t+T_{\max}-1-j) < 1-\alpha$. This latter sum is equal to $(T_{\max}-1)p_2 \alpha$ which is less than $1-\alpha$.  
Therefore, we have as result that $T_{\max} < \frac{1-\alpha}{p_2 \alpha}+1$. As there is one integer value between $\frac{1-\alpha}{p_2 \alpha}$ and $\frac{1-\alpha}{p_2 \alpha}+1$, then $T_{\max}$ doesn't depend on $t$, and satisfies:
$\frac{1-\alpha}{p_2 \alpha} \leq T_{\max} < \frac{1-\alpha}{p_2 \alpha}+1$..

\section{Proof of Proposition \ref{prop:interval_alternation_under_assump}}\label{app:prop:interval_alternation_under_assump}
We have that $w_1(n)=\frac{(n-1)p_1n}{2}+n$, and $w_2(n)=\frac{(n-1)p_2n}{2}+n $.
We start first by finding the set of states for which the Whittle index alternate between the two classes.
As we can see from the expression of the Whittle index, for a given state $n$, $w_2(n) < w_1(n)$ as $p_2 < p_1$. In order to have the condition of alternation strictly satisfied for any given state $n$, we must have $w_1(n) < w_2(n+1)$.
Hence, denoting by $f(n)$ the difference $w_2(n+1)-w_1(n)$, we study the sign of $f(n)$ to see for which $n$ $f$ is strictly positive.
\begin{Lemma}\label{lem:alternation_interval}
For all $n \in [1,D[$, $f(n) > 0$
\end{Lemma}
\begin{proof}
\renewcommand{\qedsymbol}{$\blacksquare$}
We have that:
\begin{equation}\label{eq:difference_function}
f(n)=\frac{n^2}{2}(p_2-p_1)+\frac{n}{2}(p_1+p_2)+1
\end{equation}
Hence: 
\begin{equation}
f'(n)=n(p_2-p_1)+\frac{p_1+p_2}{2}
\end{equation}
The derivative is equal to zero for $n=\frac{p_1+p_2}{2(p_1-p_2)}$, which is greater strictly than $0$. This means that $f$ is strictly increasing in $[0,\frac{p_1+p_2}{2(p_1-p_2)}]$ since $f'(n) > 0$ in $[0,\frac{p_1+p_2}{2(p_1-p_2)}[$. Providing that $f(0)=1$, then surely $f$ is strictly positive in $[0,\frac{p_1+p_2}{2(p_1-p_2)}]$. This means that, the unique positive solution for $f(n)=0$ must be in the interval $[\frac{p_1+p_2}{2(p_1-p_2)},+\infty[$, as $\underset{n \rightarrow +\infty}{\text{lim}} f(n)=-\infty$.
Indeed, the unique solution $n_0$ of $f(n)=0$ in $[\frac{p_1+p_2}{2(p_1-p_2)},+\infty[$ is the biggest root of the polynomial \eqref{eq:difference_function} which is exactly the value $D$ introduced in Assumption \ref{assump:assump_on_alpha}.
As the function $f$ is decreasing in $[ \frac{p_1+p_2}{2(p_1-p_2)}, +\infty[$, then $f$ is strictly positive in $[0,D[$. Therefore, $f(n)>0$ for $n \in [1,D[$, which concludes the proof.
\end{proof}
According to Lemma \ref{lem:alternation_interval}, the order of the Whittle index strictly alternates between the two states when $n \in [1,D[$.
Therefore, we need to prove that $T_{\max}+1$ is upper bounded by $D$ in order to prove that the alternation condition is satisfied from state $1$ to $T_{\max}+1$.\\
Indeed, as we have found an upper bound of $T_{\max}$ which is equal to $\frac{1-\alpha}{p_2 \alpha}+1$ (according to Lemma \ref{lem:upper_bound_threshold_max}), we just need to prove that $\frac{1-\alpha}{p_2 \alpha}+2$ is strictly less than $D$.\\
Under assumption \eqref{assump:assump_on_alpha}, we have that:
\begin{align}
\alpha &> \frac{1}{1+(D-2)p_2}\\
\alpha(1+p_2(D-2)) &> 1 \\
\alpha p_2(D-2) &> 1-\alpha\\ 
D-2 &> \frac{1-\alpha}{p_2 \alpha}\\ 
D &> \frac{1-\alpha}{p_2 \alpha}+2 
\end{align}
Hence, from state $1$ to $T_{\max}+1$, the order of the Whittle index strictly alternates between the two classes.
Accordingly, the proof is concluded.

\section{Proof of Proposition \ref{prop:alternation_Whittle_index}}\label{app:prop:alternation_Whittle_index}
We present first a lemma which will be helpful in proving this proposition as well as the next ones.
\begin{Lemma}\label{lem:inequalities_whittle_index_thresholds}
For any state $q$, at any time $t$, we have that:
$$w_1(q) \leq w_2(l_2(t)) \Rightarrow w_1(q) \leq w_1(l_1(t))$$
and
$$w_2(q) \leq w_1(l_1(t)) \Rightarrow w_2(q) \leq w_2(l_2(t))$$ 
\end{Lemma}
\begin{proof}
\renewcommand{\qedsymbol}{$\blacksquare$}
See appendix \ref{app:lem:inequalities_whittle_index_thresholds}
\end{proof}
We consider $t \geq t_f$. After time $T_t$, we have that: 
\begin{equation}\sum_{j=1}^{l(t+T_t)} p_1\alpha_1(t+T_t-j)+ \sum_{j=1}^{T_t} p_2\alpha_2(t+T_t-j) \geq 1-\alpha
\end{equation}
Then, as it has been showcased, at time $t+T_t$, there exists $l_1(t+T_t) \leq l(t+T_t)$, $l_2(t+T_t) \leq T_t$, $\gamma(t+T_t)=1$ and $0 < \beta(t+T_t) \leq 1$; or $0 < \gamma(t+T_t) \leq 1$ and $\beta(t+T_t)=1$ such that: 
\begin{equation}
\sum_{j=1}^{l_1(t+T_t)-1} p_1\alpha_1(t+T_t-j)+ \sum_{j=1}^{l_2(t+T_t)-1} p_2\alpha_2(t+T_t-j)+\beta(t+T_t) p_1\alpha_1(t+T_t-l_1(t+T_t))+\gamma(t+T_t) p_2\alpha_2(t+T_t-l_2(t+T_t))=1-\alpha
\end{equation}
with $l_1(t+T_t)$ and $l_2(t+T_t)$ being the instantaneous thresholds in class 1 and 2 respectively at time $t+T_t$.\\
Now, we prove by contradiction that $\max(l_1(t+T_t),l_2(t+T_t))\leq T_{max}$.\\
We prove first that $l_2(t+T_t)$ is greater than $l_1(t+T_t)$.\\
As we have that $w_2(l_1(t+T_t)) < w_1(l_1(t+T_t))$, then according to lemma \ref{lem:inequalities_whittle_index_thresholds}, $w_2(l_1(t+T_t)) \leq w_2(l_2(t+T_t))$. This implies that $l_2(t+T_t)$ is greater than $l_1(t+T_t)$.\\ 
Reasoning by contradiction, we suppose that $l_2(t+T_t) > T_{\max}$ ($l_2(t+T_t)=\max(l_1(t+T_t),l_2(t+T_t))> T_{\max}$). Based on this, we have that $w_1(T_{\max}) < w_2(l_2(t+T_t))$ because $w_1(T_{\max})< w_2(T_{\max}+1) \leq w_2(l_2(t+T_t))$ since the order of the Whittle index alternates between the two classes as it has been proved in Proposition \ref{prop:alternation_Whittle_index}.
To that extent, we distinguish between two cases:\\
1) \textbf{First case}: If $\beta(t+T_t)=1$:\\
We have that $w_1(T_{\max}) < w_2(l_2(t+T_t))$.
Then, according to Lemma \ref{lem:inequalities_whittle_index_thresholds}, we have that $w_1(T_{\max})\leq w_1(l_1(t+T_t))$. Hence, we can conclude that $T_{\max} \leq l_1(t+T_t)$ as $w_1$ is an increasing function with the age of information.\\
Moreover, since we have that $p_1\alpha_1(t+T_t-j)+ p_2\alpha_2(t+T_t-j) > p_2 \alpha$ (the strict inequality is due to the fact that $\alpha_1(t)>0$ as $t \geq t_f$ according to Lemma \ref{lem:alpha_greater_stric_0}), then according to Lemma \ref{lem:upper_bound_threshold_max}, we obtain: 
\begin{equation}\sum_{j=1}^{l_1(t+T_t)-1} p_1\alpha_1(t+T_t-j)+ \sum_{j=1}^{l_2(t+T_t)-1} p_2\alpha_2(t+T_t-j)+\beta(t+T_t) p_1\alpha_1(t+T_t-l_1(t+T_t))+ \gamma(t+T_t)p_2\alpha_2(t+T_t-l_2(t+T_t))=1-\alpha\end{equation}
\begin{equation}=\sum_{j=1}^{l_1(t+T_t)} p_1\alpha_1(t+T_t-j)+ \sum_{j=1}^{l_2(t+T_t)-1} p_2\alpha_2(t+T_t-j)+ \gamma(t+T_t)p_2\alpha_2(t+T_t-l_2(t+T_t))
\end{equation}
\begin{equation} \geq \sum_{j=1}^{T_{max}} p_1\alpha_1(t+T_t-j)+ \sum_{j=1}^{T_{max}} p_2\alpha_2(t+T_t-j) > T_{max} p_2 \alpha \geq 1-\alpha
\end{equation}
The last inequality comes from the fact that $ T_{\max} \geq \frac{1-\alpha}{p_2 \alpha}$. This implies that:
\begin{equation}1-\alpha > 1-\alpha
\end{equation}
This gives us an illogical statement. Consequently, in this case, the assumption $l_2(t+T_t) > T_{\max}$ is not true.\\
2) \textbf{Second case}: If $\beta(t+T_t)<1$:\\
As we have that $\beta(t+T_t)<1$, then $\gamma(t+T_t)$ should be equal to $1$. Therefore, all users at state $l_2(t+T_t)$ in class 2 are in the users' proportion $1-\alpha$ with the smallest Whittle index values. However, there exists users in state $l_1(t+T_t)$ in class 1 in the users' proportion $\alpha$ that has the highest Whittle index values. That is, we have $w_1(l_1(t+T_t)) \geq w_2(l_2(t+T_t))$. As it has been established before tackling the first case, $w_1(T_{\max}) < w_2(l_2(t+T_t))$, then $w_1(T_{\max}) < w_1(l_1(t+T_t))$. This means that $l_1(t+T_t) > T_{\max}$.
Therefore, we have that: 
\begin{equation}\sum_{j=1}^{l_1(t+T_t)-1} p_1\alpha_1(t+T_t-j)+ \sum_{j=1}^{l_2(t+T_t)-1} p_2\alpha_2(t+T_t-j)+\beta(t+T_t) p_1\alpha_1(t+T_t-l_1(t+T_t))+ \gamma(t+T_t)p_2\alpha_2(t+T_t-l_2(t+T_t))=1-\alpha
\end{equation}
\begin{equation} \geq \sum_{j=1}^{T_{\max}} p_1\alpha_1(t+T_t-j)+ \sum_{j=1}^{T_{\max}} p_2\alpha_2(t+T_t-j) > T_{\max} p_2 \alpha \geq 1-\alpha
\end{equation}
This implies that:
\begin{equation}1-\alpha > 1-\alpha
\end{equation}
Consequently, in this case, the assumption $l_2(t+T_t) > T_{\max}$ is not true.\\
Hence, in both cases, $l_2(t+T_t)$ must be less than $T_{\max}$, i.e. $\max(l_1(t+T_t),l_2(t+T_t))\leq T_{\max}$ for all $t$.\\ 
Thus, we end up with $T_{\max}=l_{\max}$, which concludes our proof.

\section{Proof of Lemma \ref{lem:inequalities_whittle_index_thresholds}} \label{app:lem:inequalities_whittle_index_thresholds}
We prove only the first statement as the proof steps for both cases are exactly the same. 
By definition of $l_1(t)$ and $l_2(t)$, we have that $\{z_i^1(t)\}_{1 \leq i \leq l_1(t)} \cup \{z_i^2(t)\}_{1 \leq i \leq l_2(t)}$ is exactly the set $\{z_i^k(t): w_k(i) \leq \max(w_1(l_1(t),w_2(l_1(t))\}$. Hence, if a given $q$ verifies $w_1(q) \leq w_2(l_2(t))$, then $w_1(q) \leq \max(w_1(l_1(t),w_2(l_2(t))$, that implies that $z_q^1(t) \in \{z_i^k(t): w_k(i) \leq \max(w_1(l_1(t),w_2(l_2(t))\}=\{z_i^1(t)\}_{1 \leq i \leq l_1(t)} \cup \{z_i^2(t)\}_{1 \leq i \leq l_2(t)}$. Knowing that the highest users' proportion's state of the aforementioned set in class 1 is $l_1(t)$, then $q \leq l_1(t)$. Therefore as $w_1(.)$ is increasing, $w_1(q) \leq w_1(l_1(t))$.

\section{Proof of Proposition \ref{prop:unique_threshold_and_expression}}\label{app:prop:unique_threshold_and_expression}
We have that:
\begin{equation}
\sum_{i=1}^{l_1(t)-1} z^1_i(t)+ \sum_{i=1}^{l_2(t)-1} z^2_i(t)+\beta(t) z^1_{l_1(t)}(t)+\gamma(t) z^2_{l_2(t)}(t)=1-\alpha
\end{equation} 
with $l_1(t)$ and $l_2(t)$ being the thresholds in class $1$ and $2$ respectively at time $t$, and $\beta(t)=1$ and $0<\gamma(t) \leq 1$, or $\gamma(t)=1$ and $0<\beta(t) \leq 1$.\\
Our aim in this proof is to show that there is a link between $l_1(t)$ and $l_2(t)$ when they are less than $T_{\max}$. By doing so, we find a general form of the aforementioned equation.  
To that end, we prove first that $l_1(t)$ is less than $l_2(t)$.\\
Indeed, as we have $w_2(l_1(t)) < w_1( l_1(t))$,
then according to lemma \ref{lem:inequalities_whittle_index_thresholds}, $w_2(l_1(t)) \leq w_2( l_2(t))$. Consequently, we can conclude that $l_1(t) \leq l_2(t)$.\\
Secondly, we prove that $l_2(t) \leq l_1(t)+1$.
As the order of the Whittle indices alternates between the two classes from state $1$ to state $T_{\max}+1$, $w_1(l_2(t)-1) < w_2(l_2(t))$. 
Hence, according to lemma \ref{lem:inequalities_whittle_index_thresholds}, we have that $w_1(l_2(t)-1) \leq w_1(l_1(t))$. Consequently, $l_2(t)-1 \leq l_1(t)$.\\
Given that $l_1(t) \leq l_2(t) \leq l_1(t)+1$, then $l_1(t)$ can be either $l_2(t)$ or $l_2(t)-1$.\\
The second step consists of deriving the value of $\beta(t)$ or $\gamma(t)$ depending on the value of $l_1(t)$ and $l_2(t)$.
\begin{itemize}   
\item If $l_1(t))=l_2(t)$:\\
We prove that $\gamma(t)=1$ if $z_{l_2(t)}^2>0$.
Indeed, if $\gamma(t) \neq 1$ and $z_{l_2(t)}^2>0$, thus there is at least a non empty set of users in class $2$ at state $l_2(t)$ that belongs to the users' proportion $\alpha$ with the highest Whittle index values.
However there exists always a non empty set of queues in class $1$ at state $l_1(t)$ that belong to $1-\alpha$ users' proportion with the least Whittle index values, since $\beta(t)>0$.
Then, we have that $w_2(l_2(t)) \geq w_1(l_1(t))$. However, we know that $w_2(l_2(t))=w_2(l_1(t)) < w_1(l_1(t))$. This later inequality contradicts with what precedes. Thus, the statement that $\gamma(t) \neq 1$ is not true, i.e. $\gamma(t)=1$.\\
In this case we denote $l(t)=l_1(t)=l_2(t)$.\\
We end up:
\begin{equation}
\sum_{j=1}^{l(t)-1}  z^1_{i}(t)+ \sum_{i=1}^{l(t)-1}  z^2_{i}(t)+ \beta(t)z^1_{l(t)}(t)+ z^2_{l(t)}(t)=1-\alpha
\end{equation}
If $z_{l_2(t)}^2=0$, the last equation still valid since $z_{l_2(t)}^2=0$ whatever the value of $\gamma(t)$, namely when $\gamma(t)=1$.  
\item If $l_1(t)+1=l_2(t)$:\\
We prove that $\beta(t)=1$ if $z_{l_1(t)}^1>0$. 
Indeed, if $\beta(t) \neq 1$ and $z_{l_1(t)}^1>0$, there is at least a set of users in class 1 in state $l_1(t)$ that belongs to the users' proportion $\alpha$ with the highest Whittle index values.
However there is always a set of queues in class 2 at state $l_2(t)$ that belong to $1-\alpha$ users' proportion with the least Whittle index values, since $\gamma(t)>0$.
Then, we have that $w_1(l_1(t)) \geq w_2(l_2(t))$. However, we know that $w_2(l_2(t))=w_2(l_1(t)+1) > w_1(l_1(t))$ since the order of Whittle index alternates between the two classes from state $1$ to $T_{\max}+1$ according to Proposition \ref{prop:alternation_Whittle_index}. Thus, $w_2(l_1(t)+1) > w_1(l_1(t)) \geq w_2(l_1(t)+1)$, which gives us an obvious contradiction. Therefore, we can assert that $\beta(t)=1$.\\
In this case, we consider that $l(t)=l_1(t)+1=l_2(t)$ and we get:
\begin{equation}\sum_{i=1}^{l(t)-1}  z^1_{i}(t)+ \sum_{i=1}^{l(t)-1}  z^2_{i}(t)+ \gamma(t) z^2_{l(t)}(t)=1-\alpha\end{equation} 
\end{itemize}
Similarly to the first case, if $z_{l_1(t)}^1=0$, the last equation still valid since $z_{l_1(t)}^1=0$ whatever the value of $\beta(t)$, namely when $\beta(t)=1$.
Subsequently, combining the two cases, there exists $l(t)$ such that:
\begin{equation}\sum_{i=1}^{l(t)-1}  z^1_{i}(t)+ \sum_{i=1}^{l(t)-1}  z^2_{i}(t)+\beta(t) z^1_{l(t)}(t)+ \gamma(t) z^2_{l(t)}(t)=1-\alpha
\end{equation} 
where $\beta(t)=0$ and $0 < \gamma(t) \leq 1$, or $0 < \beta(t) \leq 1$ and $\gamma(t)=1$.

\section{Proof of Proposition \ref{prop:unique_expression_for_all_T}}\label{app:prop:unique_expression_for_all_T}
We prove the Proposition by induction:
\begin{itemize}
\item For $T=T_0$, we have already proved our claim.
\item We suppose that the statement is valid for a given $T$, i.e. there exists $l(T)$, $\beta(T)$ and $\gamma(T)$ such that:
\begin{equation}
\sum_{j=1}^{l(T)-1} p_1\alpha_1(T-j)+ \sum_{j=1}^{l(T)-1} p_2\alpha_2(T-j)+\beta(T) p_1\alpha_1(T-l(T))+ \gamma(T)p_2\alpha_2(T-l(T))=1-\alpha
\end{equation} 
where $\beta(T)=0$ and $0 < \gamma(T) \leq 1$, or $0 < \beta(T) \leq 1$ and $\gamma(T)=1$.
Then, at the next time slot, among the users' proportion scheduled, $\alpha$, exactly $p_1\alpha_1(T)$ and $p_2\alpha_2(T)$ will go to state one for each class, while for the rest, their states will be incremented by one. 
Likewise, for the other users for which the action taken is passive, their states will be incremented by one. As consequence, the decreasing order according to the Whittle index value for these proportions of users at the next slot is $\beta(T)p_1 \alpha_1(T-l(T)),\gamma(T)p_2 \alpha_2(T-l(T)), p_1\alpha_1(T-l(T)+1), p_2\alpha_2(T-l(T)+1), p_1\alpha_1(T-l(T)+2), p_2\alpha_2(T-l(T)+2), p_1\alpha_1(T-l(T)+3),p_2\alpha_2(T-l(T)+3),\cdots, p_1\alpha_1(T),p_2\alpha_2(T)$ (As we have mentioned before, the order of the Whittle indices alternates between the two classes because $l(T)+1 \leq l_{\max}+1$). 
Moreover, the states of the users' proportion $(1-p_1)\alpha_1(t)$ and $(1-p_2)\alpha_2(t)$; which are scheduled but they don't transit to the state $1$ with respect to their classes; will be increased by one. Leveraging the above results, we provide the decreasing order of all users' proportions according to the Whittle index value depending on two cases of $\beta(t)$.\\
If $\beta(T)=0$, then the smallest state's value among the users' proportions $(1-p_1)\alpha_1(t)$ and $(1-p_2)\alpha_2(t)$ at time $T+1$ is $l(T)+1$. Hence, their Whittle index values will be higher than $w_2(l(T)+1)$, and consequently, they will be higher than those of users' proportion of $\gamma(T)p_2 \alpha_2(T-l(T))$ at state $l(T)+1$ in class 2.\\ 
If $\beta(T) \neq 1$, the smallest state value among the users' proportions $(1-p_1)\alpha_1(t)$ and $(1-p_2)\alpha_2(t)$ at time $T+1$ is respectively $l(T)+1$ and $l(T)+2$. Then, their Whittle index values will be higher than $w_1(l(T)+1)$ ($w_1(l(T)+1) < w_2(l(T)+2)$ as the alternation condition is satisfied from $1$ until $l_{\max}+1$). Consequently, their Whittle index values will be higher than the Whittle index of users' proportion $\beta(T)p_1 \alpha_1(T-l(T))$ at state $l(T)+1$ in class 1.\\
Thus, the decreasing order of all users' proportions according to the Whittle index value whatever the value of $\beta(T)$ at $T+1$ is: $(1-p_1)\alpha_1(t),(1-p_2)\alpha_2(t),\beta(T)p_1 \alpha_1(T-l(T)),\gamma(T)p_2 \alpha_2(T-l(T)), p_1\alpha_1(T-l(T)+1), p_2\alpha_2(T-l(T)+1), p_1\alpha_1(T-l(T)+2), p_2\alpha_2(T-l(T)+2), p_1\alpha_1(T-l(T)+3),p_2\alpha_2(T-l(T)+3),\cdots,p_1\alpha_1(T),p_2\alpha_2(T)$.\\
As we have that $(1-p_1)\alpha_1(t)+(1-p_2)\alpha_2(t) \leq \alpha$, then surely the thresholds at time $T+1$ in class 1 and in class 2 are less than the state of the users' proportion $\beta(T)p_1 \alpha_1(T-l(T))$ and $\gamma(T)p_2 \alpha_2(T-l(T))$ respectively. Therefore, there exists $l_1(T+1)$, $l_2(T+1)$, $\beta(T+1)$ and $\gamma(T+1)$ such that $0 < \beta(T+1) \leq 1$ and $\gamma(T+1)=1$, or $\beta(T+1)=1$ and $0 < \gamma(T+1) \leq 1$: 
\begin{equation}\sum_{j=1}^{l_1(T+1)-1} p_1\alpha_1(T+1-j)+ \sum_{j=1}^{l_2(T+1)-1} p_2\alpha_2(T+1-j)+\beta(T+1) p_1\alpha_1(T+1-l_1(T+1))+ \gamma(T)p_2\alpha_2(T+1-l_2(T+1))=1-\alpha\end{equation} 
Now we prove by contradiction that $\max(l_1(T+1),l_2(T+1))\leq T_{max}$.\\
We prove first that $l_2(T+1)$ is greater than $l_1(T+1)$.\\
As $w_2(l_1(T+1)) < w_1(l_1(T+1))$,
that means according to lemma \ref{lem:inequalities_whittle_index_thresholds}, $l_2(T+1)$ is greater than $l_1(T+1)$ ($w_2(l_1(T+1)) < w_2(l_2(T+1))$).\\
Reasoning by contradiction, if $l_2(T+1) > T_{max}$, then we distinguish between two cases:
\begin{itemize}
\item First case: If $\beta(T+1)=1$:\\
we have that $w_1(T_{max})< w_2(l_2(T+1))$ ($w_1(T_{max})< w_2(T_{max}+1)$ as the alternation condition is satisfied in $[1,T_{\max}+1]$),
i.e., according to lemma \ref{lem:inequalities_whittle_index_thresholds}, we have that $l_{max} \leq l_1(T+1)$. Hence, according to lemmas \ref{lem:alpha_greater_stric_0} and \ref{lem:upper_bound_threshold_max}, we have that: 
\begin{equation}\sum_{j=1}^{l_1(T+1)-1} p_1\alpha_1(T+1-j)+ \sum_{j=1}^{l_2(T+1)-1} p_2\alpha_2(T+1-j)+\beta(T+1) p_1\alpha_1(T+1-l_1(T+1))+ \gamma(T+1)p_2\alpha_2(T+1-l_2(T+1))\end{equation}
\begin{equation} 
=1-\alpha \geq \sum_{j=1}^{T_{max}} p_1\alpha_1(T+1-j)+ \sum_{j=1}^{T_{max}} p_2\alpha_2(T+1-j) > T_{max} p_2 \alpha \geq 1-\alpha
\end{equation}
Therefore we end up with:
\begin{equation}
1-\alpha > 1-\alpha
\end{equation}
Hence, the assumption that $l_2(T+1) > T_{max}$ leads us to an illogical statement. Consequently, the hypothesis of $l_2(T+1) > l_{max}$ is not valid for the first case.\\
\item Second case: If $\beta(T+1)<1$:\\
Then we have that $\gamma(T+1)=1$. Therefore, all users at state $l_2(T+1)$ in class 2 are in the proportion $1-\alpha$ with the smallest Whittle index values. However, there are users in state $l_1(T+1)$ in class 1 of the $\alpha$ proportion with the highest Whittle index values. In other words, $w_1(l_1(T+1)) \geq w_2(l_2(T+1)) > w_1(T_{max})$. This means that $l_1(T+1) > T_{max}$.
Therefore, according to lemmas \ref{lem:alpha_greater_stric_0} and \ref{lem:upper_bound_threshold_max}: 
\begin{equation}\sum_{j=1}^{l_1(T+1)-1} p_1\alpha_1(T+1-j)+ \sum_{j=1}^{l_2(T+1)-1} p_2\alpha_2(T+1-j)+\beta(T+1) p_1\alpha_1(T+1-l_1(T+1))+ \gamma(T+1)p_2\alpha_2(T+1-l_2(T+1))
\end{equation}
\begin{equation}
=1-\alpha \geq \sum_{j=1}^{T_{max}} p_1\alpha_1(T+1-j)+ \sum_{j=1}^{T_{max}} p_2\alpha_2(T+1-j) > T_{max} p_2 \alpha \geq 1-\alpha
\end{equation}
\begin{equation}
1-\alpha > 1-\alpha
\end{equation}
Therefore, the hypothesis of $l_2(T+1) > T_{max}$ is not valid for the second case.\\
\end{itemize}
Consequently, we have that $l_2(T+1) \leq T_{max}$, i.e. $\max(l_1(T+1),l_2(T+1))\leq T_{max}$. 
Then, according to Proposition \ref{prop:unique_threshold_and_expression}, there exists $l(T+1)$, and $\gamma(T+1)$ and $\beta(T+1)$ such that:
\begin{equation}\sum_{j=1}^{l(T+1)-1} p_1\alpha_1(T+1-j)+ \sum_{j=1}^{l(T+1)-1} p_2\alpha_2(T+1-j)+\beta(T+1) p_1\alpha_1(T+1-l(T+1))+ \gamma(T+1)p_2\alpha_2(T+1-l(T+1))=1-\alpha\end{equation} 
where $\beta(T+1)=0$ and $0 < \gamma(T+1) \leq 1$, or $0 < \beta(T+1) \leq 1$ and $\gamma(T+1)=1$.
\end{itemize}
To conclude, we have proved  by induction, that for all $T \geq T_0$, there exists $l(T)$, $\beta(T)$ and $\gamma(T)$, such that:
\begin{equation}\sum_{j=1}^{l(T)-1} p_1\alpha_1(T-j)+ \sum_{j=1}^{l(T)-1} p_2\alpha_2(T-j)+\beta(T) p_1\alpha_1(T-l(T))+ \gamma(T)p_2\alpha_2(T-l(T))=1-\alpha
\end{equation} 
where $\beta(T)=0$ and $0 < \gamma(T) \leq 1$, or $0 < \beta(T) \leq 1$ and $\gamma(T)=1$, which concludes our proof.\\  

\section{Proof of Proposition \ref{prop:elements_inclusion}}\label{app:prop:elements_inclusion}
We proceed by the same method used to prove the Proposition \ref{prop:unique_expression_for_all_T}.\\
We consider at time $T$:
\begin{equation}\sum_{j=1}^{l(T)-1} p_1\alpha_1(T-j)+ \sum_{j=1}^{l(T)-1} p_2\alpha_2(T-j)+\beta(T) p_1\alpha_1(T-l(T))+\gamma(T)p_2\alpha_2(T-l(T))=1-\alpha,
\end{equation} 
where $\beta(T)=0$ and $0 \leq \gamma(T) < 1$, or $0 \leq \beta(T) < 1$ and $\gamma(T)=1$.
Among the users' proportion scheduled $\alpha$, exactly $p_1\alpha_1(T)$ and $p_2\alpha_2(T)$ will go to state one for each classes, and  $(1-p_1)\alpha_1(T)$ and $(1-p_2)\alpha_2(T)$ will go to the next state.\\ 
For the other users for which the action taken is passive, their states will be increased by one, then the decreasing order according to the Whittle index value at the next time slot is $\beta(T) p_1\alpha_1(T-l(T)), \gamma(T) p_2\alpha_2(T-l(T)), p_1\alpha_1(T-l(T)+1), p_2\alpha_2(T-l(T)+1), p_1\alpha_1(T-l(T)+2),p_2\alpha_2(T-l(T)+2)\cdots p_1 \alpha_1(T),p_2 \alpha_2(T)$ (As we said before that the order based on the value of the Whittle indices, alternate between the two classes from state $1$ to $l(T) \leq l_{\max}+1$). Moreover, the users' proportion scheduled $(1-p_1)\alpha_1(T)$ and $(1-p_2)\alpha_2(T)$ will be at states that have Whittle index values higher than those of $\beta(T) p_1\alpha_1(T-l(T))$ and  $\gamma(T) p_2\alpha_2(T-l(T))$ (as we have explained in the proof of Proposition \ref{prop:unique_expression_for_all_T}).\\
Hence, the global decreasing order according to the Whittle index value is $(1-p_1)\alpha_1(T),(1-p_1)\alpha_2(T),\beta(T) p_1\alpha_1(T-l(T)), \gamma(T) p_2\alpha_2(T-l(T)), p_1\alpha_1(T-l(T)+1), p_2\alpha_2(T-l(T)+1), p_1\alpha_1(T-l(T)+2),p_2\alpha_2(T-l(T)+2)\cdots p_1 \alpha_1(T),p_2 \alpha_2(T)$.\\
Providing that $(1-p_1)\alpha_1(t)+(1-p_2)\alpha_2(t) \leq \alpha$, then at time $T+1$:
\begin{equation}\sum_{j=1}^{l(T)-1} p_1\alpha_1(T-j)+ \sum_{j=1}^{l(T)-1} p_2\alpha_2(T-j)+\beta(T) p_1\alpha_1(T-l(T))+\gamma(T)p_2\alpha_2(T-l(T))+p_1\alpha_1(T)+p_2\alpha_2(T) \geq 1-\alpha
\end{equation}
Then, there exists $\beta=0$ and $0 < \gamma \leq 1$, or $0 < \beta \leq 1$ and $\gamma=1$, and sub-set $\{\alpha_1(T),\alpha_2(T),\alpha_1(T-1),\alpha_2(T-1)\cdots \alpha_1(T-m),\alpha_2(T-m)\} \subset \{\alpha_1(T-l(T)),\alpha_2(T-l(T)),\alpha_1(T-l(T)+1),\alpha_2(T-l(T)+1), \alpha_1(T-l(T)+2),\alpha_2(T-l(T)+2)\cdots \alpha_1(T), \alpha_2(T)\}$, such that: 
\begin{equation}\sum_{j=1}^{(m+1)-1} p_1\alpha_1(T+1-j)+ \sum_{j=1}^{(m+1)-1} p_2\alpha_2(T+1-j)+\beta p_1\alpha_1(T+1-(m+1))+\gamma p_2\alpha_2(T+1-(m+1))=1-\alpha
\end{equation}
Indeed, $m+1$ is effectively $l(T+1)$, $\beta=\beta(T+1)$, $\gamma=\gamma(T+1)$, and the elements of the set $\{\alpha_1(T), \alpha_1(T-1), \cdots \alpha_1(T-m)\} \cup \{\alpha_1(T+1)\}$ and the set $\{\alpha_2(T), \alpha_2(T-1), \cdots, \alpha_2(T-m) \} \cup \{\alpha_2(T+1)\}$ are exactly the elements of the vectors $A_1(T+1)$ and $A_2(T+1)$ respectively. Given that $\{\alpha_1(T), \alpha_1(T-1), \cdots \alpha_1(T-m)\}$ and $\{\alpha_2(T), \alpha_2(T-1), \cdots, \alpha_2(T-m)\}$ are included in the set of elements of the vector $A_1(T)$ and $A_2(T)$ respectively, then for $k=1,2$, all the elements of the vector $A_k(T+1)$ except $\alpha_k(T+1)$ belong to the elements of vector $A_k(T)$.

\section{Proof of Proposition \ref{prop:inequality_satisfied_by_first_element}}\label{app:prop:inequality_satisfied_by_first_element}
According to Proposition \ref{prop:unique_expression_for_all_T}, the elements of the vectors $A_1(T)$ and $A_2(T)$ satisfy: 
\begin{equation}\sum_{j=1}^{l(T)-1} p_1\alpha_1(T-j)+ \sum_{j=1}^{l(T)-1} p_2\alpha_2(T-j)+\beta(T) p_1\alpha_1(T-l(T))+ \gamma(T) p_2\alpha_2(T-l(T))=1-\alpha,
\end{equation}
where $0 < \beta(T) \leq 1$ and $\gamma(T)=1$, or $\beta(T)=0$ and $0 < \gamma(T) \leq 1$. 
We distinguish between two cases depending on the values of $\beta$ and $\gamma$ (we drop the index $T$ on $\beta(T)$ and $\gamma(T)$ to ease the notation):\\
\begin{itemize}

\item First case: $0 < \beta \leq 1$, and $\gamma=1$:\\
Hence:
\begin{equation}
\sum_{j=1}^{l(T)-1} p_1\alpha_1(T-j)+ \sum_{j=1}^{l(T)-1} p_2\alpha_2(T-j)+\beta(T) p_1\alpha_1(T-l(T))+p_2\alpha_2(T-l(T))=1-\alpha
\end{equation} 
Our aim is to derive the expression of $\alpha_k(T+1)$ for class 1 and class 2.
Among the users' proportion scheduled $\alpha$, exactly $p_1\alpha_1(T)$ and $p_2\alpha_2(T)$ will go to state one for each class, and the rest will go to the next state. Hence: 
\begin{equation}
\alpha_1(T+1)=(1-p_1)\alpha_1(T)+B_1(T)
\end{equation}
\begin{equation}
\alpha_2(T+1)=(1-p_2)\alpha_2(T)+B_2(T)
\end{equation}
such that $B_1(T)+B_2(T)=p_1\alpha_1(T)+p_2\alpha_2(T)$.\\
At time $T+1$, the decreasing order according to the Whittle index value is $(1-p_1)\alpha_1(T),(1-p_2)\alpha_2(T),\beta p_1\alpha_1(T-l(T)), p_2\alpha_2(T-l(T)), p_1\alpha_1(T-l(T)+1), p_2\alpha_2(T-l(T)+1), p_1\alpha_1(T-l(T)+2),p_2\alpha_2(T-l(T)+2),\cdots,p_1\alpha_1(T),p_2\alpha_2(T)$.\\
In order to get $B_1(T)$ and $B_2(T)$, we sum the users' proportions at different states starting from the users' proportion $\beta p_1\alpha_1(T-l(T))$ following the decreasing order of the Whittle index until we get the sum that equals to $p_1\alpha_1(T)+p_2\alpha_2(T)$. We distinguish between six sub-cases and for each sub-case, we prove that $\alpha_k(T+1)$ is surely between two elements of the vector $A_k(T)$. In fact, if we prove it just for one class, the result will be true for the other one, since $\alpha_1(T)+\alpha_2(T)=\alpha$ for all $T$. In the following, we derive the expression of $\alpha_k(T+1)$ for $k=1,2$, in function of the elements of the vector $A_1(T)$ and $A_2(T)$ and we show that $\alpha_1(T+1)$ is surely between two elements of the vector $A_1(T)$.\\ 
1) If $p_1\alpha_1(T)+p_2\alpha_2(T) \leq p_1\beta \alpha_1(T-l(T))$:\\
In this case $p_1\alpha_1(T)+p_2\alpha_2(T)$ is less than $ p_1\beta \alpha_1(T-l(T))$. Therefore, we will take a proportion of users from $ p_1\beta \alpha_1(T-l(T))$ that equals to $p_1\alpha_1(T)+p_2\alpha_2(T)$ denoted by $C$. This users' proportion exactly equals to $B_1(T)+B_2(T)$ that we add to $(1-p_1)\alpha_1(T)$ and $(1-p_2)\alpha_2(T)$. Thus, $B_1(T)+B_2(T)=C$.
However, since all the users of the proportion $C$ belong to $p_1\beta \alpha_1(T-l(T))$, then $C$ contains only the users of the class 1. Consequently, $B_1(T)=C$ and $B_2(T)=0$.  
Hence:
\begin{equation}
\alpha_2(T+1)=(1-p_2) \alpha_2(T)
\end{equation}
As $\alpha_1(T+1)+\alpha(T+1)=\alpha$, then: 
\begin{equation}
\alpha_1(T+1)=\alpha-\alpha_2(T+1)
\end{equation}
Now we find the upper bound of $\alpha_2(T)-\alpha_2(T+1)$: 
\begin{align}
\alpha_2(T)-\alpha_2(T+1)=&p_2 \alpha(T)\\
\leq& \beta \alpha_1(T-l(T))p_1-\alpha_1(T)p_1\\
\leq &p_1(\alpha_1(T-l(T))-\alpha_1(T))\\
=&p_1(\alpha_2(T)-\alpha_2(T-l(T)))
\end{align}
The first inequality comes from the fact that $p_1\alpha_1(T)+p_2\alpha_2(T) \leq p_1\beta \alpha_1(T-l(T))$ and the second one comes from the fact that $\beta \leq 1$.\\
Given that $\alpha_2(i)-\alpha_2(j)=\alpha_1(j)-\alpha_1(i)$ for all integers $i$ and $j$, thus:\\\begin{equation}
\alpha_1(T+1)-\alpha_1(T) \leq p_1(\alpha_1(T-l(T))-\alpha_1(T))
\end{equation}
Moreover, we have that $\alpha_1(T+1)-\alpha_1(T) \geq 0$ because $\alpha_2(T+1)-\alpha_2(T) \leq 0$. Therefore, $\alpha_1(T) \leq \alpha_1(T+1)$. On the other hands, as $p_1(\alpha_1(T-l(T))-\alpha_1(T)) \geq \alpha_1(T+1)-\alpha_1(T)\geq 0$ then $\alpha_1(T-l(T))-\alpha_1(T) \geq \alpha_1(T+1)-\alpha_1(T)$. This means that $\alpha_1(T-l(T)) \geq \alpha_1(T+1)$. Consequently, we end up with:
\begin{equation}
\alpha_1(T) \leq \alpha_1(T+1) \leq \alpha_1(T-l(T))
\end{equation}\\
 
2) If $\beta \alpha_1(T-l(T))p_1 \leq p_1\alpha_1(T)+p_2\alpha_2(T) \leq \beta \alpha_1(T-l(T))p_1+\alpha_2(T-l(T))p_2$:\\ 
Hence: 
\begin{align}
\alpha_1(T+1)=&(1-p_1) \alpha_1(T)+\beta p_1 \alpha_1(T-l(T))\\
\alpha_2(T+1)=&\alpha-\alpha_1(T+1)
\end{align}
Then:
\begin{align}
\alpha_1(T+1)-\alpha_1(T)=&\beta p_1 \alpha_1(T-l(T)-p_1 \alpha_1(T)\\
\leq & p_1 (\alpha_1(T-l(T)-\alpha_1(T))
\end{align}
On the other hand, we have according to the right inequality of sub-case's assumption:
\begin{align}
\alpha_1(T+1)-\alpha_1(T)=&\beta p_1 \alpha_1(T-l(T)-p_1 \alpha_1(T)\\
\geq & p_2 \alpha_2(T)-p_2 \alpha_2(T-l(T))\\
=& p_2 (\alpha_1(T-l(T))- \alpha_1(T))
\end{align}
Hence :
\begin{align}
p_2(\alpha_1(T-l(T))-\alpha_1(T))\leq \alpha_1(T+1)-\alpha_1(T) \leq p_1(\alpha_1(T-l(T))-\alpha_1(T))
\end{align}
Knowing that $p_2<p_1$, the later inequalities imply that  $\alpha_1(T-l(T))-\alpha_1(T) \geq 0$.\\
As a result we have that: 
\begin{equation}
\alpha_1(T) \leq \alpha_1(T+1) \leq \alpha_1(T-l(T))
\end{equation}

And \begin{equation}\alpha_1(T+1)-\alpha_1(T) \leq p_1(\alpha_1(T-l(T))-\alpha_1(T))
\end{equation}

3) If $\beta \alpha_1(T-l(T))p_1+\alpha_2(T-l(T))p_2 \leq p_1\alpha_1(T)+p_2\alpha_2(T) \leq \beta \alpha_1(T-l(T))p_1+\alpha_2(T-l(T))p_2+p_1\alpha_1(T-l(T)+1)$:\\
Hence: 
\begin{align}
\alpha_2(T+1)=&(1-p_2) \alpha_2(T)+p_2\alpha_2(T-l(T))\\
\alpha_1(T+1)=&\alpha-\alpha_2(T+1)
\end{align}
Therefore:
\begin{equation}\alpha_2(T+1)-\alpha_2(T)= p_2(\alpha_2(T-l(T))-\alpha_2(T))
\end{equation}
And:
\begin{equation}\alpha_1(T)-\alpha_1(T+1)= p_2(\alpha_1(T)-\alpha_1(T-l(T)))
\end{equation}
This means that if $\alpha_1(T) \leq \alpha_1(T+1)$:\\
\begin{equation}\alpha_1(T) \leq \alpha_1(T+1) \leq \alpha_1(T-l(T))\end{equation}
And 
\begin{equation}\alpha_1(T+1)-\alpha_1(T) \leq p_1(\alpha_1(T-l(T))-\alpha_1(T))\end{equation}
If $\alpha_1(T+1) \leq \alpha_1(T)$:\\
\begin{equation}\alpha_1(T-l(T)) \leq \alpha_1(T+1) \leq \alpha_1(T)\end{equation}
And 
\begin{equation}\alpha_1(T)-\alpha_1(T+1) \leq p_1(\alpha_1(T)-\alpha_1(T-l(T)))\end{equation} 

4) If $\beta \alpha_1(T-l(T))p_1+\alpha_2(T-l(T))p_2+p_1\alpha_1(T-l(T)+1) \leq p_1\alpha_1(T)+p_2\alpha_2(T) \leq \beta \alpha_1(T-l(T))p_1+\alpha_2(T-l(T))p_2+p_1\alpha_1(T-l(T)+1)+p_2 \alpha_2(T-l(T)+1):$\\
Hence:
\begin{align}
\alpha_1(T+1)=&(1-p_1) \alpha_1(T)+p_1\beta \alpha_1(T-l(T))+p_1\alpha_1(T-l(T)+1)\\
\alpha_2(T+1)=&\alpha-\alpha_1(T+1)
\end{align}
Therefore: 
\begin{align}
\alpha_1(T+1)-\alpha_1(T)=&-p_1 \alpha_1(T)+p_1\beta \alpha_1(T-l(T))+p_1\alpha_1(T-l(T)+1)\\
\end{align}
According to the left inequality of the assumption of this case, we have that:
\begin{align}
\alpha_1(T+1)-\alpha_1(T) \leq &p_2 \alpha_2(T)-p_2\alpha_2(T-l(T))\\
=&p_2(\alpha_1(T-l(T))- \alpha_1(T))
\end{align}
On the other hand, we have that:
\begin{align}
\alpha_1(T+1)-\alpha_1(T)=&-p_1 \alpha_1(T)+p_1\beta \alpha_1(T-l(T))+p_1\alpha_1(T-l(T)+1)\\
\geq& p_1 (\alpha_1(T-l(T)+1)+\alpha_1(T)\\
\end{align}
Hence:
\begin{equation}
p_1(\alpha_1(T-l(T)+1)-\alpha_1(T))\leq \alpha_1(T+1)-\alpha_1(T) \leq p_2(\alpha_1(T-l(T))-\alpha_1(T))\end{equation}
Thus:\\
If $\alpha_1(T) \leq \alpha_1(T+1)$:\\
\begin{equation}\alpha_1(T) \leq \alpha_1(T+1) \leq \alpha_1(T-l(T))\end{equation}
And 
\begin{align}
\alpha_1(T+1)-\alpha_1(T) \leq & p_2(\alpha_1(T-l(T))-\alpha_1(T))\\
\leq & p_1(\alpha_1(T-l(T))-\alpha_1(T))
\end{align}
If $\alpha_1(T+1) \leq \alpha_1(T)$:\\
\begin{equation}\alpha_1(T-l(T)+1) \leq \alpha_1(T+1) \leq \alpha_1(T)\end{equation}
And 
\begin{align}\alpha_1(T)-\alpha_1(T+1) \leq& p_1(\alpha_1(T)-\alpha_1(T-l(T)+1))
\end{align} 

5) If there exists $m \geq 1$ such that:\\
$\beta \alpha_1(T-l(T))p_1+\alpha_2(T-l(T))p_2+p_1\alpha_1(T-l(T)+1)+\cdots+p_1\alpha_1(T-l(T)+m)+p_2\alpha_2(T-l(T)+m) \leq p_1\alpha_1(T)+p_2\alpha_2(T) \leq \beta \alpha_1(T-l(T))p_1+\alpha_2(T-l(T))p_2+p_1\alpha_1(T-l(T)+1)+\cdots+p_1\alpha_1(T-l(T)+m)+p_2\alpha_2(T-l(T)+m)+p_1 \alpha_1(T-l(T)+m+1):$\\
This means that:
\begin{align}
\alpha_2(T+1)=&(1-p_2) \alpha_2(T)+p_2 \alpha_2(T-l(T))+\cdots+ p_2\alpha_2(T-l(T)+m)\\
\alpha_1(T+1)=&\alpha-\alpha_2(T+1)
\end{align}
We have that:
\begin{align}
\alpha_2(T+1)-\alpha_2(T)&=-p_2\alpha_2(T)+p_2 \alpha_2(T-l(T))+p_2 \alpha_1(T-l(T)+1)+\cdots+ p_2\alpha_2(T-l(T)+m)\\
&\geq p_2(\alpha_2(T-l(T)+1)-\alpha_2(T))
\end{align}
On the other hand:
\begin{align}
\alpha_2(T+1)-\alpha_2(T)&=-p_2\alpha_2(T)+p_2\alpha_2(T-l(T))+p_2 \alpha_2(T-l(T)+1)+\cdots+ p_2\alpha_2(T-l(T)+m)\\
&\leq p_1 \alpha_1(T)-\beta p_1 \alpha_1(T-l(T))-\sum_{i=1}^{m} p_1 \alpha_1(T-l(T)+i)\\
&\leq  p_1(\alpha_1(T)-\alpha_1(T-l(T)+1))\\
&= p_1(\alpha_2(T-l(T)+1)-\alpha_2(T)
\end{align}
Thus:
\begin{equation}
p_2(\alpha_1(T)-\alpha_1(T-l(T)+1)\leq \alpha_1(T)-\alpha_1(T+1) \leq p_1(\alpha_1(T)-\alpha_1(T-l(T)+1))
\end{equation}
Therefore:
\begin{equation}\alpha_1(T-l(T)+1) \leq \alpha_1(T+1) \leq \alpha_1(T)
\end{equation}
And: 
\begin{equation}\alpha_1(T)-\alpha_1(T+1) \leq p_1(\alpha_1(T)-\alpha_1(T-l(T)+1))
\end{equation}

6) If there exists $m \geq 1$ such that:\\
$\beta \alpha_1(T-l(T))p_1+\alpha_2(T-l(T))p_2+p_1\alpha_1(T-l(T)+1)+
p_2\alpha_2(T-l(T)+1)+\cdots+p_1\alpha_1(T-l(T)+m)+p_2\alpha_2(T-l(T)+m)+p_1\alpha_1(T-l(T)+m+1) \leq p_1\alpha_1(T)+p_2\alpha_2(T) \leq \beta \alpha_1(T-l(T))p_1+\alpha_2(T-l(T))p_2+p_1\alpha_1(T-l(T)+1)+\cdots+p_1\alpha_1(T-l(T)+m)+p_2\alpha_2(T-l(T)+m)+p_1\alpha_1(T-l(T)+m+1)+p_2 \alpha_2(T-l(T)+m+1):$\\
Hence:
\begin{align}
\alpha_1(T+1)=&(1-p_1) \alpha_1(T)+p_1\beta \alpha_1(T-l(T))+ \cdots+ p_1\alpha_1(T-l(T)+m+1)\\
\alpha_2(T+1)=&\alpha-\alpha_1(T+1)
\end{align}
We have that:
\begin{align}
\alpha_1(T+1)-\alpha_1(T)&=-p_1\alpha_1(T)+p_1\beta \alpha_1(T-l(T))+p_1 \alpha_1(T-l(T)+1)+\cdots+ p_1\alpha_1(T-l(T)+m+1)\\
&\geq p_1(\alpha_1(T-l(T)+1)-\alpha_1(T))
\end{align}
On the other hand:
\begin{align}
\alpha_1(T+1)-\alpha_1(T)&=-p_1\alpha_1(T)+p_1\beta \alpha_1(T-l(T))+p_1 \alpha_1(T-l(T)+1)+\cdots+ p_1\alpha_1(T-l(T)+m+1)\\
&\leq p_2 \alpha_2(T)-\sum_{i=0}^{m} p_2 \alpha_2(T-l(T)+i)\\
&\leq  p_2(\alpha_2(T)-\alpha_2(T-l(T)+1))\\
&= p_2(\alpha_1(T-l(T)+1)-\alpha_1(T))
\end{align}
Thus:
\begin{equation}p_1(\alpha_1(T-l(T)+1)-\alpha_1(T))\leq \alpha_1(T+1)-\alpha_1(T) \leq p_2(\alpha_1(T-l(T)+1)-\alpha_1(T))\end{equation}
Therefore:
\begin{equation}\alpha_1(T-l(T)+1) \leq \alpha_1(T+1) \leq \alpha_1(T)\end{equation}
And: 
\begin{equation}\alpha_1(T)-\alpha_1(T+1) \leq p_1(\alpha_1(T)-\alpha_1(T-l(T)+1))\end{equation}

\item Second case: $\beta=0$ and $0<\gamma \leq 1$:\\
Hence, we have that:\\ 
\begin{equation}
\sum_{j=1}^{l(T)-1} p_1\alpha_1(T-j)+ \sum_{j=1}^{l(T)-1} p_2\alpha_2(T-j)+\gamma p_2\alpha_2(T-l(T))=1-\alpha
\end{equation}
Then, at time $T+1$, the decreasing order according to the Whittle index value is $(1-p_1)\alpha_1(T),(1-p_2)\alpha_2(T),\gamma p_2\alpha_2(T-l(T)), p_1\alpha_1(T-l(T)+1), p_2\alpha_2(T-l(T)+1), p_1\alpha_1(T-l(T)+2),p_2\alpha_2(T-l(T)+2),\cdots,p_1\alpha_1(T),p_2\alpha_2(T)$.
In order to obtain $B_1(T)$ and $B_2(T)$, we sum the users' proportions at different states starting from the users' proportion $\gamma p_2\alpha_2(T-l(T))$ following the decreasing order of the Whittle index until we get the sum that equals to $p_1\alpha_1(T)+p_2\alpha_2(T)$.
For this case, we distinguish between five sub-cases, and for each sub-case, we prove that $\alpha_1(T+1)$ is surely between two elements of the vector $A_1(T)$.\\

1) If $ p_1\alpha_1(T)+p_2\alpha_2(T)\leq \gamma \alpha_2(T-l(T))p_2$:\\
Hence: \begin{equation}
\alpha_1(T+1)=(1-p_1) \alpha_1(T)
\end{equation}
\begin{equation}
\alpha_2(T+1)=\alpha-\alpha_1(T+1)
\end{equation}
We have that:
\begin{align}
\alpha_1(T)-\alpha_1(T+1)=&p_1 \alpha_1(T)\\
\leq& \gamma\alpha_2(T-l(T))p_2-\alpha_2(T)p_2\\
\leq& p_2(\alpha_2(T-l(T))-\alpha_2(T))\\
=&p_2(\alpha_1(T)-\alpha_1(T-l(T)))
\end{align}
Thus:\begin{equation}\alpha_1(T)-\alpha_1(T+1) \leq p_1(\alpha_1(T)-\alpha_1(T-l(T))\end{equation} 
And: \begin{equation}\alpha_1(T-l(T)) \leq \alpha_1(T+1) \leq \alpha_1(T)\end{equation}\\

2) If $\gamma\alpha_2(T-l(T))p_2 \leq p_1\alpha_1(T)+p_2\alpha_2(T) \leq \gamma\alpha_2(T-l(T))p_2+\alpha_1(T-l(T)+1)p_1$\\ 
Consequently:
\begin{align}
\alpha_2(T+1)=&(1-p_2) \alpha_2(T)+\gamma p_2 \alpha_2(T-l(T))\\
\alpha_1(T+1)=&\alpha-\alpha_2(T+1)
\end{align}
Hence:
\begin{align}
\alpha_2(T+1)-\alpha_2(T)=&-p_2 \alpha_2(T)+\gamma p_2 \alpha_2(T-l(T))\\
\leq&p_2(\alpha_2(T-l(T))-\alpha_2(T))\\
\end{align}
On the other hand, according to the right inequality of the assumption of this case, we have that:
\begin{align}
\alpha_2(T+1)-\alpha_2(T)=&-p_2 \alpha_2(T)+\gamma p_2 \alpha_2(T-l(T))\\
\geq&p_1(\alpha_1(T)-\alpha_1(T-l(T)+1))\\
=&p_1(\alpha_2(T-l(T)+1)-\alpha_2(T))
\end{align}
That means:
\begin{equation}p_1(\alpha_2(T-l(T)+1)-\alpha_2(T))\leq \alpha_2(T+1)-\alpha_2(T) \leq p_2(\alpha_2(T-l(T))-\alpha_2(T))\end{equation}
i.e.\begin{equation}p_1(\alpha_1(T)-\alpha_1(T-l(T)+1))\leq \alpha_1(T)-\alpha_1(T+1) \leq p_2(\alpha_1(T)-\alpha_1(T-l(T)))\end{equation}
Therefore:\\
If $\alpha_1(T) \leq \alpha_1(T+1)$:\\
\begin{equation}\alpha_1(T) \leq \alpha_1(T+1) \leq \alpha_1(T-l(T)+1)\end{equation}
And: 
\begin{equation}\alpha_1(T+1)-\alpha_1(T) \leq p_1(\alpha_1(T-l(T)+1)-\alpha_1(T))\end{equation}
If $\alpha_1(T+1) \leq \alpha_1(T)$:\\
\begin{equation}\alpha_1(T-l(T)) \leq \alpha_1(T+1) \leq \alpha_1(T)\end{equation}
And: 
\begin{equation}\alpha_1(T)-\alpha_1(T+1) \leq p_1(\alpha_1(T)-\alpha_1(T-l(T)))\end{equation}

3) If $\gamma\alpha_2(T-l(T))p_2+\alpha_1(T-l(T)+1)p_1 \leq p_1\alpha_1(T)+p_2\alpha_2(T) \leq \gamma\alpha_2(T-l(T))p_2+\alpha_1(T-l(T)+1)p_1+p_2\alpha_2(T-l(T)+1)$.\\
Hence:
\begin{align}
\alpha_1(T+1)=&(1-p_1) \alpha_1(T)+p_1\alpha_1(T-l(T)+1)\\
\alpha_2(T+1)=&\alpha-\alpha_1(T+1)
\end{align}
We have that:
\begin{equation}\alpha_1(T+1)-\alpha_1(T)= p_1(\alpha_1(T-l(T)+1)-\alpha_1(T)\end{equation}
If $\alpha_1(T) \leq \alpha_1(T+1)$:\\
\begin{equation}\alpha_1(T) \leq \alpha_1(T+1) \leq \alpha_1(T-l(T)+1)\end{equation}
And:
\begin{equation}\alpha_1(T+1)-\alpha_1(T) \leq p_1(\alpha_1(T-l(T)+1)-\alpha_1(T))\end{equation}
If $\alpha_1(T+1) \leq \alpha_1(T)$:\\
\begin{equation}\alpha_1(T-l(T)+1) \leq \alpha_1(T+1) \leq \alpha_1(T)\end{equation}
And:
\begin{equation}\alpha_1(T)-\alpha_1(T+1) \leq p_1(\alpha_1(T)-\alpha_1(T-l(T)+1))
\end{equation} 

4) If there exists $m \geq 1$ such that:\\
$\gamma\alpha_2(T-l(T))p_2+\cdots+\alpha_1(T-l(T)+m)p_1+p_2\alpha_2(T-l(T)+m) \leq p_1\alpha_1(T)+p_2\alpha_2(T) \leq \gamma\alpha_2(T-l(T))p_2+\cdots+\alpha_1(T-l(T)+m)p_1+p_2\alpha_2(T-l(T)+m)+p_1 \alpha_1(T-l(T)+m+1)$:\\
Hence:
\begin{align}
\alpha_2(T+1)=&(1-p_2) \alpha_2(T)+p_2\gamma\alpha_2(T-l(T))+\cdots+p_2\alpha_2(T-l(T)+m)\\
\alpha_1(T+1)=&\alpha-\alpha_2(T+1)
\end{align}
\begin{align}
\alpha_2(T+1)-\alpha_2(T)&=-p_2\alpha_2(T)+p_2\gamma \alpha_2(T-l(T))+p_2 \alpha_2(T-l(T)+1)+\cdots+ p_2\alpha_2(T-l(T)+m)\\
&\geq p_2(\alpha_2(T-l(T)+1)-\alpha_2(T))
\end{align}
On the other hand:
\begin{align}
\alpha_2(T+1)-\alpha_2(T)&=-p_2\alpha_2(T)+p_2\gamma \alpha_2(T-l(T))+p_2 \alpha_2(T-l(T)+1)+\cdots+ p_2\alpha_2(T-l(T)+m)\\
&\leq p_1 \alpha_1(T)-\sum_{i=1}^{m} p_1 \alpha_1(T-l(T)+i)\\
&\leq  p_1(\alpha_1(T)-\alpha_1(T-l(T)+1))\\
&= p_1(\alpha_2(T-l(T)+1)-\alpha_2(T))
\end{align}
Thus:
\begin{equation}p_2(\alpha_1(T)-\alpha_1(T-l(T)+1)\leq \alpha_1(T)-\alpha_1(T+1) \leq p_1(\alpha_1(T)-\alpha_1(T-l(T)+1))\end{equation}
Therefore:
\begin{equation}\alpha_1(T-l(T)+1) \leq \alpha_1(T+1) \leq \alpha_1(T)\end{equation}
And:
\begin{equation}\alpha_1(T)-\alpha_1(T+1) \leq p_1(\alpha_1(T)-\alpha_1(T-l(T)+1))\end{equation}

5) If there exists $m \geq 1$ such that:\\
$\gamma\alpha_2(T-l(T))p_2+\cdots+\alpha_1(T-l(T)+m)p_1+p_2\alpha_2(T-l(T)+m)+p_1\alpha_1(T-l(T)+m+1) \leq p_1\alpha_1(T)+p_2\alpha_2(T) \leq \gamma\alpha_2(T-l(T))p_2+\cdots+\alpha_1(T-l(T)+m)p_1+p_2\alpha_2(T-l(T)+m)+p_1 \alpha_1(T-l(T)+m+1)+p_2\alpha_2(T-l(T)+m+1)$:\\
That implies that:
\begin{align}
\alpha_1(T+1)=&(1-p_1) \alpha_1(T)+\cdots+p_1\alpha_1(T-l(T)+m)+p_1\alpha_1(T-l(T)+m+1)\\
\alpha_2(T+1)=&\alpha-\alpha_1(T+1)
\end{align}
\begin{align}
\alpha_1(T+1)-\alpha_1(T)&=-p_1\alpha_1(T)+p_1 \alpha_1(T-l(T)+1)+\cdots+ p_1\alpha_1(T-l(T)+m+1)\\
&\geq p_1(\alpha_1(T-l(T)+1)-\alpha_1(T))
\end{align}
On the other hand:
\begin{align}
\alpha_1(T+1)-\alpha_1(T)&=-p_1\alpha_1(T)+p_1 \alpha_1(T-l(T)+1)+\cdots+ p_1\alpha_1(T-l(T)+m+1)\\
&\leq p_2 \alpha_2(T)-\gamma p_2 \alpha_2(T-l(T))-\sum_{i=1}^{m} p_2 \alpha_2(T-l(T)+i)\\
&\leq  p_2(\alpha_2(T)-\alpha_2(T-l(T)+1))\\
&= p_2(\alpha_1(T-l(T)+1)-\alpha_1(T))
\end{align}
Thus:
\begin{equation}p_1(\alpha_1(T-l(T)+1)-\alpha_1(T)\leq \alpha_1(T+1)-\alpha_1(T) \leq p_2(\alpha_1(T-l(T)+1)-\alpha_1(T))\end{equation}
Therefore:
\begin{equation}
\alpha_1(T-l(T)+1) \leq \alpha_1(T+1) \leq \alpha_1(T)
\end{equation}
And: 
\begin{equation}
\alpha_1(T)-\alpha_1(T+1) \leq p_1(\alpha_1(T)-\alpha_1(T-l(T)+1))
\end{equation}
\end{itemize}

In conclusion, all these six sub-cases when $\gamma=1$ and $0<\beta \leq 1$, plus the five sub-cases when $\beta=0$ and $0<\gamma\leq 1$, can be summarized in four cases:\\

1) $\alpha_1(T) \leq \alpha_1(T+1) \leq \alpha_1(T-l(T))$, and $\alpha_1(T+1)-\alpha_1(T) \leq p_1(\alpha_1(T-l(T))-\alpha_1(T))$.\\

2) $\alpha_1(T-l(T)) \leq \alpha_1(T+1) \leq \alpha_1(T)$, and $\alpha_1(T)-\alpha_1(T+1) \leq p_1(\alpha_1(T)-\alpha_1(T-l(T)))$.\\

3) $\alpha_1(T-l(T)+1) \leq \alpha_1(T+1) \leq \alpha_1(T)$, and $\alpha_1(T)-\alpha_1(T+1) \leq p_1(\alpha_1(T)-\alpha_1(T-l(T)+1))$.\\

4) $\alpha_1(T) \leq \alpha_1(T+1) \leq \alpha_1(T-l(T)+1)$, and $\alpha_1(T+1)-\alpha_1(T) \leq p_1(\alpha_1(T-l(T)+1)-\alpha_1(T))$.\\

Thus, the proof is concluded.

\section{Proof of Proposition \ref{prop:max_alpha_less_l_2}}\label{app:prop:max_alpha_less_l_2}
In all the proof, we consider that $\epsilon \leq (l_2-l_1)\frac{(1-p_1)^L}{1-(1-p_1)^L}$.\\
Before tackling the proof, we give a brief insight about the procedure adopted to establish the desired result: 
We start by finding a given time denoted $T_2 \geq T_{\epsilon}$ where $\alpha_1(T_2)$ is less than $l_1$. Then, we show that $\alpha_1(T_2),\cdots,\alpha_1(T_2+L)$ are strictly less than $l_2$. To that end, we start first by defining a relevant sequence $u_n$ in function of $\epsilon$, $l_1$, $l_2$ and $p_1$ when $n \in [0,L]$. After that, we prove that $u_n$ is increasing with $n$ and strictly less than $l_2$. Next, we establish that $u_n$ is an upper bound of $\alpha_1(\cdot)$ in $[T_2,T_2+L]$. More precisely, we show that $\alpha_1(T_2+n)\leq u_n$ for $n \in [0,L]$. For that purpose, we proceed with two following steps: The first one consists of deriving an inequality verified by two consecutive terms of the sequence $\alpha_1(\cdot)$, namely $\alpha_1(T)$ and $\alpha_1(T+1)$ using the Proposition \ref{prop:inequality_satisfied_by_first_element} given that $T \geq T_{\epsilon}$. As for the second step, we use essentially the aforementioned result to demonstrate by induction that $u_n$ is indeed an upper bound of $\alpha_1(T_2+n)$. Finally, based on these results, we show that there exists $T_d$ such that $\max A_1(T_d)<l_2$.\\

To find a time $T_2 \geq T_{\epsilon}$ such that $\alpha_1(T_2)$ is less than $l_1$, we use the fact that $\min A_1(t) \leq l_1$ for all $t$. At time $T_{\epsilon}+L$, we have the vector $A_1(T_{\epsilon}+L)=(\alpha_1(T_{\epsilon}+L),\alpha_1(T_{\epsilon}+L-1),\cdots,\alpha_1(T_{\epsilon}+L-l(T_{\epsilon}+L)))$. Providing that $\min A_1(T_{\epsilon}+L) \leq l_1$, then there exists an element from the vector $A_1(T_{\epsilon}+L)$ less than $l_1$ denoted by $\alpha_1(T_2)$. According to \ref{prop:unique_expression_for_all_T}, we have for all $T\geq T_0$, $l(T) \leq l_{\max}=L$, then $l(T_{\epsilon}+L)\leq L$. That is, $T_2$ is greater than $T_{\epsilon}$ since $T_2 \geq T_{\epsilon}+L-l(T_{\epsilon}+L) \geq T_{\epsilon}$. Therefore, we find an element of the sequence $\alpha_1(\cdot)$ at time $T_2\geq T_{\epsilon}$ such that  $\alpha_1(T_2)\leq l_1$.   
To that extent, we are interested in proving that $\alpha_1(T_2),\cdots,\alpha_1(T_2+L)$ are strictly less than $l_2$.\\
To do so, we define a sequence $u_n$ which will constitute an upper bound of the function $\alpha_1(T)$.   
\begin{mydef}\label{def:suite}
We define a sequence $u_n$ by induction:
\begin{equation}
\left\{
    \begin{array}{ll}
        u_0=l_1& if n=0\\
        u_{n+1}=p_1(l_2+\epsilon)+(1-p_1)u_n & if n > 0
    \end{array}
\right.
\end{equation}
\end{mydef}
Next, we prove that the $L$ first terms of this sequence are strictly less than $l_2$. We detail this in the following.
\begin{Lemma}
For $n\in [0,L]$, $u_n<l_2$
\end{Lemma}
\begin{proof}
renewcommand{\qedsymbol}{$\blacksquare$}
In fact, the sequence $u_n$ satisfies for all $n$:
\begin{equation}u_n=\lambda(1-p_1)^n+(l_2+\epsilon)
\end{equation}
where $\lambda=-(\epsilon+l_2-l_1)$.\\
$u_n$ is clearly increasing with $n$, then for all $n \in [0,L]$:
\begin{equation}
u_n \leq u_L=\lambda(1-p_1)^L+(l_2+\epsilon)=\epsilon(1-(1-p_1)^L)+l_2-(l_2-l_1)(1-p_1)^L
\end{equation}\\
We have that:
\begin{equation}\epsilon < (l_2-l_1)(\frac{(1-p_1)^L}{1-(1-p)^L})\end{equation}
Given that $1-(1-p_1)^L \geq 0$, then:
\begin{equation}(1-(1-p_1)^L)\epsilon < (l_2-l_1)(1-p_1)^L\end{equation}
\begin{equation}(1-(1-p_1)^L)\epsilon+l_2-(l_2-l_1)(1-p_1)^L < l_2\end{equation}
Therefore, $u_L < l_2$.\\
\end{proof}

Based on the lemma above, we prove that for any element of the set $\{\alpha_1(T_2),\cdots,\alpha_1(T_2+L)\}$ must be less than $u_L$.\\ 
For that, we introduce a useful Lemma:
\begin{Lemma}\label{lem:relation_between_alpha}
If for $T \in [T_2,T_2+L-1]$, we have that:
\begin{equation}\alpha_1(T) \leq \alpha_1(T+1)\end{equation}
Then, we have that:
\begin{equation} \alpha_1(T+1) \leq p_1(l_2+\epsilon)+(1-p_1)\alpha_1(T)
\end{equation}  
\end{Lemma}
\begin{proof}
\renewcommand{\qedsymbol}{$\blacksquare$}
Before starting the proof, we recall that, according to the first result of Proposition \ref{prop:inequality_satisfied_by_first_element}, the four possible inequalities satisfied by $\alpha_1(T)$, $\alpha_1(T+1)$, $\alpha_1(T-l(T))$, $\alpha_1(T-l(T)+1)$ are: 
\begin{equation}\alpha_1(T) \leq \alpha_1(T+1) \leq \alpha_1(T-l(T))
\end{equation}  
\begin{equation}\alpha_1(T-l(T)) \leq \alpha_1(T+1) \leq \alpha_1(T)
\end{equation}
\begin{equation}\alpha_1(T-l(T)+1) \leq \alpha_1(T+1) \leq \alpha_1(T)
\end{equation}
\begin{equation}\alpha_1(T) \leq \alpha_1(T+1) \leq \alpha_1(T-l(T)+1)
\end{equation} 
Therefore, the two cases for which $\alpha_1(T) \leq \alpha_1(T+1)$ are: 
\begin{itemize}
\item $\alpha_1(T) \leq \alpha_1(T+1) \leq \alpha_1(T-l(T))$.\\
\item $\alpha_1(T) \leq \alpha_1(T+1) \leq \alpha_1(T-l(T)+1)$.\\
\end{itemize}
Hence, according to the results of Proposition \ref{prop:inequality_satisfied_by_first_element}, the inequalities satisfied by $\alpha_1(T+1)-\alpha_1(T)$ are:\\
If $\alpha_1(T) \leq \alpha_1(T+1) \leq \alpha_1(T-l(T))$, then:
\begin{equation}\alpha_1(T+1)-\alpha_1(T) \leq p_1(\alpha_1(T-l(T))-\alpha_1(T))\end{equation} 
If $\alpha_1(T) \leq \alpha_1(T+1) \leq \alpha_1(T-l(T)+1)$, then:
\begin{equation}\alpha_1(T+1)-\alpha_1(T) \leq p_1(\alpha_1(T-l(T)+1)-\alpha_1(T))\end{equation}
Since, by assumption of the Lemma, $T \geq T_2 \geq T_{\epsilon}$, then $\max A_1(T) \leq l_2+\epsilon$. As a consequence, $\alpha_1(T-l(T)+1)$ and $\alpha_1(T-l(T))$ which are elements of the vector $A_1(T)$, are less than $l_2+\epsilon$.\\
Hence, for T $\in [T_2,T_2+L-1]$:
\begin{equation}
\alpha_1(T+1)-\alpha_1(T) \leq p_1(l_2+\epsilon-\alpha_1(T))
\end{equation}
Therefore:
\begin{equation} \alpha_1(T+1) \leq p_1(l_2+\epsilon)+(1-p_1)\alpha_1(T)
\end{equation}  
\end{proof}

Now we should prove that for all possible sequences of $\alpha_1$ in $[T_2,T_2+L]$, their values can not exceed $\lambda(1-p_1)^L+(l_2+\epsilon_2)=u_L$.
\begin{Lemma}\label{lem:upper_bound_alpha}
For all sequences of $\alpha_1$ when $T \in [T_2,T_2+L]$, $\alpha_1(T) \leq u_{T-T_2}$
\end{Lemma}
\begin{proof}
\renewcommand{\qedsymbol}{$\blacksquare$}
We prove this result by induction.\\
For $T=T_2$, we have that:
\begin{equation}\alpha_1(T_2)\leq l_1=u_0\end{equation}
We suppose that at time $T$, $\alpha_1(T) \leq u_{T-T_2}$, then at time $T+1$:\\
If $\alpha_1(T+1) \leq \alpha_1(T)$:\\
Then as $u_{T-T_2}$ is increasing in $T$:
\begin{equation}
\alpha_1(T+1) \leq u_{T-T_2} \leq u_{T-T_2+1}
\end{equation}
If $\alpha_1(T+1) \geq \alpha_1(T)$:\\
Then, according to Lemma \ref{lem:relation_between_alpha}:
\begin{align}
\alpha_1(T+1)& \leq p_1(l_2+\epsilon)+(1-p_1)\alpha_1(T)\\
&\leq p_1(l_2+\epsilon)+(1-p_1)u_{T-T_2}\\
&=u_{T-T_2+1}
\end{align}

Therefore, $\alpha_1(T+1) \leq u_{T-T_2+1}$.\\
Hence, we have proved by induction that for all $T \in [T_2,T_2+L]$, $\alpha_1(T) \leq u_{T-T_2}$
\end{proof}

As $u_{T-T_2}$ is less than $u_L$ for $T \in [T_2,T_2+L]$, then according to Lemma \ref{lem:upper_bound_alpha}, the elements $\alpha_1(T_2+1),\cdots,\alpha_1(T_2+L)$ are less than $u_L < l_2$.\\
Thus, we have found $T_2 \geq T_{\epsilon}$ such that $\alpha_1(T_2),\alpha_1(T_2+1),\cdots,\alpha_1(T_2+l_{\max})$ are strictly less than $l_2$. We denote $T_2+l_{\max}$ by $T_d$ and we verify that $\max A_1(T_d) < l_2$. Indeed, we now that $T_d-l(T_d) \geq T_d-l_{\max}=T_2$, then the elements of the vector $A_1(T_d)$ are included in the set of elements $\{\alpha_1(T_2),\alpha_1(T_2+1),\cdots,\alpha_1(T_2+l_{\max})\}$. That is $\max A_1(T_d) < l_2$.\\ 
Hence, we have found $T_d \geq T_{\epsilon}$, such that $\max A_1(T_d) < l_2$.

\section{Proof of Proposition \ref{prop:z_convergence}}\label{app:prop:z_convergence}
In this proof, we show that for each state $i$ in class $k$, $z_i^k(t)$ converges. To that end, we start first by specifying the eventual limit of $z_i^k(t)$ for each $i$. To do so, we decompose $1-\alpha$ as follows:
\begin{equation}
l(p_1\alpha_1^*+p_2\alpha_2^*)+\gamma p_2 \alpha_2^*+\beta p_1 \alpha_1^*=1-\alpha
\end{equation}
where $l$ is the biggest integer such that: $l(p_1\alpha_1^*+p_2\alpha_2^*) < 1-\alpha$, and $0<\gamma\leq 1$ and $\beta=0$; or $\gamma=1$  and $0< \beta \leq 1$.
Then, we proceed with these following steps:
\begin{itemize}
    \item We prove by induction that for all states $1 \leq i \leq l+1$, $z_i^k(t)$ converges to $p_k\alpha_k^*$.
    \item Based on the theoretical findings of the first step, we prove that $z_{l+2}^1(t)$ converges to $(\beta +(1-p_1)(1-\beta))p_1 \alpha_1^*$ and $z_{l+2}^2(t)$ converges to $(\gamma+(1-p_2)(1-\gamma))p_2 \alpha_2^*$.
    \item Finally, we show that for all states $ i > l+2$, $z_i^1(t)$ converges to $(1-p_1)^{i-l-2}(\beta +(1-p_1)(1-\beta))p_1 \alpha_1^*$ and $z_i^2(t)$ converges to $(1-p_2)^{i-l-2}(\gamma+(1-p_2)(1-\gamma))p_2 \alpha_2^*$ 
\end{itemize}
\begin{enumerate}
    \item For all states $1 \leq i \leq l+1$, $z_i^k(t)\rightarrow p_k\alpha_k^*$:\\
We prove this result by induction
\begin{itemize}
\item For $i=1$, we have that $z^k_1(t)=p_k \alpha_k(t-1)$. Therefore, $z^k_1(t)$ converge to $p_k\alpha_k^*$ as $\alpha_k(t)$ converges to $\alpha_k^*$.\\
\item We consider that for a certain $j \leq l$, for each $1\leq i \leq j$, $z_i^k(t)$ converges to $p_k \alpha_k^*$ and we show that $z_{j+1}^k(t)$ converges also to $p_k \alpha_k^*$.\\
Given that $j \leq l$:
$$j(p_1 \alpha_1^*+p_2 \alpha_2^*) <1-\alpha$$
We consider $0 < \epsilon \leq 1-\alpha-j(p_1 \alpha_1^*+p_2 \alpha_2^*)$. Providing that $z_i^k(t)$ converges to $p_k \alpha_k^*$ for all $1\leq i \leq j$, that means there exists $t_j$ such that for $t \geq t_j$, for $1 \leq i \leq j$:
$$|z_i^k(t)-p_k\alpha_k^*| < \frac{\epsilon}{2j}$$
Hence: $$\sum_{i=1}^j |z_i^1(t)-p_1\alpha_1^*|+\sum_{i=1}^j |z_i^2(t)-p_2\alpha_2^*| < \epsilon$$
That is,$$\sum_{i=1}^j z_i^1(t)+\sum_{i=1}^j z_i^2(t) < \epsilon+j(p_1 \alpha_1^*+p_2 \alpha_2^*)$$
As consequence, for all $t \geq t_j$, we have that:
$$\sum_{i=1}^j z_i^1(t)+\sum_{i=1}^j z_i^2(t) < 1-\alpha$$
Thus, for all $t\geq t_j$, the action prescribed to the users' proportion $z^k_j(t)$ is the passive action \footnote{Knowing that the order of the proportions of the users according to the Whittle's index value alternates between the two classes in the set $[1,l_{\max}+1]$ as was established in \ref{prop:interval_alternation_under_assump}, then for all integer $b \in [1,l_{\max}]$, the set $\{z_i^k: k=1,2; 1\leq i\leq b\}$ is the set of users with the lowest Whittle's index value. Therefore, $\sum_{i=1}^b z_i^1(t)+\sum_{i=1}^b z_i^2(t)<1-\alpha$ implies that the actions prescribed to the users belonging to the set $\{z_i^k: k=1,2; 1\leq i\leq b\}$ is the passive action. By definition of $l$, $l<\frac{1-\alpha}{p_2\alpha}$, then, $l\leq l_{\max}$ (see Lemma \ref{lem:upper_bound_threshold_max}). Hence, the above reasoning can be applied as well when $b=l$.}. Then, for all $t \geq t_j$:
$$z_{j+1}^k(t+1)=z_j^k(t)$$
Therefore, $z_{j+1}^k(t)$ converges to $p_k\alpha_k^*$.
\end{itemize}
Consequently, we prove by induction that for all $1\leq i \leq l+1$, $z_i^k(t)$ converges to $p_k\alpha_k^*$.

\item $z_{l+2}^1(t)\rightarrow(\beta+(1-p_1)(1-\beta))p_1 \alpha_1^*$ and $z_{l+2}^2(t)\rightarrow(\gamma +(1-p_2)(1-\gamma))p_2 \alpha_2^*$.\\
To avoid redundancy , we will be limited to the first case when $0<\gamma\leq 1$ and $\beta=0$, since the proof's steps for both cases are exactly the same.
We have that: 
$$l(p_1\alpha_1^*+p_2\alpha_2^*)+\gamma p_2\alpha_2^*=1-\alpha$$
As $\sum_{i=1}^l z_i^1(t)+\sum_{i=1}^lz_i^2(t)$ converges to $l(p_1\alpha_1^*+p_2\alpha_2^*)$ which is strictly less than $1-\alpha$, then there exists $t_l$ such that for all $t \geq t_l$, we have that: 
$$\sum_{i=1}^l z_i^1(t)+\sum_{i=1}^lz_i^2(t)<1-\alpha$$
As $\sum_{i=1}^{l+1} z_i^1(t)+\sum_{i=1}^{l+1} z_i^2(t)$ converges to $(l+1)(p_1\alpha_1^*+p_2\alpha_2^*)$ which is strictly greater than $1-\alpha$, then there exists $t_{l+1}$ such that for all $t \geq t_{l+1}$, we have that: 
$$\sum_{i=1}^{l+1} z_i^1(t)+\sum_{i=1}^{l+1}z_i^2(t) >1-\alpha$$
For $t \geq \max \{t_l,t_{l+1}\}$, we have that:
$$\sum_{i=1}^{l} z_i^1(t)+\sum_{i=1}^{l}z_i^2(t) < 1-\alpha < \sum_{i=1}^{l+1} z_i^1(t)+\sum_{i=1}^{l+1}z_i^2(t)$$
Denoting $\gamma(t)$ and $\beta(t)$ the users' proportion of $z_{l+1}^2(t)$ and $z_{l+1}^1(t)$ respectively which are not scheduled, therefore, the relation that links $z_{l+2}^1(t+1)$ and $z_{l+2}^2(t+1)$ to $z_{l+1}^1(t)$ and $z_{l+1}^1(t)$ when $t \geq \max \{t_l,t_{l+1}\}$: 
$$z_{l+2}^1(t+1)=\beta(t) z_{l+1}^1(t)+(1-p_1) (1-\beta(t))z_{l+1}^1(t)$$ 
$$z_{l+2}^2(t+1)=\gamma(t) z_{l+1}^2(t)+(1-p_2) (1-\gamma(t))z_{l+1}^2(t)$$
with $0<\gamma(t)\leq 1$ and $\beta(t)=0$; or $\gamma(t)=1$  and $0< \beta(t) \leq 1$.
To that extent, we show that $\beta(t)$ tends to $\beta=0$ and $\gamma(t)$ tends to $\gamma$. For that purpose, we give the following equation which is always satisfied when $t \geq \max \{t_l,t_{l+1}\}$:  
\begin{equation}\label{eq:1-alpha_expression}
\sum_{i=1}^{l} z_i^1(t)+\sum_{i=1}^{l}z_i^2(t)+\gamma(t) z_{l+1}^2(t)+\beta(t) z_{l+1}^1(t)= 1-\alpha
\end{equation}
Tending $t$ to $+ \infty$ in the equation \ref{eq:1-alpha_expression}, we obtain: 
$$\underset{t \rightarrow +\infty}{\lim} [\gamma(t) z_{l+1}^2(t)+\beta(t) z_{l+1}^1(t)]= \gamma p_2 \alpha_2^*$$
We consider the set $\{t: \beta(t) \neq 0\}$. If this set is infinite, then there exists a strictly increasing function $n(.)$ from $\mathbf{N}$ to $\{t \in \mathbf{N} \beta(t) \neq 0\}$, such that $\beta(n(t))$ is a sub-sequence of $\beta(t)$.
As $\beta(n(t)) \neq 0$, then $\gamma(n(t))=1$. Therefore, we get:
$$\underset{t \rightarrow +\infty}{\lim} [ z_{l+1}^2(n(t))+\beta(n(t)) z_{l+1}^1(n(t))]= \gamma p_2 \alpha_2^*$$
Since $z_{l+1}^2(n(t))$ converges to $p_2\alpha_2^*$, then:
$$\underset{t \rightarrow +\infty}{\lim} [\beta(n(t)) z_{l+1}^1(n(t))]= (\gamma-1) p_2 \alpha_2^*$$
$(\gamma-1) p_2 \alpha_2^*$ is less than $0$, and $\beta(n(t)) z_{l+1}^1(n(t))$ is greater than $0$ for all $t$. Thus:
$$\underset{t \rightarrow +\infty}{\lim} [\beta(n(t)) z_{l+1}^1(n(t))]= (\gamma-1) p_2 \alpha_2^*=0$$
This implies that $\gamma=1=\gamma(n(t))$, and $\underset{t \rightarrow +\infty}{\lim} \beta(n(t))=0$ because $z_{l+1}^1(n(t))$ converges to $p_1\alpha_1^* \neq 0$.
Hence $\underset{t \rightarrow +\infty}{\lim} \beta(t)=0=\beta$, i.e. $\underset{t \rightarrow +\infty}{\lim} \gamma(t)=\gamma=1$.\\
If $\{t: \beta(t) \neq 0\}$ is finite, then there exists $t_e$ such that for all $t \geq t_e$, $\beta(t)=0$. Therefore, for all $t \geq t_e$, we have that: 
$$\underset{t \rightarrow +\infty}{\lim} [\gamma(t) z_{l+1}^2(t)]= \gamma p_2 \alpha_2^*$$
That means $\underset{t \rightarrow +\infty}{\lim} \beta(t)=0$, and $\underset{t \rightarrow +\infty}{\lim} \gamma(t)=\gamma$.
Hence, in both cases, $\beta(t)\rightarrow\beta=0$ and $\gamma(t)\rightarrow\gamma$.\\
Consequently, combining the last result with the one derived in the first step, we conclude that $z_{l+2}^1(t)$ converges to $(\beta+(1-p_1)(1-\beta))p_1 \alpha_1^*$ and $z_{l+2}^2(t)$ converges to $(\gamma +(1-p_2)(1-\gamma))p_2 \alpha_2^*$. 
Similar analysis can be applied to come with the aforementioned result when $\gamma(t)=1$ and $0<\beta(t) \leq 1$.
\item For $i>l+2$, $z_i^1(t)\rightarrow(1-p_1)^{i-l-2}(\beta+(1-p_1)(1-\beta))p_1 \alpha_1^*$ and $z_i^2(t)\rightarrow(1-p_2)^{i-l-2}(\gamma +(1-p_2)(1-\gamma))p_2 \alpha_2^*$:\\
For $t \geq \max\{t_l,t_{l+1}\}$, we are sure that the action prescribed to $z^k_i(t)$ for all $i \geq l+2$ is the active action. As consequence, $z^k_{i+1}(t+1)$ satisfies:
$$z^k_{i+1}(t+1)=(1-p_k)z^k_i(t)$$ 
Therefore, as $z_{l+2}^1(t)$ converges to $(\beta+(1-p_1)(1-\beta))p_1 \alpha_1^*$ and $z_{l+2}^2(t)$ converges to $(\gamma +(1-p_2)(1-\gamma))p_2 \alpha_2^*$, one can easily establish by induction that  $z_i^1(t)$ converges to $(1-p_1)^{i-l-2}(\beta+(1-p_1)(1-\beta))p_1 \alpha_1^*$ and $z_i^2(t)$ converges to $(1-p_2)^{i-l-2}(\gamma+(1-p_2)(1-\gamma))p_2 \alpha_2^*$ for all $i > l+2$.
\end{enumerate}
We conclude that for all states $i$ and $k=1,2$, $z^k_i(t)$ converges. On the other hands, according to Proposition \ref{prop:fixed_point}, the only possible limit of $\boldsymbol{z}(t)$ is $\boldsymbol{z}^*$. As consequence, for each $k$ and $i$, $z_i^k(t)$ converges to $z_i^{k,*}$.

**here**

\section{Proof of Proposition \ref{prop:kurth_theorem_age}}\label{app:prop:kurth_theorem_age}
For a given $\boldsymbol{z}$, let $m_1(\boldsymbol{z})$ and $m_2(\boldsymbol{z})$ be the highest states of the class 1 and the class 2 respectively and $l_1(\boldsymbol{z})$ and $l_2(\boldsymbol{z})$ be the thresholds of class 1 and 2 respectively at time $t$ when $\boldsymbol{Z}^N(t)=\boldsymbol{z}$.  Given that, we introduce the following lemma. 
\begin{Lemma}\label{lem:prob_bound_last_state}
For any $\mu$, there exists positive constant $C(\boldsymbol{z})$ such that:
\begin{equation}
P( ||\boldsymbol{Z}^N(t+1)-\boldsymbol{z}'|| \geq \mu | \boldsymbol{Z}^N(t)=\boldsymbol{z}) \leq \frac{C(\boldsymbol{z})}{N}
\end{equation}
where $C(\boldsymbol{z})$ is independent of $N$ and $\boldsymbol{z}'=Q(\boldsymbol{z})\boldsymbol{z}=\mathbf{E}(\boldsymbol{Z}^N(t+1)|\boldsymbol{Z}^N(t)=\boldsymbol{z})$
\end{Lemma}  
\begin{proof}
\renewcommand{\qedsymbol}{$\blacksquare$}
By definition of $m_1(\boldsymbol{z})$ and $m_2(\boldsymbol{z})$, we have that $\boldsymbol{z}=(z^1_1,\cdots,z_{m_1(\boldsymbol{z})}^1,z_1^2,\cdots,z^2_{m_2(\boldsymbol{z})})$. On can easily show that $m_1(\boldsymbol{z}')=m_1(\boldsymbol{z})+1$ and $m_2(\boldsymbol{z}')=m_2(\boldsymbol{z})+1$ since the users' proportions at states $m_1(\boldsymbol{z})$ and $m_2(\boldsymbol{z})$ in class 1 and class 2 will become at states $m_1(\boldsymbol{z})+1$ and $m_2(\boldsymbol{z})+1$ at the next time slot respectively.
To prove this lemma, we use the Chebychev inequality presented as follows:
\begin{equation}
P(|X-\mathbb{E}(X)| > \mu) \leq \frac{Var(X)}{\mu^2}
\end{equation}
for any $\mu>0$ and random variable $X$.\\
As $\boldsymbol{z}'=\mathbf{E}(\boldsymbol{Z}^N(t+1)|\boldsymbol{Z}^N(t)=\boldsymbol{z})$, we can apply the Chebychev inequality. However we need to find the distribution of $\boldsymbol{Z}^N(t+1)$ knowing $\boldsymbol{Z}^N(t)=\boldsymbol{z}$ in order to derive the expression of $Var(\boldsymbol{Z}^N(t+1)|\boldsymbol{Z}^N(t)=\boldsymbol{z})$.
It is more simple to study the parameters of one dimensional random variable than multi-dimensional random variable. Hence, instead of investigating $\boldsymbol{Z}^N(t+1)$, we look into $Z_i^{N,k}$. In this regard, we have that:
\begin{equation}
\{\boldsymbol{Z}^N(t+1): \ ||\boldsymbol{Z}^N(t+1)-\boldsymbol{z}'|| \geq \mu \} \subset \underset{k,i}{\cup} \{\boldsymbol{Z}^N(t+1): \ ||Z_i^{N,k}(t+1)-z_i^{'k}||i > \frac{\mu}{m_1(\boldsymbol{z}')+m_2(\boldsymbol{z}')}\}
\end{equation}  
Therefore:
\begin{align}
P( ||\boldsymbol{Z}^N(t+1)-\boldsymbol{z}'|| \geq \mu | \boldsymbol{Z}^N(t)=\boldsymbol{z})& \leq P( \underset{k,i}{\cup} \{||Z_i^{N,k}(t+1)-z_i^{'k}||i > \frac{\mu}{m_1(\boldsymbol{z}')+m_2(\boldsymbol{z}')} | \boldsymbol{Z}^N(t)=\boldsymbol{z}\})\\
&\leq  \sum_{k,i} P( \{||Z_i^{N,k}(t+1)-z_i^{'k}||i > \frac{\mu}{m_1(\boldsymbol{z}')+m_2(\boldsymbol{z}')} | \boldsymbol{Z}^N(t)=\boldsymbol{z}\})
\end{align}
Now, we look for the distribution of $Z_i^{N,k}(t+1)$ knowing $\boldsymbol{Z}^N(t)=\boldsymbol{z}$.\\ 
For $2 \leq i \leq l_k(\boldsymbol{z})$, as all the users at state $i-1$ less strictly than $l_k(\boldsymbol{z})$ will transit to the state $i$ at the next time slot, then we have $Z_i^{N,k}(t+1)=z_{i-1}^{k}=z_i^{'k}$. This implies that:
\begin{equation}
 P( \{||Z_i^{N,k}(t+1)-z_i^{'k}||i > \frac{\mu}{m_1(\boldsymbol{z}')+m_2(\boldsymbol{z}')} | \boldsymbol{Z}^N(t)=\boldsymbol{z}\})=0
\end{equation}
For $i=1$, defining $\alpha_1(\boldsymbol{z})$ and $\alpha_2(\boldsymbol{z})$ as the proportions of the scheduled users in class 1 an class 2 respectively when $\boldsymbol{Z}^N(t)=\boldsymbol{z}$, then $N Z_1^{N,k}(t+1) | \boldsymbol{Z}^N(t)=\boldsymbol{z}$ follows a binomial distribution with parameters $p_k$ and $\alpha_k(\boldsymbol{z}) N$. Therefore, $Var(N Z_1^{N,k}(t+1) | \boldsymbol{Z}^N(t)=\boldsymbol{z})= p_k (1-p_k) \alpha_k(\boldsymbol{z}) N$, which means that $Var(Z_1^{N,k}(t+1) | \boldsymbol{Z}^N(t)=\boldsymbol{z})= \frac{p_k (1-p_k) \alpha_k(\boldsymbol{z})}{N}$.
As a results, according to Chebychev inequality, we have that:
\begin{equation}
 P( \{||Z_1^{N,k}(t+1)-z_1^{'k}|| > \frac{\mu}{m_1(\boldsymbol{z}')+m_2(\boldsymbol{z}')} | \boldsymbol{Z}^N(t)=\boldsymbol{z}\})\leq \frac{p_k (1-p_k) \alpha_k(\boldsymbol{z})}{N \mu^2}(m_1(\boldsymbol{z}')+m_2(\boldsymbol{z}'))^2 
\end{equation}
For $i \geq l_k(\boldsymbol{z})+2$, $N Z_i^{N,k}(t+1) | \boldsymbol{Z}^N(t)=\boldsymbol{z}$ follows a binomial distribution with parameters $1-p_k$ and $z_{i-1}^{k} N$. Hence, $Var(Z_i^{N,k}(t+1) | \boldsymbol{Z}^N(t)=\boldsymbol{z})=  \frac{p_k (1-p_k) z_{i-1}^{k}}{N}$.
Thus:  
\begin{equation}
 P( \{||Z_i^{N,k}(t+1)-z_i^{'k}|| > \frac{\mu}{i(m_1(\boldsymbol{z}')+m_2(\boldsymbol{z}'))} | \boldsymbol{Z}^N(t)=\boldsymbol{z}\})\leq \frac{p_k (1-p_k) z_{i-1}^{k}}{N \mu^2}(m_1(\boldsymbol{z}')+m_2(\boldsymbol{z}'))^2 i^2 
\end{equation}
Denoting $\beta_k(\boldsymbol{z})$ the users' proportion of $z_{l_k(\boldsymbol{z})}^k$ that will not be transmitted, then for $i=l_k(\boldsymbol{z})+1$, $N Z_i^{N,k}(t+1) | (\boldsymbol{Z}^N(t)=\boldsymbol{z})= \beta_k(\boldsymbol{z}) N z_{i-1}^k+X$, where $X$ follows a binomial distribution with parameters $1-p_k$ and $(1-\beta_k(\boldsymbol{z}))z_{i-1}^{k} N$, then:
\begin{equation}
P( \{||Z_i^{N,k}(t+1)-z_i^{'k}|| > \frac{\mu}{i(m_1(\boldsymbol{z}')+m_2(\boldsymbol{z}'))} | \boldsymbol{Z}^N(t)=\boldsymbol{z}\})\leq \frac{p_k (1-p_k)(1-\beta_k(\boldsymbol{z})) z_{i-1}^{k}}{N \mu^2}(m_1(\boldsymbol{z}')+m_2(\boldsymbol{z}'))^2 i^2
\end{equation}
We end up with:
$$P( ||\boldsymbol{Z}^N(t+1)-\boldsymbol{z}'|| \geq \mu | \boldsymbol{Z}^N(t)=\boldsymbol{z}) \leq $$
$$(m_1(\boldsymbol{z}')+m_2(\boldsymbol{z}'))^2 .[ \frac{p_1(1-p_1) \alpha_1(\boldsymbol{z})}{N \mu^2}+\frac{p_2(1-p_2) \alpha_2(\boldsymbol{z})}{N \mu^2}+ \sum_{i \geq l_1(\boldsymbol{z})+2} \frac{p_1(1-p_1)i^2 z^1_{i-1}}{N \mu^2}+$$
$$\sum_{i \geq l_2(\boldsymbol{z})+2} \frac{p_2(1-p_2)i^2 z^2_{i-1}}{N \mu^2}+ \frac{p_1(1-p_1) (l_1(\boldsymbol{z})+1)^2 (1-\beta_1(\boldsymbol{z})) z^1_{l_1(\boldsymbol{z})}}{N \mu^2}+ \frac{p_1(1-p_2) (l_2(\boldsymbol{z})+1)^2 (1-\beta_2(\boldsymbol{z})) z^2_{l_2(\boldsymbol{z})}}{N \mu^2}]$$
Knowing that $\alpha_k(\boldsymbol{z}) \leq 1$, $\sum_{i \geq l_k(\boldsymbol{z})} z^k_i \leq 1$, $1-\beta_k(\boldsymbol{z}) \leq 1$, and for all state $i$ in the vector $\boldsymbol{z}'$, $i \leq m_1(\boldsymbol{z}')+m_2(\boldsymbol{z}')$ then:
$$P( ||\boldsymbol{Z}^N(t+1)-\boldsymbol{z}'|| \geq \mu | \boldsymbol{Z}^N(t)=\boldsymbol{z}) \leq  \frac{(m_1(\boldsymbol{z}')+m_2(\boldsymbol{z}'))^4}{\mu^2 N} [2 p_1(1-p_1)+2 p_2(1-p_2)]$$
Hence, denoting by $C(\boldsymbol{z})$, $\frac{(m_1(\boldsymbol{z}')+m_2(\boldsymbol{z}'))^4}{\mu^2} [2 p_1(1-p_1)+2 p_2(1-p_2)]=\frac{(m_1(\boldsymbol{z})+1+m_2(\boldsymbol{z})+1)^4}{\mu^2} [2 p_1(1-p_1)+2 p_2(1-p_2)]$, we obtain as a result:
\begin{equation}
P( ||\boldsymbol{Z}^N(t+1)-\boldsymbol{z}'|| \geq \mu | \boldsymbol{Z}^N(t)=\boldsymbol{z}) \leq \frac{C(\boldsymbol{z})}{N}
\end{equation}
\end{proof}
Now, we give a lemma that bounds the probability knowing the initial state $\boldsymbol{z}(0)=\boldsymbol{x}$.
One can easily verifies that $m_1(\boldsymbol{z}(t))=m_1(\boldsymbol{x})+t$ and $m_2(\boldsymbol{z}(t))=m_2(\boldsymbol{x})+t$ by induction. Without loss of generality, we let $m_k(\boldsymbol{z}(t))=m_k(t)$ for $k=1,2$.\\
\begin{Lemma}\label{lem:prob_bound_initial_state}
For any $\mu$, there exists positive constant $C(t+1)$ such that:
\begin{equation}
P_{\boldsymbol{x}}( ||\boldsymbol{Z}^N(t+1)-\boldsymbol{z}(t+1)|| \geq \mu) \leq \frac{C(t+1)}{N}
\end{equation}
where $C(t+1)$ is independent of $N$.
\end{Lemma}
\begin{proof}
\renewcommand{\qedsymbol}{$\blacksquare$}
We recall from Lemma \ref{lem:prob_bound_last_state} that for any $\mu>0$, there exists a constant $C(\boldsymbol{z})$ independent of $N$ such that:
\begin{equation}
P( ||\boldsymbol{Z}^N(t+1)-Q(\boldsymbol{z})\boldsymbol{z}|| \geq \mu | \boldsymbol{Z}^N(t)=\boldsymbol{z}) \leq \frac{C(\boldsymbol{z})}{N}
\end{equation}
Before proving the present lemma, we give an important lemma that will helps us in the later analysis. 
\begin{Lemma}
For any proportion vector $\boldsymbol{z}$, there exists $\sigma>0$ such that if $||\boldsymbol{Z}^N(t)-\boldsymbol{z}||\leq \sigma$, then $\boldsymbol{Q}(\boldsymbol{Z}^N(t))=\boldsymbol{Q}(\boldsymbol{z})$.
\end{Lemma}
\begin{proof}
One can deduce from the analysis done in \cite[Section~IV-C]{maatouk2020optimality} that there exists $\sigma>0$ such that if $\boldsymbol{Z}^N(t) \in \Omega_{\sigma}(\boldsymbol z)$, $\boldsymbol{Q}(\boldsymbol{Z}^N(t))$ is constant and doesn't depend on $\boldsymbol Z^N(t)$. 
Therefore, there exists $\sigma>0$ such that $\boldsymbol{Q}(\boldsymbol{Z}^N(t))=\boldsymbol{Q}(\boldsymbol{z})$. That concludes the proof.
\end{proof}
\begin{corollary}\label{cor:continuous_Q}
For any $v>0$, there exists $\rho$ such that $||\boldsymbol{Z}^N(t)-\boldsymbol{z}(t)||\leq\rho \Rightarrow ||\boldsymbol{Q}(\boldsymbol{Z}^N(t))\boldsymbol{Z}^N(t)-\boldsymbol{Q}(\boldsymbol{z}(t))\boldsymbol{z}(t)||\leq v$
\end{corollary}
\begin{proof}
According to the previous lemma, if $||\boldsymbol{Z}^N(t)-\boldsymbol{z}(t)||\leq \sigma$, then $\boldsymbol{Q}(\boldsymbol{Z}^N(t))=\boldsymbol{Q}(\boldsymbol{z}(t))$. This implies that $||\boldsymbol{Q}(\boldsymbol{Z}^N(t))\boldsymbol{Z}^N(t)-\boldsymbol{Q}(\boldsymbol{z}(t))\boldsymbol{z}(t)||=||\boldsymbol{Q}(\boldsymbol{z}(t))\boldsymbol{Z}^N(t)-\boldsymbol{Q}(\boldsymbol{z}(t))\boldsymbol{z}(t)||\leq ||\boldsymbol{Q}(\boldsymbol{z}(t))|| ||\boldsymbol{Z}^N(t)-\boldsymbol{z}(t)||$. 
That is, choosing $\rho=\min\{\frac{v}{||\boldsymbol{Q}(\boldsymbol{z}(t))||},\sigma\}$, we get $||\boldsymbol{Q}(\boldsymbol{Z}^N(t))\boldsymbol{Z}^N(t)-\boldsymbol{Q}(\boldsymbol{z}(t))\boldsymbol{z}(t)||\leq v$.
\end{proof}
With the above corollary being laid out, we prove the statement by a mathematical induction. 

For $t=1$, applying Lemma \ref{lem:prob_bound_last_state}, the following holds:
\begin{align}
{\Pr}_{\boldsymbol{x}}(||\boldsymbol{Z}^N(1)-\boldsymbol{z}(1)||\geq\mu)
=&P( ||\boldsymbol{Z}^N(t+1)-Q(\boldsymbol{x})\boldsymbol{x}|| \geq \mu | \boldsymbol{Z}^N(t)=\boldsymbol{x}) \leq \frac{C(x)}{N}\nonumber\\&=\frac{C(1)}{N}
\end{align}
and the desired result holds for $t=1$ by simply choosing $C(1)=\frac{(m_1(\boldsymbol{x})+1+m_2(\boldsymbol{x})+1)^4}{\mu^2} [2 p_1(1-p_1)+2 p_2(1-p_2)]$. Let us suppose that the statement holds for any $t\geq 1$. We investigate the property for $t+1$. To that end, let us consider $\nu<\mu$. Therefore, according to Corollary \ref{cor:continuous_Q}, there exists $\rho$ such that:  
\begin{equation}\label{eq:relation_less_v}
||\boldsymbol{Z}^N(t)-\boldsymbol{z}(t)||\leq\rho \Rightarrow ||\boldsymbol{Q}(\boldsymbol{Z}^N(t))\boldsymbol{Z}^N(t)-\boldsymbol{Q}(\boldsymbol{z}(t))\boldsymbol{z}(t)||\leq v
\end{equation}
Bearing that in mind, we have that:
\begin{align}
{\Pr}_{\boldsymbol{x}}(||\boldsymbol{Z}^N(t+1)-\boldsymbol{z}(t+1)||\geq\mu)=&{\Pr}_{\boldsymbol{x}}(||\boldsymbol{Z}^N(t+1)-\boldsymbol{z}(t+1)||\geq\mu\Big|||\boldsymbol{Z}^N(t)-\boldsymbol{z}(t)||\geq\rho){\Pr}_{\boldsymbol{x}}(||\boldsymbol{Z}^N(t)-\boldsymbol{z}(t)||\geq\rho)\nonumber \\
&+{\Pr}_{\boldsymbol{x}}(||\boldsymbol{Z}^N(t+1)-\boldsymbol{z}(t+1)||\geq\mu\Big|||\boldsymbol{Z}^N(t)-\boldsymbol{z}(t)||<\rho){\Pr}_{\boldsymbol{x}}(||\boldsymbol{Z}^N(t)-\boldsymbol{z}(t)||<\rho)\nonumber\\
\leq^{(a)}& \frac{C'(t)}{N}+{\Pr}_{\boldsymbol{x}}(||\boldsymbol{Z}^N(t+1)-\boldsymbol{z}(t+1)||\geq\mu\Big|||\boldsymbol{Z}^N(t)-\boldsymbol{z}(t)||<\rho)
\label{firststepinit}
\end{align}
where $(a)$ follows from ${\Pr}_{\boldsymbol{x}}(||\boldsymbol{Z}^N(t+1)-\boldsymbol{z}(t+1)||\geq\mu\Big|||\boldsymbol{Z}^N(t)-\boldsymbol{z}(t)||\geq\rho)\leq1$ and $C'(t)$ being the constant related to the statement holding for $t$ and for $\rho$. Next, we tackle the second term of the inequality in (\ref{firststepinit}):
\begin{align}
{\Pr}_{\boldsymbol{x}}(||\boldsymbol{Z}^N(t+1)-&\boldsymbol{z}(t+1)||\geq\mu\Big|||\boldsymbol{Z}^N(t)-\boldsymbol{z}(t)||<\rho)\nonumber\\
=&{\Pr}_{\boldsymbol{x}}(||\boldsymbol{Z}^N(t+1)-Q(\boldsymbol{Z}^N(t))\boldsymbol{Z}^N(t)+Q(\boldsymbol{Z}^N(t))\boldsymbol{Z}^N(t)-\boldsymbol{z}(t+1)||\geq\mu\Big|||\boldsymbol{Z}^N(t)-\boldsymbol{z}(t)||<\rho)\nonumber\\
\leq^{(a)}&{\Pr}_{\boldsymbol{x}}(||\boldsymbol{Z}^N(t+1)-Q(\boldsymbol{Z}^N(t))\boldsymbol{Z}^N(t)||+||Q(\boldsymbol{Z}^N(t))\boldsymbol{Z}^N(t)-Q(\boldsymbol{z}(t))\boldsymbol{z}(t)||\geq\mu\Big|||\boldsymbol{Z}^N(t)-\boldsymbol{z}(t)||<\rho)\nonumber\\
\leq^{(b)}&{\Pr}_{\boldsymbol{x}}(||\boldsymbol{Z}^N(t+1)-Q(\boldsymbol{Z}^N(t))\boldsymbol{Z}^N(t)||\geq\mu-\nu \Big|||\boldsymbol{Z}^N(t)-\boldsymbol{z}(t)||<\rho)\nonumber\\
=&\sum_{\substack{\boldsymbol{z}\in\Omega_{\rho}(\boldsymbol{z}(t))\\ m_k(\boldsymbol z) \leq m_k(\boldsymbol z(t)) \\ k=1,2}}{\Pr}_{\boldsymbol{x}}(\boldsymbol{Z}^N(t)=\boldsymbol{z}\Big|\boldsymbol{Z}^N(t)\in\Omega_{\rho}(\boldsymbol{z}(t))){\Pr}_{\boldsymbol{x}}(||\boldsymbol{Z}^N(t+1)-Q(\boldsymbol{z})\boldsymbol{z}||\geq\mu-\nu|\boldsymbol{Z}^N(t)=\boldsymbol{z})\nonumber\\
&+\sum_{\substack{\boldsymbol{z}\in\Omega_{\rho}(\boldsymbol{z}(t))\\ m_1(\boldsymbol z) > m_1(\boldsymbol z(t)) \\ or \\ m_2(\boldsymbol z) > m_2(\boldsymbol z(t))}}{\Pr}_{\boldsymbol{x}}(\boldsymbol{Z}^N(t)=\boldsymbol{z}\Big|\boldsymbol{Z}^N(t)\in\Omega_{\rho}(\boldsymbol{z}(t))){\Pr}_{\boldsymbol{x}}(||\boldsymbol{Z}^N(t+1)-Q(\boldsymbol{z})\boldsymbol{z}||\geq\mu-\nu|\boldsymbol{Z}^N(t)=\boldsymbol{z})
\label{secondstepinit}
\end{align}
where $(a)$ and $(b)$ follows from the triangular inequality and the relationship in \eqref{eq:relation_less_v}. 
One can notice that at any time slot $t$, $m_k(\boldsymbol Z^N(t)) \leq m_k(\boldsymbol z(t))$. In light of that fact, the second term of the equation \eqref{secondstepinit} is equal to $0$. Bearing that in mind,
We have for $\boldsymbol{z} \in \Omega_{\rho}(\boldsymbol{z}(t))$ such that  $m_k(\boldsymbol z) \leq m_k(\boldsymbol z(t))$:
\begin{align}
&{\Pr}_{\boldsymbol{x}}(||\boldsymbol{Z}^N(t+1)-Q(\boldsymbol{z})\boldsymbol{z}||\geq\mu-\nu|\boldsymbol{Z}^N(t)=\boldsymbol{z})\leq \frac{C_1(\boldsymbol{z}(t))}{N}
\end{align}
where $C_1(t)=\frac{(m_1(\boldsymbol{z}(t))+m_2(\boldsymbol{z}(t))+2)^4}{(\mu-\nu)^2} [2 p_1(1-p_1)+2 p_2(1-p_2)]=\frac{(m_1(t)+m_2(t)+2)^4}{(\mu-\nu)^2} [2 p_1(1-p_1)+2 p_2(1-p_2)]$. By substituting the above results in (\ref{secondstepinit}), we get:
\begin{equation}
{\Pr}_{\boldsymbol{x}}(||\boldsymbol{Z}^N(t+1)-\boldsymbol{z}(t+1)||\geq\mu\Big|||\boldsymbol{Z}^N(t)-\boldsymbol{z}(t)||<\rho)\leq \frac{C_1(t)}{N}
\end{equation}
Combining this with (\ref{firststepinit}), we can conclude that there exists a constant $C(t+1)$ such that:
\begin{equation}
{\Pr}_{\boldsymbol{x}}(||\boldsymbol{Z}^N(t+1)-\boldsymbol{z}(t+1)||\geq\mu)\leq \frac{C(t+1)}{N}
\end{equation}
which concludes our inductive proof.
\end{proof}
Knowing that: 
\[P_x(\underset{0 \leq t < T}{\text{sup}} ||\boldsymbol{Z}^N(t)-\boldsymbol{z}(t)|| \geq \mu)\leq \sum_{t=0}^{T-1}P_x( ||\boldsymbol{Z}^N(t)-\boldsymbol{z}(t)|| \geq \mu)\]
Therefore, from Lemma \ref{lem:prob_bound_initial_state}, there exists a constant $C$ which doesn't depend on $N$ such that: 
\[P_x(\underset{0 \leq t < T}{\text{sup}} ||\boldsymbol{Z}^N(t)-\boldsymbol{z}(t)|| \geq \mu) \leq \frac{C}{N}\]
Which concludes the proof.

\section{Proof of Lemma \ref{lem:result_kurth_theorem}}\label{app:lem:result_kurth_theorem}

We show first of all that $\boldsymbol{z}(t)$ converges to $\boldsymbol{z}^*$ with respect to our considered norm, i.e. $\underset{t \rightarrow +\infty}{\lim}\sum_{i=1}^{+\infty} |z_i^k(t)-z_i^{k,*}|i=0$ for $k=1,2$. For that purpose, we use the limit inversion theorem which states that:
\begin{itemize}
    \item If the series $\sum_i f_i(t)$ is uniformly convergent on $\mathbb{R}^+$ 
    \item If for each integer $i$, $f_i(t)$ admits a finite limit $r_i$ when $t$ tends to $+\infty$.
\end{itemize}
Therefore, $\underset{t \rightarrow +\infty}{\lim}\sum_{i=1}^{+\infty} f_i(t)=\sum_{i=1}^{+\infty}\underset{t \rightarrow +\infty}{\lim} f_i(t)=\sum_{i=1}^{+\infty} r_i$.\\
By letting  $f_i(t)$ denotes $|z_i^k(t)-z_i^{k,*}|i$ for a given $k$, proving the result above is equivalent to establish that:
$$\underset{t \rightarrow +\infty}{\lim}\sum_{i=1}^{+\infty} |z_i^k(t)-z_i^{k,*}|i=\sum_{i=1}^{+\infty}\underset{t \rightarrow +\infty}{\lim} |z_i^k(t)-z_i^{k,*}|i$$
To that extent, we check if the aforementioned conditions are satisfied for this specific function $f_i(t)=|z_i^k(t)-z_i^{k,*}|i$.
\begin{itemize}
    \item Uniform convergence: According to Weierstrass criterion, $\sum_i f_i(t)$ is uniformly convergent if for each $i$ the function $f_i(t)$ is bounded by a constant $c_i$ such that $\sum_i c_i$ is convergent. Based on the proof of the Proposition \ref{prop:z_convergence}, one can deduce that for large enough $t$ denoted by $t_l$, the following induction relation always holds for $t\geq t_l$ and $i\geq l_{\max}+1$:
    $$z_{i+1}^k(t+1)=p_k z^k_i(t)$$
    That is, choosing $t_0$ greater than $t_l$, and denoting by $i_0=m_k(t_0)$ the highest state of the vector $\boldsymbol{z}(t_0)$ which is greater than $l_{\max}+1$, we have that for each $i>i_0$:
    \begin{equation}
    z_i^k(t)=\left\{
    \begin{array}{ll}
        0 & if  \ \ t_0 \leq t < t_0+i-i_0\\
        p_k^{i-i_0}z_{i_0}^k(t-(i-i_0)) & if \ \ t \geq t_0+i-i_0 
    \end{array}
    \right.
    \end{equation}
Based on the above equation, for each $i>i_0$, $z^k_i(t)$ is less than $p_k^{i-i_0}$ for all $t\geq t_0$. To that extent, we investigate the evolution of the series of interest only when $t\geq t_0$ (the limit inversion theorem still applicable since $+\infty>t_0$).
Moreover, we have that for all $t \geq t_0$: $$\sum_i |z_i^k(t)-z_i^{k,*}|i=\sum_{i=1}^{i_0} |z_i^k(t)-z_i^{k,*}|i+\sum_{i_0+1}^{+\infty} |z_i^k(t)-z_i^{k,*}|i \leq i_0^2+\sum_{i=i_0+1}^{+\infty} (p_k^{i-i_0}i+z_i^{k,*}i)$$
This last sum is known to be a finite sum since $\sum_{i=1}^{+\infty} z_i^{k,*}i$ is the optimal average age of the relaxed problem for the class $k$ which is finite, and $\sum_{i=1}^{+\infty} p^{i}i$ is a finite sum for any $0\leq p< 1$. Hence, the uniform convergence can be accordingly concluded.
    \item Existence of the limit of $f_i(t)=|z_i^k(t)-z_i^{k,*}|i$: According to the result of Proposition \ref{prop:z_convergence}, we have $\underset{t \rightarrow +\infty}{\lim}|z_i^k(t)-z_i^{k,*}|i=0$ which is finite. Therefore, the second condition is satisfied.   
\end{itemize}
Leveraging these findings, we can inverse the order between the limit and the sum. Subsequently: 
$$\underset{t \rightarrow +\infty}{\lim}\sum_{i=1}^{+\infty} |z_i^k(t)-z_i^{k,*}|i=\sum_{i=1}^{+\infty}\underset{t \rightarrow +\infty}{\lim} |z_i^k(t)-z_i^{k,*}|i=0$$
In other words, for $k=1,2$, $\sum_{i=1}^{+\infty} |z_i^k(t)-z_i^{k,*}|i$ tends to $0$ when $t$ grows. Consequently, $\boldsymbol{z}(t)$ converges to $\boldsymbol{z}^*$ with respect to our defined norm.\\
Therefore, for $0<\nu<\mu$, there exists $T_0$ such that for any $t\geq T_0$:
\begin{equation}
||\boldsymbol{z}(t)-\boldsymbol{z}^*||\leq\nu
\end{equation}
By leveraging Proposition \ref{prop:kurth_theorem_age}, we have:
\begin{align}
&{\Pr}_{\boldsymbol{x}}(\underset{T_0\leq t<T}{\sup} ||\boldsymbol{Z}^N(t)-\boldsymbol{z}^*||\geq\mu)
\nonumber\\&
\leq {\Pr}_{\boldsymbol{x}}(\underset{T_0\leq t<T}{\sup} ||\boldsymbol{Z}^N(t)-\boldsymbol{z}(t)||+||\boldsymbol{z}(t)-\boldsymbol{z}^*||\geq\mu)\nonumber\\&\leq {\Pr}_{\boldsymbol{x}}(\underset{T_0\leq t<T}{\sup} ||\boldsymbol{Z}^N(t)-\boldsymbol{z}(t)||\geq\mu-\nu)\nonumber\\&
\leq
{\Pr}_{\boldsymbol{x}}(\underset{0\leq t<T}{\sup} ||\boldsymbol{Z}^N(t)-\boldsymbol{z}(t)||\geq\mu-\nu)\leq \frac{s}{N}
\end{align}
which concludes the proof.

\section{Proof of Proposition \ref{prop:optimality_whittle_index}}\label{app:prop:optimality_whittle_index}

We have that: 
\begin{align}
\big |\frac{1}{T} \mathbb{E}^{wi}\left[ \sum_{t=0}^{T-1}  \sum_{k=1}^{K} \sum_{i=1}^{+\infty} Z_i^{k,N}(t)i \Big | \boldsymbol{Z}^N(0)=\boldsymbol{x}\right]-\sum_{k=1}^K \sum_{i=1}^{+\infty} z_i^{k,*}i \big |=& \big |\frac{1}{T} \mathbb{E}^{wi}\left[ \sum_{t=0}^{T-1}  \sum_{k=1}^{K} \sum_{i=1}^{+\infty} (Z_i^{k,N}(t)i- z_i^{k,*}i) \Big | \boldsymbol{Z}^N(0)=\boldsymbol{x} \right] \big |\\
\leq& \big |\frac{1}{T}  \sum_{t=0}^{T_0-1}  \sum_{k=1}^{K} \sum_{i=1}^{+\infty} \mathbb{E}^{wi}\left[Z_i^{k,N}(t)i - z_i^{k,*}i\Big | \boldsymbol{Z}^N(0)=\boldsymbol x \right] \big | \label{eq:first_term} \\
+& \big |\frac{1}{T}  \sum_{t=T_0}^{T-1}  \sum_{k=1}^{K} \sum_{i=1}^{+\infty} \mathbb{E}^{wi}\left[Z_{i}^{k,N}(t)i - z_i^{k,*}i\Big | \boldsymbol{Z}^N(0)=\boldsymbol{x} \right] \big | \label{eq:second_term}
\end{align}
We start by bounding \eqref{eq:first_term}. We have that: 
\begin{align} 
\big |\frac{1}{T}  \sum_{t=0}^{T_0-1}  \sum_{k=1}^{K} \sum_{i=1}^{+\infty} \mathbb{E}^{wi}\left[Z_i^{k,N}(t)i - z_i^{k,*}i\Big | \boldsymbol{Z}^N(0)=\boldsymbol{x} \right] \big | &\leq \frac{1}{T}  \sum_{t=0}^{T_0-1}  \sum_{k=1}^{K} \sum_{i=1}^{+\infty} \mathbb{E}^{wi}\left[ \big |Z_i^{k,N}(t)i - z_i^{k,*}i \big | \Big | \boldsymbol{Z}^N(0)=\boldsymbol{x} \right]\\ 
&\leq \frac{1}{T}  \sum_{t=0}^{T_0-1}  \sum_{k=1}^{K} \sum_{i=1}^{+\infty} \mathbb{E}^{wi}\left[Z_i^{k,N}(t)i\Big | \boldsymbol{Z}^N(0)=\boldsymbol{x} \right]] + \frac{1}{T}  \sum_{t=0}^{T_0-1}  \sum_{k=1}^{K} \sum_{i=1}^{+\infty}z_i^{k,*}i\\
&=\frac{1}{T}  \sum_{t=0}^{T_0-1}  \sum_{k=1}^{K} \sum_{i=1}^{\max\{m_1(t),m_2(t)\}} \mathbb{E}^{wi}\left[Z_i^{k,N}(t)i\Big | \boldsymbol{Z}^N(0)=\boldsymbol{x} \right] + \frac{1}{T}  \sum_{t=0}^{T_0-1} C^{RP}
\end{align}
As $m_k(.)$ is increasing with $t$, then denoting $m(t)=\max\{m_1(t),m_2(t)\}$, we get:
\begin{align}
\frac{1}{T}  \sum_{t=0}^{T_0-1}  \sum_{k=1}^{K} \sum_{i=1}^{\max\{m_1(t),m_2(t)\}} \mathbb{E}^{wi}\left[Z_i^{k,N}(t)i\Big | \boldsymbol{Z}^N(0)=\boldsymbol{x} \right] + \frac{1}{T}  \sum_{t=0}^{T_0-1} C^{RP} \leq \frac{(m(T_0)+C^{RP})T_0}{T} \label{eq:first_bound}
\end{align}


We denote $Y_N$ the event {$\underset{T_0 \leq t < T}{\text{sup}} ||\boldsymbol{Z}^N(t)-\boldsymbol{z}^*|| \geq \mu$}, and we proceed to bound the second term \eqref{eq:second_term}. 
\begin{align}
\big |\frac{1}{T}  \sum_{t=T_0}^{T-1} \sum_{k=1}^{K} \sum_{i=1}^{+\infty} \mathbb{E}^{wi}\left[Z_{i}^{k,N}(t)i - z_i^{k,*}i\Big | \boldsymbol{Z}^N(0)=\boldsymbol{x} \right] \big |&=P_{\boldsymbol{x}}(Y_N) \big |\frac{1}{T}  \sum_{t=T_0}^{T-1}  \sum_{k=1}^{K} \sum_{i=1}^{+\infty} \mathbb{E}^{wi}\left[Z_{i}^{k,N}(t)i - z_i^{k,*}i\Big | Y_N, \boldsymbol{Z}^N(0)=\boldsymbol{x} \right] \big |+\\
&(1-P_{\boldsymbol{x}}(Y_N)) \big |\frac{1}{T}  \sum_{t=T_0}^{T-1}  \sum_{k=1}^{K} \sum_{i=1}^{+\infty} \mathbb{E}^{wi}\left[Z_{i}^{k,N}(t)i - z_i^{k,*}i \Big | \overline{Y_N}, \boldsymbol{Z}^N(0)=\boldsymbol{x} \right] \big |\\
& \leq^{(a)} \frac{(T-T_0)(m(T)+C^{RP})}{T}P_{\boldsymbol{x}}(Y_N)+(1-P_{\boldsymbol{x}}(Y_N))\mu \label{eq:second_bound}
\end{align}
where $(a)$ results from:
\begin{align}
\big |\frac{1}{T}  \sum_{t=T_0}^{T-1}  \sum_{k=1}^{K} \sum_{i=1}^{+\infty} \mathbb{E}^{wi}\left[Z_{i}^{k,N}(t)i - z_i^{k,*}i \Big | \overline{Y_N}, \boldsymbol{Z}^N(0)=\boldsymbol{x} \right] \big |&\leq \underset{T_0 \leq t < T}{\text{sup}} \mathbb{E}^{wi}\left[ \sum_{k=1}^{K} \sum_{i=1}^{+\infty} |Z_{i}^{k,N}(t)i - z_i^{k,*}i | \Big| \overline{Y_N}, \boldsymbol{Z}^N(0)=\boldsymbol{x} \right]\\
&= \mathbb{E}^{wi}\left[\underset{T_0 \leq t < T}{\text{sup}} ||\boldsymbol{Z}^N(t)-\boldsymbol{z}^*|| \Big | \overline{Y_N}, \boldsymbol{Z}^N(0)=\boldsymbol{x} \right]<\mu
\end{align}
According to Lemma \ref{lem:result_kurth_theorem}, we have $\lim_{N \rightarrow \infty}P_{\boldsymbol{x}}(Y_N)=0$. Thus, combining the result \eqref{eq:first_bound} and \eqref{eq:second_bound}, we obtain:
\begin{equation}
\lim_{N \rightarrow \infty} \big |\frac{1}{T} \mathbb{E}^{wi}\left[ \sum_{t=0}^{T-1}  \sum_{k=1}^{K} \sum_{i=1}^{+\infty} Z_i^{k,N}(t)i \Big | \boldsymbol{Z}^N(0)=\boldsymbol{x}\right]-\sum_1^K  \sum_{i=1}^{+\infty} z_i^{k,*}i \big |\leq \frac{T_0(m(T_0)+C^{RP})}{T}+\mu
\end{equation}
This inequality is true for all $\mu > 0$, then:
\begin{equation}
\lim_{N \rightarrow \infty} \big |\frac{1}{T} \mathbb{E}^{wi}\left[ \sum_{t=0}^{T-1}  \sum_{k=1}^{K} \sum_{i=1}^{+\infty} Z_i^{k,N}(t)i \Big | \boldsymbol{Z}^N(0)=\boldsymbol{x}\right]-\sum_{k=1}^K  \sum_{i=1}^{+\infty} z_i^{k,*}i \big |\leq \frac{T_0(m(T_0)+C^{RP})}{T}
\end{equation}
Finally we have:
\begin{equation}
\lim_{T \rightarrow \infty} \lim_{N \rightarrow \infty} \big |\frac{1}{T} \mathbb{E}^{wi}\left[ \sum_{t=0}^{T-1}  \sum_{k=1}^{K} \sum_{i=1}^{+\infty} Z_i^{k,N}(t)i \Big | \boldsymbol{Z}^N(0)=\boldsymbol{x}\right]-\sum_{k=1}^K  \sum_{i=1}^{+\infty} z_i^{k,*}i \big |=0
\end{equation}
As consequence:
\begin{equation}
\underset{ T \rightarrow +\infty }{\text{lim}} \lim_{N \rightarrow \infty}  \frac{1}{T} \mathbb{E}^{wi}\left[ \sum_{t=0}^{T-1}  \sum_{k=1}^{K} \sum_{i=1}^{+\infty} Z_i^{k,N}(t)i \Big | \boldsymbol{Z}^N(0)=\boldsymbol{x}\right]=\sum_{k=1}^K  \sum_{i=1}^{+\infty} z_i^{k,*}i 
\end{equation}

\end{appendices}
\end{document}